%% file: main.tex
\documentclass[12pt]{article}

\usepackage[T1]{fontenc}
\usepackage{microtype}          %
\usepackage[full]{textcomp}

\usepackage{amsmath}
\usepackage{amssymb}
\usepackage{amsthm}
\usepackage{fullpage}
\usepackage{mrtorres}

\usepackage{lipsum}
\usepackage{xspace}
\usepackage{relsize}
\usepackage{framed}
\usepackage{xcolor}
\usepackage{graphicx}
\usepackage{multirow}
\usepackage{subcaption}
\usepackage{float}
\usepackage[noend]{algorithmic} 

\usepackage{hyperref} 
\hypersetup{
    colorlinks=true,
    linkcolor=[rgb]{0,0,0.6},
    citecolor=[rgb]{0,0,0.6}
}

\newtheorem{proposition}{Proposition}[section]
\newtheorem{problem}[proposition]{Problem}
\newtheorem{definition}[proposition]{Definition}
\newtheorem{theorem}[proposition]{Theorem}
\newtheorem{lemma}[proposition]{Lemma}
\newtheorem{remark}[proposition]{Remark}

\newcommand{\dsg}{\textsc{DSG}\xspace}
\newcommand{\dss}{\textsc{DSS}\xspace}
\newcommand{\apdsg}[1]{$#1$\textsc{-mean DSG}\xspace}
\newcommand{\pdsg}{\apdsg{p}}

\newcommand{\aec}[1]{\textsc{Exact $#1$-Cover}\xspace}
\newcommand{\ec}{\aec{3}}

\protected\def \plusplus#1{#1\normalfont\nolinebreak[4]\hspace{-.05em}\raisebox{.4ex}{\relsize{-3}{\textbf{++}}}\xspace}

\newcommand{\origgreedy}{\mathsf{Greedy}}
\newcommand{\origgpp}{\plusplus{\mathsf{Greedy}}}

\newcommand{\supergreedy}{\mathsf{SuperGreedy}}
\newcommand{\supergreedypp}{\plusplus{\mathsf{SuperGreedy}}}

\newcommand{\agreedy}[1]{\mathsf{Greedy\text{-}}#1}
\newcommand{\greedy}{\agreedy{p}}

\newcommand{\asgreedy}[1]{\mathsf{Simple\text{-}Greedy\text{-}}#1}
\newcommand{\sgreedy}{\asgreedy{p}}

\newcommand{\agpp}[1]{\plusplus{\mathsf{Greedy}}\text{-}#1}
\newcommand{\gpp}{\agpp{p}}

\newcommand{\asgpp}[1]{\mathsf{Simple\text{-}\plusplus{Greedy}\text{-}}#1}
\newcommand{\sgpp}{\asgpp{p}}

\newcommand{\algreedy}[1]{\mathsf{Lazy\text{-}Greedy\text{-}}#1}
\newcommand{\lgreedy}{\algreedy{p}}

\newcommand{\algpp}[1]{\mathsf{Lazy\text{-}\plusplus{Greedy}\text{-}}#1}
\newcommand{\lgpp}{\algpp{p}}

\newcommand{\abblgreedy}{\mathsf{Lazy\text{-}G\text{-}}p}
\newcommand{\abbsgreedy}{\mathsf{Simple\text{-}G\text{-}}p}
\newcommand{\abbsgpp}{\mathsf{Simple\text{-}\plusplus{G}\text{-}}p}

\newcommand{\cF}{\mathcal{F}}
\newcommand{\cS}{\mathcal{S}}
\newcommand{\cU}{\mathcal{U}}

\newcommand{\tO}{\tilde{O}}

\newcommand{\redopt}{\rho^*}

\newif\ifdraft
\draftfalse

\newcommand{\NameColorComment}[3]{%
  \ifdraft%
  {\renewcommand\thefootnote{\textcolor{#2}{\arabic{footnote}}}%
    \footnote{\color{#2}#1: #3}%
  }%
  \fi%
}

\newcommand{\manuel}{\NameColorComment{Manuel}{magenta}}%
\newcommand{\chandra}{\NameColorComment{Chandra}{blue}}%

\begin{document}
\algsetup{indent=1.1em}  

\title{On the Generalized Mean Densest Subgraph Problem: Complexity and Algorithms}
\author{
Chandra Chekuri\thanks{Dept.\ of Computer Science, Univ.\ of Illinois, Urbana-Champaign, Urbana,
  IL 61801. {\tt chekuri@illinois.edu}. Supported in part by NSF grant CCF-1910149.}
\and
Manuel R.\ Torres\thanks{Dept.\ of Computer Science, Univ.\ of Illinois, Urbana-Champaign, Urbana,
  IL 61801. {\tt manuelt2@illinois.edu}. Supported in part by fellowships from NSF
    and the Sloan Foundation, and NSF grant CCF-1910149.}
}

\hypersetup{pageanchor=false}
\maketitle

\begin{abstract}
  Dense subgraph discovery is an important problem in graph mining
  and network analysis with several applications.  Two canonical
  problems here are to find a \emph{maxcore} (subgraph of maximum min
  degree) and to find a \emph{densest subgraph} (subgraph of maximum
  average degree). Both of these problems can be solved in polynomial time.
  Veldt, Benson, and Kleinberg \cite{vbk-21} introduced the generalized 
  $p$-mean densest subgraph problem which
  captures the maxcore problem when
  $p=-\infty$ and the densest subgraph problem when $p=1$. They
  observed that the objective leads to a supermodular function when
  $p \ge 1$ and hence can be solved in polynomial time; for this case,
  they also developed a simple greedy peeling algorithm with a bounded
  approximation ratio. In this paper, we make several
  contributions. First, we prove that for any $p \in (-\frac{1}{8}, 0) \cup (0, \frac{1}{4})$ 
  the problem is NP-Hard and for any $p \in (-3,0) \cup (0,1)$ the
  weighted version of the problem is NP-Hard, partly resolving a question left open in
  \cite{vbk-21}. Second, we describe two simple $1/2$-approximation
  algorithms for all $p < 1$, and show that our analysis of these
  algorithms is tight.  For $p > 1$ we develop a fast near-linear time
  implementation of the greedy peeling algorithm from \cite{vbk-21}.
  This allows us to plug it into the iterative peeling algorithm that
  was shown to converge to an optimum solution \cite{cqt-22}.
  We demonstrate the efficacy of our algorithms by running extensive experiments
  on large graphs. Together, our results provide a comprehensive understanding of the
  complexity of the $p$-mean
  densest subgraph problem and lead to fast and provably good algorithms
  for the full range of $p$.
\end{abstract}

\thispagestyle{empty}
\newpage
\setcounter{page}{1}

\hypersetup{pageanchor=true}

\input{introduction}

\input{preliminaries}

\input{lazy-greedy}

\input{new_apx}

\input{hardness}

\input{experiments}

\input{conclusion}

\bibliographystyle{alpha}
\bibliography{main}

\newpage
\appendix

\input{appendix}

\end{document}

%% file: introduction.tex
\section{Introduction} Dense subgraph discovery is an essential tool
in graph mining and network analysis. One can view the approach
as finding clusters or communities in a graph where the edges between the
nodes in the cluster are denser compared to those in the entire graph.
There are a number applications of dense subgraph discovery in biological
settings~\cite{hu-05,fratkin-06,lanciano-20}, protein-protein
interaction networks~\cite{bader-03,spirin-03}, web
mining~\cite{gibson-05,dourisboure-07}, social network
analysis~\cite{kumar-06}, real-time story
identification~\cite{angel-14}, and finance and fraud
detection~\cite{du-09,zhang-17,ji-22}.

Different density definitions are used and studied in the literature, motivated
by the needs of applications and theoretical considerations 
(see \cite{f-19,gt-15,lrja-10,tc-21} for some surveys).
Each density definition leads to a corresponding combinatorial optimization problem:
given a graph $G$, find a subgraph of maximum density. Two of the most popular
density measures in the literature are (i) the minimum degree of the
subgraph and (ii) the average degree of the subgraph. These measures lead
to the maxcore problem and densest subgraph problem (\dsg); the goal
is to find the subgraph with the maximum minimum degree and
maximum average degree, respectively. They are both polynomial-time
solvable and have been extensively studied. We briefly describe
them before discussing a common generalization that is the focus
of this paper.

A $k$-core of a graph is a maximal connected subgraph with all
vertices of degree at least $k$. The largest value of $k$ for which
$G$ contains a $k$-core is known as the degeneracy. We refer to the
$k$-core obtaining this maximum as the \emph{maxcore}. $k$-cores are
a popular notion of density, commonly finding use in what is known as
the $k$-core decomposition, a nested sequence of subgraphs that
captures all $k$-cores. One nice feature of a $k$-core
decomposition is that there is a simple linear-time peeling
algorithm $\origgreedy$ to compute it. We refer the reader to \cite{mgpv-20} for a survey
on $k$-core decomposition and applications.

The densest subgraph problem (\dsg), where the goal is to find a
subgraph of maximum average degree, is a classical problem in
combinatorial optimzation that is polynomial time
solvable~\cite{c-00,ggt-89} via network flow techniques among others.
It is widely studied in graph mining. Even though
\dsg can be solved exactly, the algorithms are slow and this has spurred a study
of approximation algorithms for
\dsg~\cite{c-00,bgm-14,dcs-17,bsw-19,boob-20,cqt-22,hqc-22}.
Amongst these approximation algorithms is the 
algorithm $\origgreedy$ introduced by Asahiro et al.~\cite{aitt-00} to solve \dsg
and shown to be a $\frac{1}{2}$-approximation by
Charikar~\cite{c-00}. Note that the peeling order is the same
as the one for computing the optimum $k$-core decomposition; it is only
in the second step where one checks all suffixes that the specific
density measure for \dsg is used. 
Charikar's analysis has spurred the
development and analysis of a variety of peeling algorithms for
several variants of \dsg in both graphs and 
hypergraphs~\cite{ac-09,t-14,t-15,hwc-17,km-18,vbk-21}.

Veldt, Benson and Kleinberg \cite{vbk-21} introduced the
generalized mean densest subgraph problem. It is a common
generalization of the two problems we described in the preceding
paragraphs and is the focus of this paper. Given a real
parameter $p \in \R \cup \{-\infty, \infty\}$ and an undirected graph
$G = (V,E)$, the density of a subgraph $G[S]$ induced by a set $S
\subseteq V$ is defined as:
\[ M_p(S) := \left(\frac{1}{\abs{S}} \sum_{v \in S}
    d_S(v)^p\right)^{1/p}
  \] where $d_S(v)$ is the
  degree of the vertex $v$ in $G[S]$. To make the
dependence on $p$ more explicit, we refer to this problem also as the
$p$-mean densest subgraph problem (\pdsg).  Note that $M_{-\infty}(S)
= \min_{v \in S} d_S(v)$ is the minimum degree in $G[S]$,
$M_{\infty}(S) = \max_{v \in S} d_S(v)$ is the maximum degree,
$M_1(S) = 2|E(S)|/|S|$ is twice the average degree of $G[S]$, and
$M_0(S) = (\prod_{v\in S} d_S(v))^{1/\abs{S}}$, which can be transformed
to $\frac{1}{\abs{S}}\sum_{v\in S} \log d_S(v)$ by taking the logarithm of the 
objective.
Thus, \pdsg for different values of $p$ captures two of the most
well-studied problems in dense subgraph discovery. As $p$ varies
from $-\infty$ to $\infty$, $M_p(S)$ prioritizes the
smallest degree in $S$ to the largest degree in $S$ and provides
a smooth way to generate subgraphs with different density properties.

Veldt et al.\ made several contributions to \pdsg.  They observed that
\apdsg{1} is equivalent to \dsg and that \apdsg{-\infty} is equivalent
to finding the maxcore.  For $p \ge 1$, they observe the set function
$f_p:2^V \rightarrow \mathbb{R}_+$ where $f_p(S) := \sum_{v \in S}
d_S(v)^p$ is a supermodular function\footnote{A real-valued set
function $f : 2^V \to \R$ is \emph{supermodular} if $f(B + v) - f(B)
\ge f(A + v) - f(A)$ for all $A \subseteq B\subseteq V$ and $v \in V
\setminus B$. Equivalently, $f(A) + f(B) \le f(A \cup B) + f(A \cap
B)$ for all $A, B \subseteq V$.  $f$ is supermodular iff $-f$ is
submodular.}. 
This implies that one can solve \pdsg in polynomial time
for all $p > 1$ via a standard reduction to submodular set function
minimization, a classical result in combinatorial optimization
\cite{s-03}. 
Note that the function $f_p(S)$ is not supermodular 
when $p < 1$, which partially stems from the fact that $x^p$ is not 
convex in $x$ when $p < 1$. 
Motivated by the fact that exact algorithms are very slow
in practice, they describe a greedy peeling algorithm $\greedy$
that runs in $O(mn)$ time and is a
$\left(\frac{1}{p+1}\right)^{1/p}$-approximation (here $m$ and $n$ are
the number of edges and nodes of the graph).
Note that the peeling order of $\greedy$ is \emph{not} the same
as that of $\origgreedy$ and depends on $p$. They supplement this
work with experiments, showing that $\greedy$ returns solutions with
desirable characteristics for values of $p$ in the range $[1,2]$.

\paragraph{Motivation for this work.} In this paper we are interested
in algorithms for \pdsg that are theoretically and empirically sound, and
its complexity status for $p < 1$ which was left open in
\cite{vbk-21}. It is intriguing that $p = -\infty$ and $p \ge 1$ are
both polynomial-time solvable while the status of $p \in (-\infty,1)$ is
non-trivial to understand. We focus on the case of $p \in (-\infty, 0) \cup (0, 1)$
as the objective for $p = 0$ is significantly different. Second, we are
interested in developing fast approximation algorithms for \pdsg when
$p > 1$ and when $p < 1$. The $O(mn)$-time algorithm in \cite{vbk-21} is 
slow for large graphs. Moreover, for  $p > 1$, it is of substantial interest to
find algorithms that provably converge to an optimum solution rather
than provide a constant factor approximation (since the problem is
efficiently solvable). Such algorithms have been of much interest
for the classical \dsg problem~\cite{bgm-14,dcs-17,bsw-19,boob-20,cqt-22,hqc-22}.

This paper is also motivated by a few recent works. One of them
is a paper of Chekuri, Quanrud,
and Torres~\cite{cqt-22} that introduced a framework under which one
can understand some of the results of~\cite{vbk-21}.  In~\cite{cqt-22}, they
observe that many notions of density in the literature are of the form
$\frac{f(S)}{\abs{S}}$ where $f:2^V \to \R$ is a supermodular
function. They refer to this general problem as the densest supermodular
subset problem (\dss).  \dss captures \pdsg for $p \ge 1$. One can see
this by noting that for $p \ge 1$, finding $\argmax_S M_p(S)$ is
equivalent to finding $\argmax_S M_p(S)^p$. Thus, for $p \ge 1$, \pdsg
is equivalent to $\max_S \frac{f_p(S)}{\abs{S}}$.
\cite{cqt-22} describes a simple peeling algorithm for \dss
whose approximation guarantee depends on the supermodular function
$f$; they refer to this as $\supergreedy$. In particular, $\greedy$ is
the same as $\supergreedy$ when specialized to $f_p$, and the analysis in \cite{cqt-22} recovers
the one in \cite{vbk-21} for \pdsg.
It is also shown in \cite{cqt-22} that
an \emph{iterated} peeling algorithm, called $\supergreedypp$,
converges to an optimum solution for \dss (thus,
for \emph{any} supermodular function). $\supergreedypp$ is motivated by
and generalizes the iterated greedy algorithm $\origgpp$ that was 
suggested by Boob et al.~\cite{boob-20} for \dsg. When one specializes
$\supergreedypp$ to \pdsg, the algorithm converges to a $(1-\eps)$-approximation 
in $O(\frac{\Delta_p \log n}{\lambda^*\eps^2})$ iterations where
$\Delta_p = \max_{v \in V} f_p(V) - f_p(V-v)$ and $\lambda^*$ is the optimal
density. A naive implementation yields a per-iteration
running time of $O(mn)$.
Another iterative algorithm that converges to an optimum
solution for \dss is described in the work of Harb et al.\ \cite{hqc-22}
and is based on the well-known Frank-Wolfe method applied to \dss; this
is inspired by the algorithm of Danisch et al.~\cite{dcs-17} for \dsg.
When one specializes this algorithm to \pdsg, a naive implementation yields 
a per-iteration running time of $O(m + n \log n)$.
Since \pdsg for $p \ge 1$ is a special case of \dss, these two iterative
algorithms will also converge to an optimum solution for \pdsg.

\subsection{Our Results}
As remarked earlier, our goal is to understand the complexity of \pdsg
and develop improved algorithms for all values of $p$. We make four
contributions in this paper towards this goal and we outline them
below with some discussion of each.

\paragraph{NP-Hardness of \pdsg for $p < 1$.} We prove that \pdsg is
NP-Hard for any fixed $p \in (-\frac{1}{8},0)\cup(0,\frac{1}{4})$ and
\pdsg for weighted graphs is NP-Hard for any fixed $p \in (-3, 0) \cup (0,1)$. The
hardness reduction is technically involved, especially
due to the non-linear objective function. Our proof involves numerical
computation and we
restrict our attention to the range $(-\frac{1}{8}, 0) \cup (0,\frac{1}{4})$ for the
unweighted case and $(-3,0) \cup(0,1)$ for the weighted case to make the calculations
and the proof transparent. We believe that our proof methodology
extends to all $p \in (-\infty, 0) \cup (0, 1)$ and describe an outline to do
so. 

\paragraph{$\frac{1}{2}$-approximation algorithms for $p \in
  (-\infty, 1)$.}  Our NP-Hardness result for \pdsg motivates the search for
approximation algorithms and heuristics when $p < 1$. In \cite{vbk-21}
the authors do some empirical evaluation using $\greedy$ for $p \in (0,1)$
even though the corresponding function $f_p$ is \emph{not}
supermodular; no approximation guarantee is known for this
algorithm. $\greedy$ is only well-defined for $p > 0$.
\chandra{Do we know a bad example for $\greedy$ when $p <1$?}
We describe two different and simple
$\frac{1}{2}$-approximation algorithms for \pdsg when $p \in (-\infty,
1)$, one based on simple greedy peeling and the other based on an exact
solution to \dsg.  These are the first algorithms with 
approximation guarantees for
this regime of $p$. We also show that $\frac{1}{2}$ is tight for these
algorithms. We describe another algorithm, based on iterative peeling,
that interpolates between the preceding two algorithms. It is also
guaranteed to be a $\frac{1}{2}$-approximation but allows us to generate
several candidate solutions along the way in a natural fashion. 

\paragraph{Fast and improved algorithms for $p > 1$.}
Greedy-$p$ from \cite{vbk-21} runs in $O(mn)$ time, which is slow for
large high-degree graphs. We describe a near-linear time algorithm for \pdsg
when $p > 1$.  This algorithm, $\lgreedy$,
is effectively a faster implementation of $\greedy$ that utilizes lazy
updates to improve the running time\footnote{We use $\tO$ notation to
suppress polylogarithmic factors.}  to $\tO(\frac{pm}{\eps})$ while
only losing a factor of $(1 - \eps)$ in the approximation ratio
compared to $\greedy$. Another important motivation and utility for
$\lgreedy$ is the following. Recall the preceding discussion that
iterating the greedy peeling algorithm yields an algorithm that
converges to the optimum solution \cite{cqt-22}.  However, iterating
the vanilla implementation of $\greedy$ is too slow for large graphs.
In contrast, the faster $\lgreedy$ allows us to run
many iterations on large graphs and gives good results in a reasonable amount
of time. 

\paragraph{Experiments.} We empirically evaluate the efficacy of our algorithms 
and some variations on a collection of ten publicly available large real-world graphs. 
For $p > 1$, we show that $\lgreedy$ returns solutions with densities close to that
of $\greedy$ while running typically at least $2$ to $4$ times as
fast. We evaluate the performance of the two iterative algorithms
that provably converge to an optimum solution and compare their
performance. We find that our algorithm converges faster on all graphs
and values of $p$ tested.  
For $p < 1$, we empirically evaluate our two new approximation algorithms,
and we compare against $\greedy$ which was tested in \cite{vbk-21} as a heuristic.
We compare the three algorithms both for their running time and solution quality
and report our findings for different values of $p \in (0,1)$. One of our
approximation algorithms significantly outperforms all algorithms tested in terms
of running time while returning a subgraph with comparable solution quality.

\paragraph{Organization.}
In an effort to highlight the algorithmic aspects of our work, we start by describing 
and analyzing $\lgreedy$ in Section~\ref{sec:fast} and presenting our two new 
approximation algorithms for $p < 1$ in Section~\ref{sec:apx}. We also discuss our
heuristics for $p > 1$ and $p < 1$ based on iterative algorithms for
\pdsg in Section~\ref{sec:apx}. We then present our 
hardness result in Section~\ref{sec:hardness}. We conclude with experiments in
Section~\ref{sec:experiments}.
All proofs that are omitted from the main body of the
paper are contained in the appendix.

\section{Related work}\label{sec:rel-work}
  \dsg has been the impetus for a wide-ranging and extensive
  subfield of research. The problem is solvable in polynomial time via different methods,
  such as a reduction to maximum flow~\cite{g-84,l-76,pq-82,ggt-89}, a reduction to 
  submodular minimization, 
  and via an exact LP relaxation~\cite{c-00}. Running times of these algorithms 
  are relatively slow and thus approximation algorithms have been considered. 
  The theoretically fastest known approximations run in near-linear time, obtaining
  $(1 - \eps)$-approximations in $\tO(m \cdot \poly(\frac{1}{\eps}))$ 
  time~\cite{bgm-14,bsw-19,cqt-22}. The fastest known 
  running time is $\tO(\frac{m}{\eps})$~\cite{cqt-22}. 
  These algorithms are relatively complex to implement and practitioners often prefer
  simpler approximation algorithms but with worse approximation ratios, such
  as $\origgreedy$~\cite{aitt-00,c-00} and $\origgpp$~\cite{boob-20,cqt-22}. 
  However, there has been recent work that showed that one can obtain
  near-optimal solutions using continuous optimization methods that are both simple and
  efficient on large real-world graphs~\cite{dcs-17,hqc-22}.
  
  We discuss some of the density definitions considered in the literature.
  Given a graph $G = (V,E)$ and a finite collection of pattern graphs $\cF$, many densities 
  take the form $\frac{f(S)}{\abs{S}}$ where $f: 2^V \to \R$ counts the number of occurrences of the
  patterns in $\cF$ in the induced subgraph $G[S]$. This exact problem was considered
  in~\cite{f-08}. \cite{t-14} considered the special case where $\cF$ is a single triangle graph
  and \cite{t-15} considered the special case where $\cF$ is a single clique on $k$ vertices.
  One can also consider a version of \dsg for hypergraphs, where the density is defined as
  $\frac{\abs{E(S)}}{\abs{S}}$ where $E(S)$ is the set of hyperedges with 
  all vertices in $S$~\cite{hwc-17}. We can reduce the 
  pattern graph problem to the hypergraph
  problem by introducing a hyperedge for each occurrence of the pattern in the input graph. 
  Other density definitions look at modifying \dsg by considering the density 
  $\frac{\abs{E(S)}}{g(\abs{S})}$ where 
  $g$ is an arbitrary function of $\abs{S}$~\cite{km-18}.
  Given a parameter $\alpha$, another density considers
  $\frac{\abs{E(S)} - \alpha \abs{\delta(S)}}{\abs{S}}$~\cite{mk-18}, which has connections to
  modularity density maximization~\cite{li-08,smt-22}. The \pdsg objective $M_p$ is a class of 
  densities that captures a wide-range of different objectives, including \dsg when 
  $p = 1$ and the maxcore problem when $p = -\infty$~\cite{vbk-21}.
  Aside from the variation modifying the denominator of \dsg from~\cite{km-18}, 
  all of the different densities mentioned above fall into the framework of \dss where we
  want to maximize $\frac{f(S)}{\abs{S}}$ for a supermodular function $f$~\cite{cqt-22}.

  This field of research is vast and we cannot hope to do it justice
  here, so we point the reader to five separate tutorials/surveys on the 
  topic and the references therein~\cite{f-19,gt-15,lrja-10,tc-21,lmfb-23}.

%% file: preliminaries.tex
\section{Preliminaries}\label{sec:prelim}
  Fix a graph $G = (V,E)$. Let $d_S(v)$ denote the degree of $v$ in the 
  induced subgraph $G[S]$. Let $N(v)$ denote the neighborhood of 
  $v$ in $G$.
  For any $p \in \R \cup \{-\infty, \infty\}$,
  and any $S \subseteq V$, we define $f_p(S) := \sum_{v \in S} d_S(v)^p$,
  $\rho_p(S) := \frac{f_p(S)}{\abs{S}}$ and 
  $M_p(S) := \rho_p(S)^{1/p}$.
  We let $S_p^*(G)$ denote $\argmax_{S \subseteq V} M_p(S)$ and let
  $M_p^*(G) := M_p(S_p^*)$. Note we use $S_p^*$ and $M_p^*$
  in place of $S_p^*(G)$ and $M_p^*(G)$ when the graph $G$ is clear
  from context. For $S \subseteq V$ where $v \in V \setminus S$ and $u \in S$, 
  we denote $S \cup \{v\}$ as $S+ v$ and $S \setminus \{u\}$ as $S -u$.
  For $v \in V$ and $S \subseteq V$, let
  $f_p(v \mid S) := f_p(v + S) - f_p(S)$. As we often consider 
  $p \in \{-\infty, \infty\}$, we use $[-\infty, \infty]$
  to denote the extended real line $\R \cup \{-\infty, \infty\}$.
  
  \paragraph{Peeling algorithms.} 
  Recall that the greedy peeling algorithm $\greedy$ of Veldt et al.\ for 
  \pdsg is a $\left(\frac{1}{p+1}\right)^{1/p}$-approximation and runs in 
  $O(mn)$ time~\cite{vbk-21}. 
  We give the pseudocode in Figure~\ref{fig:prelim-greedy}.
  
  \begin{figure}
  \begin{framed}
      $\greedy(G = (V, E))$\
    \begin{algorithmic}[1]
      \STATE $S_1 \gets V$
      \FOR{$i=1$ \TO $n-1$}
        \STATE $v_i \gets \argmin_{v \in S_i} f_p(v \mid S_i - v)$
          \label{line:greedy-3}
        \STATE $S_{i+1} \gets S_i - v_i$
      \ENDFOR
      \RETURN{$\argmax_{S_i} M_p(S_i)$}\label{alg:greedy-p-ret}
    \end{algorithmic}
  \end{framed}
  \caption{Greedy peeling algorithm for \pdsg of Veldt et al.~\cite{vbk-21}.}
  \label{fig:prelim-greedy}
  \end{figure}    
  
  We now introduce the iterative peeling algorithm $\supergreedypp$ of Chekuri et al.\
  for \dss applied specifically to the function $f_p$~\cite{cqt-22}, which is
  a generalization of the $\origgpp$ algorithm of Boob et al.~\cite{boob-20}. 
  We refer to the specialization of $\supergreedypp$ to \pdsg as $\gpp$.
  $\gpp$ maintains weights for each of the vertices and instead of 
  peeling the vertex $v$ minimizing $f_p(v \mid S_i)$ like $\greedy$, it peels the 
  vertex $v$ minimizing $f_p(v\mid S_i)$ plus the weight of $v$. Since 
  $\gpp$ initializes the weights to $0$, the first iteration is exactly
  $\greedy$. However, subsequent iterations have positive values for the weights 
  and therefore different orderings of the vertices are considered. \cite{cqt-22} 
  shows a connection between this algorithm and the multiplicative 
  weights updates method and proves convergence via this connection.
  We give the pseudocode for $\gpp$ in Figure~\ref{fig:prelim-greedypp}. 
  
  \begin{figure}[h]
  \begin{framed}
      $\gpp(G = (V, E), T)$\
    \begin{algorithmic}[1]
      \STATE for all $v \in V$, set $\ell_v^{(0)} = 0$
      \FOR{$i=1$ \TO $T$}
        \STATE $S_{i,1} \gets V$
        \FOR{$j=1$ \TO $n-1$}
          \STATE $v_{i,j} \gets \argmin_{v \in S_i} \ell_v^{(i-1)} + f_p(v \mid S_{i,j} - v)$
          \STATE $\ell_{v_{i,j}}^{(i)} \gets \ell_{v_{i,j}}^{(i-1)} + f_p(v_{i,j} \mid S_{i,j}-v_{i,j})$
          \STATE $S_{i, j+1} \gets S_{i,j} - v_{i,j}$
        \ENDFOR
      \ENDFOR
      \RETURN{$\argmax_{S_{i,j}} M_p(S_{i,j})$}
    \end{algorithmic}
  \end{framed}
  \caption{Specialization of $\supergreedypp$ from~\cite{cqt-22} to \pdsg.}
  \label{fig:prelim-greedypp}
  \end{figure}    
  
  \paragraph{Properties of $M_p$.} We give a known fact about the monotonicity 
  of the $M_p$ objective in the graph setting we consider. We provide a short proof
  in the appendix (see Appendix~\ref{sec:prelim-proofs-app}).
  \begin{proposition}\label{prop:means}
    Let $S \subseteq V$. For $p \le q$, we have $M_p(S) \le M_q(S)$.
  \end{proposition}  
  
  \paragraph{Degeneracy.}
  The degeneracy of a graph is a notion often used to measure
  sparseness of the graph. 
  There are a few different commonly-used definitions of degeneracy. We use the
  following definition for convenience.
  \begin{definition}[degeneracy and maxcore]
    The \emph{degeneracy} $d(G)$ of a graph $G$ is the maximum min-degree
    over all subgraphs i.e.\ $\max_S \min_{v \in S} d_S(v)$.
    The subset of vertices attaining the maximum min-degree is the 
    \emph{maxcore} i.e.\ $\argmax_S \min_{v \in S} d_S(v)$.
  \end{definition}   
  
  The degeneracy, which is exactly $M_{-\infty}^*$, is easy to compute via the 
  standard greedy peeling algorithm.
  We state this well-known fact in the following proposition (e.g., see~\cite{mb-83}).
  Note that the algorithm constructs the same ordering of the vertices as the
  algorithm $\agreedy{1}$. 
  
  \begin{proposition}\label{prop:greedy-peeling}
    Let $v_1, \ldots, v_n$ be the order of the vertices produced by the standard
    greedy peeling algorithm for computing the degeneracy. 
    For $i \in [n]$, let $S_i = \{v_i, v_{i+1}, \ldots, v_n\}$.
    The subset $S_i$ maximizing the minimum degree is the maxcore and
    therefore the minimum degree of $S_i$, $d_{S_i}(v_i)$, is the degeneracy of $G$.
  \end{proposition}  
  
  We have the following statement connecting different values of 
  $M_p^*$. The first two inequalities follow directly from 
  Proposition~\ref{prop:means} and the last follows via a simple known argument 
  connecting the degeneracy to a subgraph with maximum average degree
  $S_1^*$ (e.g., see~\cite{ft-14}).
  \begin{proposition}\label{prop:core-to-dsg}
    For any graph and any $p \in [-\infty, 1]$, we have $M_{-\infty}^*
    \le M_p^* \le M_1^* \le 2M_{-\infty}^*$.
  \end{proposition}

%% file: lazy-greedy.tex
\section{Fast implementation of $\greedy$}\label{sec:fast}
  In this section, we present a fast algorithm $\lgreedy$
  that runs in near-linear time for constant $p > 1$ that is effectively a faster implementation
  of $\greedy$. 
  
    The naive implementation of $\greedy$ dynamically maintains a data
    structure of the values $f_p(v\mid V - v)$ for all $v$ subject to vertex deletions in
    overall time $O(mn)$. The idea behind our near-linear time implementation of 
    $\greedy$ is simple: we can dynamically maintain a $(1+\eps)$-approximation 
    to the values $f_p(v \mid V-v)$ in $\tilde{O}(\frac{pm}{\eps})$ time while
    only losing a factor of $1 - \eps$ in the approximation ratio.
    To dynamically maintain such a data structure, we 
    use \emph{approximate values} of vertex degrees to compute a proxy for $f_p(v\mid V-v)$
    and will only update $f_p(v \mid V-v)$ when an approximate vertex degree changes. 
    
    We show that if we keep $(1+\frac{\eps}{p})$-approximations of degrees, then this leads
    to a $(1+\eps)$-approximation of $f_p(v\mid V-v)$ values. In particular, note 
    that for $p \in \R \setminus \{0\}$, for any vertex 
    $v \in V$ and $S \subseteq V$ with $v \in S$, we have
    \begin{align}
      f_p(v \mid S - v)
      &=
      \sum_{u \in S} d_S(u)^p - \sum_{u \in S - v} d_{S-v}(u)^p\label{eq:marginal-0}\\
      &=
      d_S(v)^p + \sum_{u \in N(v) \cap S} d_S(u)^p - (d_S(u) - 1)^p,\label{eq:marginal}
    \end{align}
    where~(\ref{eq:marginal}) follows from~(\ref{eq:marginal-0}) as vertices not 
    incident to $v$ have the same degree in $G[S]$ and $G[S-v]$.
    Our algorithm will always exactly update the first term in~(\ref{eq:marginal}) but will only
    update the sum when approximate degrees change.

    We give pseudocode for this algorithm in Figure~\ref{fig:lazy-greedy}. 
    It is important to note that for $p > 0$,
    maximizing $M_p(S)$ is equivalent to maximizing $\frac{f_p(S)}{\abs{S}}$. This implies
    that Line~\ref{alg:greedy-p-ret} of $\greedy$ and Line~\ref{alg:lazy-greedy-p-ret} of
    $\lgreedy$ are equivalent for $p > 0$.
      
    \begin{figure}[t]
      \begin{framed}
          $\lgreedy(G = (V, E), \eps\ge0)$\
        \begin{algorithmic}[1]
          \STATE $S_1 \gets V$
          \STATE $D[v] \gets d_G(v)$ for all $v \in V$ 
          \STATE $D'[v] \gets D[v]$ for all $v \in V$ 
          \STATE $A[v] \gets D[v]^p + \sum_{u \in N(v)} D[u]^p - (D[u]-1)^p$,  $\forall v \in V$ 
            \COMMENT{$A$ stored as a min-heap}
          \FOR{$i=1$ \TO $n-1$}
            \STATE $v_i \gets \argmin_{v \in S_i} A[v]$
            \STATE $S_{i+1} \gets S_i - v_i$; remove $v_i$ from $A$
            \FOR{$u \in N(v_i) \cap S_{i+1}$}
              \STATE $A[u] \gets A[u] - (D[u]^p - (D[u]-1)^p)$ 
              \STATE $A[u] \gets A[u] - (D'[v_i]^p - (D'[v_i]-1)^p)$
              \STATE $D[u] \gets D[u] - 1$  
              \IF{$D'[u] > (1 + \frac{\eps}{p}) D[u]$}\label{line:lazy-greedy-13}
                \FOR{$w \in N(u) \cap S_{i+1}$}\label{line:lazy-greedy-14}
                  \STATE $A[w] \gets A[w] - (D'[u]^p - (D'[u]-1)^p)$ 
                  \STATE $A[w] \gets A[w] + (D[u]^p - (D[u]-1)^p)$ 
                \ENDFOR
                \STATE $D'[u] \gets D[u]$                
              \ENDIF
            \ENDFOR
          \ENDFOR
          \RETURN{$\argmax_{S_i} \frac{f_p(S_i)}{\abs{S_i}}$}\label{alg:lazy-greedy-p-ret}
        \end{algorithmic}
      \end{framed}
      \caption{$(\frac{1-\eps}{p+1})^{1/p}$-approximation for \pdsg for $p \ge 1$ that runs in 
        $\tilde{O}(\frac{p m}{\eps})$ time.}
      \label{fig:lazy-greedy}    
    \end{figure}

    \begin{theorem}\label{thm:fast}
      Let $p \ge 1$, $G = (V,E)$ be an undirected graph, and let $\eps \in (0,\frac{1}{2}]$. Then
      $\lgreedy(G, \eps)$ is a $(\frac{1-\eps}{p+1})^{1/p}$-approximation to 
      \pdsg with an $O\left(\frac{p m \log ^2 n}{\eps}\right)$ running time.
    \end{theorem}
    
    Note that when $\eps = 0$, the approximation guarantee of Theorem~\ref{thm:fast} 
    matches that of $\greedy$~\cite{vbk-21}. $\lgreedy$ is exactly $\greedy$ when $\eps = 0$
    and thus the running time is $O(mn)$ in this case.
      
    To prove Theorem~\ref{thm:fast}, we first observe that at each iteration of 
    $\greedy$, if $S$ is the current vertex set and we return a vertex $v$ satisfying
    \begin{equation}\label{eq:apx}
      f_p(v \mid S - v)
      \le
      (1+\eps)\min_{u \in S} f_p(u\mid S - u),
    \end{equation}
    we only lose a $(1 - \eps)$-multiplicative factor in the approximation ratio.
    We show this in the following lemma. Note that the proof only requires a slight
    modification to the proof of Theorem 3.1 in~\cite{cqt-22} and is therefore
    included in the appendix for the sake of completeness 
    (see Appendix~\ref{sec:fast-proofs-app}).
      
    \begin{lemma}\label{lem:approx-peel}
      Assume $p \ge 1$ and $\eps \in [0,1]$.
      Suppose we greedily peel vertices with the update rule in~(\ref{eq:apx})
      (i.e.\ modify update rule, Line~(\ref{line:greedy-3}), of $\greedy$ with the update 
      rule in~(\ref{eq:apx})). Then the output is a $(\frac{1-\eps}{p + 1})^{1/p}$-approximation
      for \pdsg.
    \end{lemma}
      
    The algorithm $\lgreedy$ dynamically maintains an approximate value
    of $f_p(v \mid V -v)$ subject to vertex deletions for each vertex $v$.     
    $\lgreedy$ specifically uses approximate vertex degrees to estimate
    the sum in~(\ref{eq:marginal}). We use the following 
    lemma to show that approximate degrees suffice in maintaining a close approximation 
    to $\sum_{u \in N(v) \cap S} d_S(v)^p - (d_S(v) - 1)^p$.
    \begin{lemma}\label{lem:fast-ind-approx}
      Let $\eps \in [0,\frac{1}{2}]$, 
      $d$ be an integer and $p \ge 1$. Let $\alpha \in [0, 1 + \frac{\eps}{p}]$. Then 
      $
        (\alpha d)^p - (\alpha d - 1)^p 
        \le 
        (1 +\eps)(d^p - (d-1)^p).
      $
    \end{lemma}
    The proof of the preceding lemma follows easily by essentially arguing that
    $\frac{(\alpha d)^p - (\alpha d - 1)^p }{1+\eps}$ is decreasing in $\eps$.

    The following lemma handles the issue of running time. Note that the statement
    also holds for $p \in (0,1)$.
    \begin{lemma}\label{lem:fast-time}
      Fix $\eps \in (0,1]$. Suppose $p > 0$. $\lgreedy(G,\eps)$ runs in 
      $O\left(m\log n + \frac{m\log^2 n}{\log(1 + \frac{\eps}{p})}\right)$ time.
      If $p \ge \eps$, this simplifies to
      $O\left(\frac{pm\log^2 n}{\eps}\right)$ time.  
      \manuel{Can cut a line here if necessary.}
    \end{lemma}
    The running time of $\lgreedy$ is dominated by the inner for loop on 
    Line~(\ref{line:lazy-greedy-14}). The proof of Lemma~\ref{lem:fast-time} proceeds 
    by recognizing that this for loop is only ever called when the approximate degree
    of a vertex $u$ exceeds $(1 + \frac{\eps}{p})D[u]$, where $D[u]$ is the exact current
    degree. Thus, the maximum number of times we update any vertex is 
    $O(\log_{1 + \frac{\eps}{p}}(n)) = O(\frac{p \log n}{\eps})$. Since the for loop
    only requires $O(d_G(u) \log n)$ time to run, the total time spent on a single vertex
    $u$ in this for loop is $O(d_G(u) \cdot \frac{p \log^2n}{\eps})$. Summing over all
    vertices, this leads to the desired running time.
    
    The analysis above for $\eps > 0$ implies a better analysis for the case of $\eps = 0$.
    If $\eps = 0$, the condition on Line~(\ref{line:lazy-greedy-13}) is always
    satisfied. As noted above, the inner for loop always requires $O(d_G(u)\log n)$ time
    to run. Further, as $u$ can only be the neighbor of a peeled vertex at most $d_G(u)$
    times, the total amount of time spent on Line~(\ref{line:lazy-greedy-14}) for vertex $u$ is 
    $O(d_G(u)^2 \cdot \log n)$. This leads to an overall running time of 
    $O(\sum_{v \in V} d_G(v)^2 \log n)$, which is $\tilde{O}(mn)$. 
    
    We leave the proof of Theorem~\ref{thm:fast} to the appendix (see 
    Appendix~\ref{sec:fast-proofs-app}).

%% file: new_apx.tex
\section{Approximation algorithms}\label{sec:apx}
  We give two new approximation algorithms for $p$-mean DSG 
  when $p \in (-\infty, 1)$. The
  algorithms rely on the fact that $S_1^*$ and $S_{-\infty}^*$ can be found in 
  polynomial time. We show each of these subgraphs are a 
  $\frac{1}{2}$-approximation 
  for this regime of $p$. We complement this result with a family of graphs where
  $\frac{1}{2}$ is the best one can do for each algorithm.
  We also briefly discuss our iterative heuristics for both $p > 1$ and $p < 1$ 
  in Section~\ref{sec:iter-heuristics}.

  \subsection{$1/2$-approximation via the maxcore}
  \label{sec:apx-maxcore}
    Our algorithm $\sgreedy$ that leverages the standard greedy peeling algorithm for the 
    maxcore is given in Figure~\ref{fig:simple-greedy-p}. 
    The algorithm is exactly Charikar's greedy peeling algorithm when $p = 1$.

    \begin{figure}[t]
      \begin{framed}
          $\sgreedy(G = (V, E))$\
        \begin{algorithmic}[1]
          \STATE $S_1 \gets V$
          \FOR{$i=1$ \TO $n-1$}
            \STATE $v_i \gets \argmin_{v \in S_i} d_{S_i}(v)$
            \STATE $S_{i+1} \gets S_i - v_i$
          \ENDFOR
          \RETURN{$\argmax_{S_i} M_p(S_i)$}
        \end{algorithmic}
      \end{framed}
      \caption{$\frac{1}{2}$-approximation via greedy peeling for \pdsg where $p < 1$.}
      \label{fig:simple-greedy-p}
    \end{figure}    
    
    \begin{theorem}
      Let $p \in [-\infty, 1]$. Let $S_{out} = \sgreedy(G, p)$. Then 
      $M_p(S_{out}) \ge \frac{1}{2} M_p^*$.
    \end{theorem}
    \begin{proof}
      By Proposition~\ref{prop:greedy-peeling}, there exists $i \in [n]$ with
      $M_{-\infty}(S_i) = M_{-\infty}^*$. By Proposition~\ref{prop:means}, 
      $M_{-\infty}(S_i) \le M_p(S_i)$ and by choice of $S_{out}$, we have
      $M_p(S_i) \le M_p(S_{out})$. Therefore,  $M_{-\infty}^* \le M_p(S_{out})$.
      Finally, by Proposition~\ref{prop:core-to-dsg},  we have 
      $\frac{1}{2} M_p^* \le M_{-\infty}^* $. Combining these two statements,
      $\frac{1}{2} M_p^* \le M_p(S_{out})$. This concludes the proof.
    \end{proof}
    
    \begin{remark}
      In~\cite{vbk-21}, when $p > 1$, they show that $\sgreedy$ can perform
      arbitrarily poorly
      by constructing the following graph: the disjoint union of the 
      complete bipartite graph $K_{d,D}$ with $d$ vertices on one side and $D$ 
      vertices on the other and $r$ cliques of size $d+2$ where $d \ll D$.
      $\sgreedy$ first peels all of the vertices in $K_{d,D}$ then peels all of the cliques. 
      It is not hard to show that $K_{d,D}$ is the optimal solution with density proportional
      to $(dD^{p-1})^{1/p}$ and the highest density suffix $\sgreedy$ finds is the entire 
      graph with density proportional to $D^{(p-1)/p}$. Then $d^{-1/p}$ is the best 
      approximation $\sgreedy$ can achieve. This is a bad example is because
      the high degree vertices in $K_{d,D}$ are not prioritized and these vertices 
      contribute significantly to the density of solutions in this graph for $p > 1$.
    \end{remark}

  \subsection{$1/2$-approximation via the $1$-mean densest subgraph}
  \label{sec:apx-density}
    We analyze the algorithm that simply returns the $1$-mean densest subgraph.
    
    \begin{theorem}
      Let $p \in [-\infty, 1]$. Recall $S_1^* = \argmax_{S \subseteq V} M_1(S)$. We have
      $M_p(S_1^*) \ge \frac{1}{2} M_p^*$.
    \end{theorem}
    \begin{proof}
      We first argue that
      \begin{equation}\label{eq:ks}
        M_{-\infty}(S_1^*) 
        \ge 
        \frac{1}{2} M_1^*.
      \end{equation}
      It suffices to show that $d_{S_1^*}(v) \ge \frac{\abs{E(S_1^*)}}{\abs{S_1^*}}$ for every 
      $v \in S_1^*$. Suppose towards a contradiction that there exists $v \in S_1^*$ such 
      that $d_{S_1^*}(v) < \frac{\abs{E(S_1^*)}}{\abs{S_1^*}}$. Using this and observing
      $\abs{E(S_1^*)} - \abs{E(S_1^*- v)} = d_{S_1^*}(v)$, after rearranging, we have
      $\frac{\abs{E(S_1^*-v)}}{\abs{S_1^* -v}} > \frac{\abs{E(S_1^*)}}{\abs{S_1^*}}$.
      Multiplying through by $2$, we obtain $M_1(S_1^* -v) > M_1(S_1^*)$, contradicting
      the optimality of $S_1^*$.
      
      We then have
      \[
        M_p(S_1^*)
        \ge
        M_{-\infty}(S_1^*)
        \ge
        \frac{1}{2}M_1^*
        \ge
        \frac{1}{2} M_p^*
      \]
      where the first and last inequality are via Proposition~\ref{prop:means} and the 
      second inequality is via~(\ref{eq:ks}). This concludes the proof.
    \end{proof}
    
  \subsection{Tight examples: showing $M_p^* \approx 2M_{-\infty}^*$}
    The goal of this section is to show that $\frac{1}{2}$ is tight for the two approximation
    algorithms from the previous sections. 
    
    \begin{theorem}\label{thm:half-lb}
      There exists a family of graphs showing the algorithms from 
      Section~\ref{sec:apx-maxcore}  and
      Section~\ref{sec:apx-density} are at best a $\frac{1}{2}$-approximation.
    \end{theorem}
    
    We want to point out that it is a nontrivial task to devise such a family of graphs, and the
    construction we give is not immediately obvious. In the following lemma, we show that
    it suffices in our construction in Theorem~\ref{thm:half-lb} to find a graph $H$ where the 
    degeneracy $M_{-\infty}^*(H)$ is roughly half of $M_p^*(H)$. 
    From Proposition~\ref{prop:core-to-dsg}, we know that $M_p^*(H) \le M_1^*(H)
    \le 2M_{-\infty}^*(H)$, so we are essentially trying to find a family of graphs where these
    inequalities are tight. There is a simple family of graphs, namely 
    complete bipartite graphs where one partition is much larger than the other, where
    $M_1^* \approx 2M_{-\infty}^*$. It is not immediately clear, however, if such a family exists
    where $M_p^* \approx 2M_{-\infty}^*$. One might 
    quickly realize that the path satisfies this condition, which would give a relatively 
    simple construction to prove Theorem~\ref{thm:half-lb}. However, this construction would
    be brittle in the sense that it leaves open the possibility that the problem gets easier
    as the degeneracy of the graph increases. To resolve this concern, in 
    Theorem~\ref{thm:comparing-p-to-1}, we show how one can construct a graph with an 
    arbitrary degeneracy such that $M_p^* \approx 2M_{-\infty}^*$. This result is
    potentially of independent interest as it has some nice graph theoretic connections.

    \begin{figure}
      \centering
      \includegraphics[scale=0.65]{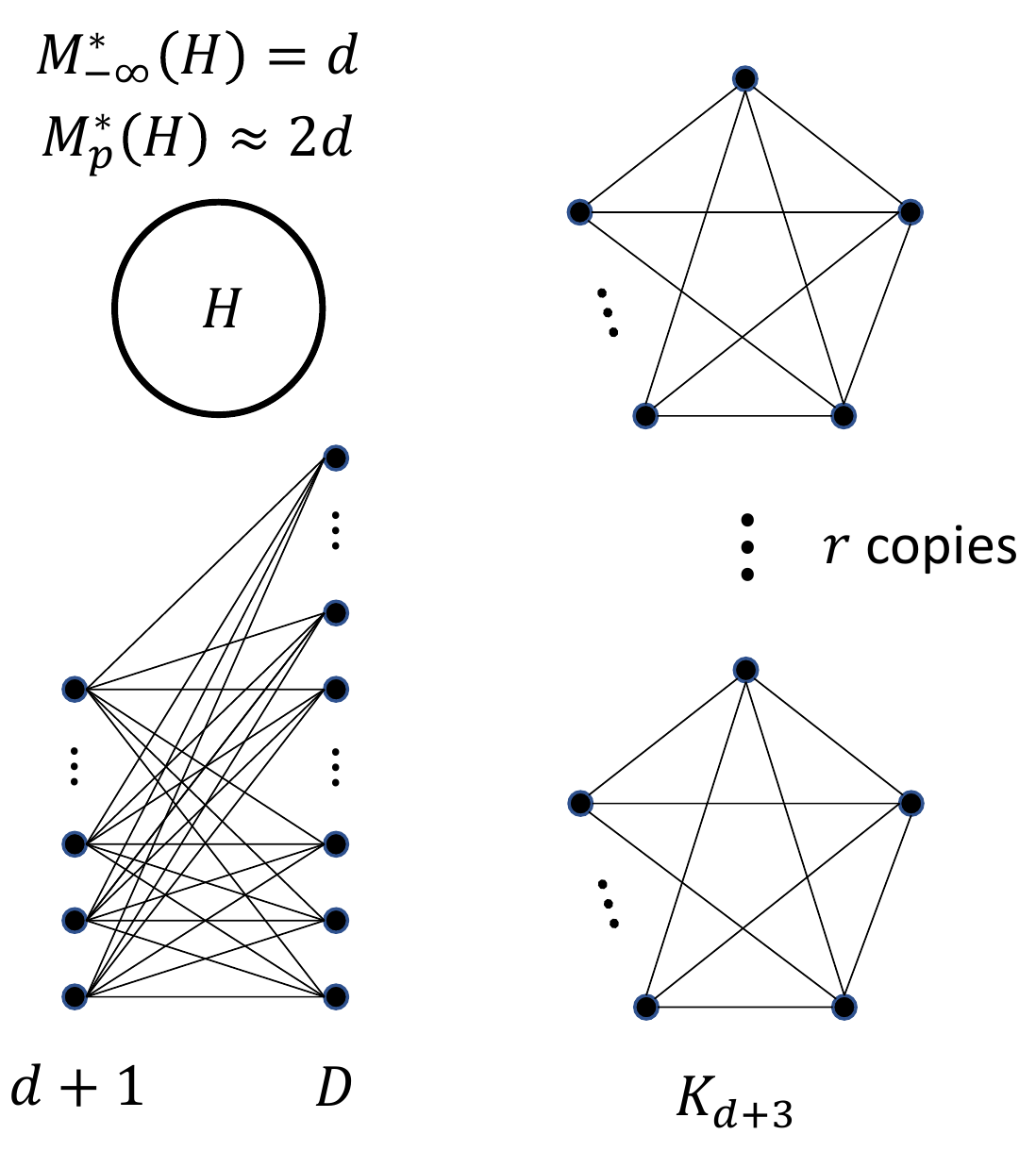}
      \caption{Tight instance for both approximation algorithms from 
      Section~\ref{sec:apx}.}
      \label{fig:tight-ex}
    \end{figure}
    
    \begin{lemma}\label{lem:bad-construct}
      Let $p \in (-\infty,1)$. Fix $\alpha \in [1,2]$ and an integer
      $d \ge 1$. Assume there exists a graph $H$ where
      $d = M_{-\infty}^*(H)$ and $M_p^*(H) \ge \alpha d$.
      Let $G$ be the disjoint union of
      (i) $H$, (ii) $r$ copies of a clique on $d+3$ vertices $K_{d+3}$, and (iii)
      the complete bipartite graph $K_{d+1,D}$ (see Figure~\ref{fig:tight-ex}).  
      Let $n_H$ be the number of vertices in $H$. 
      Assume $n_H^2 = o(r)$, $d^2 = O(D)$ and $dD = o(r)$. 
      
      Then (1) $M_p^*(G) \ge \alpha d$, (2) letting $S_1 = \sgreedy(G)$, we have
      $\displaystyle\lim_{r \to \infty} M_p(S_1) = d+2$ and (3)
      $\displaystyle\lim_{D \to \infty} M_p(M_1^*(G)) = d+1$.

    \end{lemma}
    
    In the graph $G$ in Lemma~\ref{lem:bad-construct},
    the optimal solution is $H$. When we run $\sgreedy$ on $G$, $H$ is peeled first
    and the remainder of the graph has density roughly $d+2$ when there
    are sufficiently many cliques.
    For the approximation algorithm from Section~\ref{sec:apx-density}, the densest $1$-mean
    subgraph of $G$ is $K_{d+1, D}$, which satisfies $M_p^*(K_{d+1,D}) \approx d+1$
    for sufficiently large $D$.
    With more effort, one can argue that, for a fixed $p$ and $\eps$, it suffices to take 
    a graph $G$ whose size is polynomial in $D$ and $r$.

    With Lemma~\ref{lem:bad-construct}, all that remains to prove
    Theorem~\ref{thm:half-lb} is to construct a graph $H$ where
    $d = M_{-\infty}^*(H)$ and $M_p^*(H) \ge (2 - \eps)d$ for a given $d$ and $\eps$.    
    
    \begin{theorem}\label{thm:comparing-p-to-1}
      Let $\eps \in (0,1)$ and $d \ge 1$ be an integer.
      Assume $p \in (-1,1)$. For all integers $n \ge \frac{2d}{\eps}$, there 
      exists a graph $G$ on $n$ vertices with degeneracy $d$ such that 
      $M_p^*(G) \ge (1-\eps)2M_{-\infty}^*(G)$.
                    
      For $p \in (-\infty, -1]$, we can obtain the same guarantee if
      $n \ge \frac{\binom{d+1}{2}((1 - \frac{1}{2d})^p - 1) + d(2^{-p} - 1)}
      {\eps\cdot \abs{p}}$.
    \end{theorem}
    
    The proof of the proceeding theorem is constructive, using a result of Bickle~\cite{b-12} 
    to aid in constructing a $d$-degenerate graph where most of the vertices
    have degree $2d$. Because this graph $G$ has most vertices equal to $2d$,
    then we expect that $M_p^*(G) \approx 2d$ as the $M_p^*$ value of any 
    regular graph is simply the degree of the graph. 

    \begin{proof}[Proof of Theorem~\ref{thm:half-lb}] 
      Theorem~\ref{thm:comparing-p-to-1} implies that for sufficiently
      large $n$, there exists a graph $H$ on $n$ vertices such that $d = M_{-\infty}^*(H)$
      and $M_p^*(H) \ge (1 - \eps)2d$. Therefore, Lemma~\ref{lem:bad-construct} 
      holds with $\alpha = (1 - \eps)2d$. Thus, there exists an infinite class of graphs where
      both algorithms from Section~\ref{sec:apx-maxcore} and~\ref{sec:apx-density}
      are at best a $\frac{1}{2}$-approximation.
    \end{proof}
    
  \subsection{Iterative heuristics for $p > 1$ and $p < 1$}\label{sec:iter-heuristics}
    We introduce iterative heuristics for $p > 1$ and $p < 1$ with the goal
    of producing better solutions than the algorithms that only consider a single ordering of the
    vertices. 
    
    For $p > 1$, recall that the algorithm $\gpp$ converges
    to a near-optimal solution for \pdsg. $\gpp$ runs $\greedy$
    at each iteration, however, this is computationally prohibitive on large graphs.
    We therefore introduce $\lgpp$, which runs $\lgreedy$ at each iteration.
    We show the benefit of iteration on real-world graphs in Section~\ref{sec:exp-near-opt}.

    For $p < 1$, we consider a heuristic $\sgpp$ that essentially takes the best of both of
    our approximation algorithms from this section. For the approximation algorithm
    from Section~\ref{sec:apx-density} that returns the $1$-mean densest subgraph, 
    we use the algorithm $\origgpp$ to compute a near-optimal solution $S_1$ to
    \apdsg{1}. Our heuristic $\sgpp$ runs 
    $\origgpp$ and finds the largest $M_p$-density suffix of \emph{all} orderings 
    produced by $\origgpp$. This implies that the first iteration of $\sgpp$ is 
    exactly $\sgreedy$. As we use $\origgpp$ for computing $S_1$, we have
    that $\sgpp$ produces a subgraph that has density at least as good as both 
    of our approximation algorithms. It could even potentially produce a larger 
    density as it considers many more orderings than just the ones that correspond to 
    $S_{-\infty}^*$ and $S_1$. We run experiments testing the benefit of iteration on
    real-world graphs in Section~\ref{sec:exp-near-opt}.

%% file: hardness.tex
\section{Hardness of \pdsg}\label{sec:hardness}
  The goal of this section is to prove the following theorem.
  \begin{theorem}\label{thm:hard-main}
    \pdsg is NP-hard for $p \in (-\frac{1}{8}, 0) \cup (0, \frac{1}{4})$ and
    weighted \pdsg is NP-hard for $p \in (-3, 0) \cup (0,1)$.
  \end{theorem}

  In this section, we present our hardness 
  results for $p > 0$. In the appendix, we discuss how one can formally 
  extend the arguments to $p < 0$ (see 
  Appendix~\ref{sec:hardness-neg-app}) and we 
  informally discuss how to extend the argument for more values of $p$ 
  in $(-\infty, 0) \cup (0, 1)$ (see Appendix~\ref{sec:hardness-all-neg-app}). 
  
  Let $p \in (0,1)$. Recall $f_p(S) = \sum_{v\in S} d_S(v)^p$ and 
  $\rho_p(S) = \frac{f_p(S)}{\abs{S}}$. As $p > 0$, $\max_S M_p(S)$ is equivalent to the 
  problem of $\max_S \rho_p(S)$. We therefore focus on the problem 
  of $\max_S \rho_p(S)$. We give a reduction from the standard NP-Complete 
  problem \ec.
  
  \begin{problem}[\ec]
    In the \ec problem, the input is a family of subsets 
    $\cS = \{S_1,S_2,\ldots, S_m\}$ each of cardinality $3$ over a ground set 
    $\cU  = \{e_1,e_2,\ldots,e_{3n}\}$ and the goal is to determine if there 
    exists a collection of sets $S_{i_1}, S_{i_2},\ldots, S_{i_n}$ that form a 
    partition of $\cU$.  If such sets exist, we say that 
    $\{S_{i_1},\ldots, S_{i_n}\}$ is an exact
    $3$-cover.
  \end{problem}
  
  We give a formal definition of weighted \pdsg.
  
  \begin{problem}[weighted \pdsg]
  \label{prob:weighted-p-mean}
    The input is an edge-weighted graph $G = (V,E, c : E \to \R_+)$. 
    The weighted $p$-mean DSG problem is the same as the $p$-mean DSG problem
    if one defines $d_S(v)$ as the weighted degree $\sum_{e \in \delta(v) \cap E(S)} c_e$
    where $\delta(v)$ is the set of edges leaving $v$.
  \end{problem}
  
  \paragraph{The reduction from \ec.}
  We first give the reduction for the weighted case.
  Let $\cS = \{S_1,S_2,\ldots, S_m\}$ and 
  $\cU = \{e_1,\ldots, e_{3n}\}$ be an instance of \ec. We first 
  construct a graph $G = (L \cup A, E)$ as follows. $L$ has a vertex $v_i$ for each 
  set $S_i$ and $A$ has a vertex $u_j$ for each element $e_j$. 
  We add a weight $1$ edge from every ``set vertex" in $L$ to the corresponding 
  ``element vertices" in $A$ that the set contains. Formally, for all $i \in [m]$, we 
  add the edge set $\{(v_i u_j , 1) \mid e_j \in S_i\}$ to $E$ where $(v_iu_j, 1)$ is
  the undirected edge from $v_i$ to $u_j$ of weight $1$.
  Further, we add an edge to $E$ between all pairs of vertices in $A$ 
  with weight $\frac{d}{\abs{A} - 1}$ where $d := 1.23p + 4.77$. 
  Let $\redopt := \frac{3^p + 3(d+1)^p}{4}$ and let 
  $\OPT_G = \argmax_S \rho_p(S)$. The reduction constructs this graph and returns 
  TRUE iff $\OPT_G \ge \rho^*$. We provide an illustration of the reduction in 
  Figure~\ref{fig:hardness}.
  The weights of the edges in our reduction are rational numbers and
  could be made integral by scaling the edge weights by an appropriate integer.
  One could then argue that an unweighted version of the problem defined 
  with multigraphs is NP-Hard.
  
  For the reduction for the unweighted case, we again construct the graph
  $G = (L\cup A, E)$. All weight-$1$ edges from $L$ to $A$ remain 
  but now they are unweighted. Instead of $G[A]$ being a clique with equal weight edges, 
  we let $G[A]$ be a connected $d$-regular graph where $d = 5$. (We can assume
  $n$ is even so that such a $d$-regular graph must exist.) 
  Again, the reduction constructs this graph $G$ and returns TRUE iff $\OPT_G \ge \redopt$.
  
  \begin{figure}[t]
    \centering
    \includegraphics[scale=0.65]{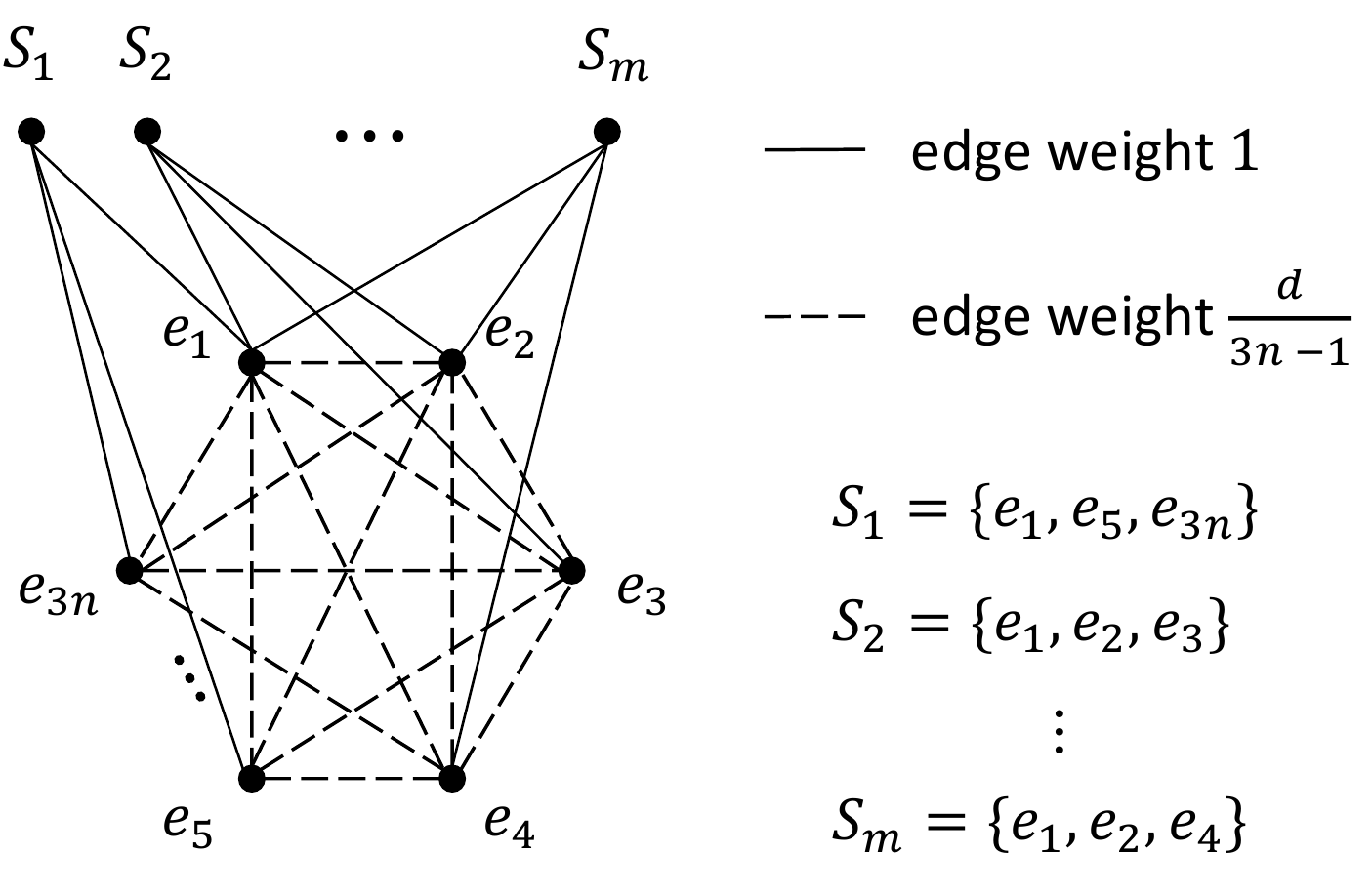}
    \caption{Graph constructed in our reduction from \ec. The \ec instance is the family
    of sets $\{S_1,\ldots, S_m\}$ defined over the ground set 
    $\{e_1,\ldots, e_{3n}\}$.}
    \label{fig:hardness}
  \end{figure}
    
  \paragraph{Proof outline.}
  The degree of the vertices in $A$ in the
  subgraph $G[A]$ is $d = 1.23p + 4.77$ for the weighted case and $d = 5$
  for the unweighted case, which are both larger than the degree-$3$ 
  vertices in $L$. $d$ is chosen to be large enough in both cases so that any 
  optimal solution will necessarily take all of $A$, but small enough so that optimal
  solutions still need to take vertices in $L$. Considering only subsets 
  $S\subseteq L$ of a fixed size, the objective function $\rho_p(S \cup A)$ favors
  solutions $S\cup A$ with more uniform degrees. We take advantage 
  of this to show that an exact 3-cover, which will add exactly one to the degree of each
  vertex in $A$, attains the largest possible density $\redopt$.
  
  Now suppose $\cS$ contains an exact $3$-cover $S_{i_1}, \ldots, S_{i_n}$. Let 
  $S = \{v_{i_1}, \ldots,v_{i_n}\}$ be the subset of $L$ corresponding to the exact 
  $3$-cover. Then as $\abs{S} = n = \frac{\abs{A}}{3}$,
  \begin{equation}\label{eq:hardness-easy-dir}
    \OPT_G 
    \ge
    \frac{f_p(S \cup A)}{\abs{S \cup A}}
    =
    \frac{3^p \cdot \abs{S} + (d+1)^p \abs{A}}{\abs{S} + \abs{A}}
    \\=
    \frac{3^p \cdot \frac{\abs{A}}{3} + (d+1)^p \abs{A}}{\frac{\abs{A}}{3}+\abs{A}}
    =
    \redopt.
  \end{equation}
  
  Now assume $\cS$ does not contain an exact $3$-cover. Letting $S \subseteq L$ and 
  $A' \subseteq A$, we want to show that $\rho_p(S \cup A') < \redopt$. We focus on 
  the case where $A' = A$. Suppose $3\cdot\abs{S} = \alpha \abs{A}$
  for some $\alpha \in \R_{\ge 0}$. We then have 
  \begin{align}
    \rho_p(S \cup A)
    &=\label{eq:outline-1}
    \frac{3^p \cdot \abs{S} + \sum_{v \in A} d_{S\cup A}(v)^p}{\abs{S} + \abs{A}}\\
    &=\label{eq:outline-1.5}
    \frac{3^p \cdot \abs{S} + \sum_{v \in A} (d + d_{S+v}(v))^p}{\abs{S} + \abs{A}}\\
    &\le\label{eq:concavity}
    \frac{3^p \cdot \abs{S} + \abs{A}\cdot (d + \frac{3 \abs{S}}{\abs{A}})^p}
      {\abs{S} + \abs{A}}\\
    &=\label{eq:outline-2}
    \frac{3^p \cdot \alpha + 3\cdot (d + \alpha)^p}{\alpha +3},
  \end{align}
  where the first equality holds as $G[A]$ is $d$-regular, 
  the inequality uses the fact that $\sum_{i=1}^n (d+x_i)^p$ is a concave 
  function in $x$ and, subject to the constraint that $\sum_i x_i = s$, is maximized
  when each $x_i = \frac{s}{n}$, and the final equality holds as
  $3 \cdot\abs{S} = \alpha \abs{A}$. Define
  \[
    g(\alpha) 
    :=
    \frac{3^p \cdot \alpha + 3\cdot (d + \alpha)^p}{\alpha +3}.    
  \]
  Note that $g(1) = \redopt$.
  
  We claim it suffices to choose $d$ such that $g(\alpha)$ is uniquely
  maximized at $\alpha =1$. So if $\alpha \ne 1$, then
  $\rho_p(S \cup A) < g(1)$ by the choice of $d$. 
  Now consider the case when $\alpha =1$. As $\cS$ does not contain an exact 
  $3$-cover, it must be the case that there exists $u,v\in A$ such that
  $d_{S \cup A}(u) \ne d_{S \cup A}(v)$. This implies that Inequality~(\ref{eq:concavity})
  is strict, which again implies $\rho_p(S \cup A) < g(1)$.
  
  A natural direction for the proof would then be to choose 
  $d$ so that $g(\alpha)$ is uniquely maximized at $\alpha = 1$. The main issue with
  this approach is that it leads to solutions where $d$ is irrational even when 
  $p$ is rational. One might think we should then consider reducing from 
  \aec{l} where $\ell>3$ is an integer. However, it is a challenging problem
  to choose such an $\ell$ and $d$. 
  
  To remedy these issues, in Lemma~\ref{lem:int-flat}, we consider a tighter
  bound than the one in Inequality~(\ref{eq:concavity}), which utilizes the fact
  that we are optimizing the function of interest over the integers. We provide a proof of the lemma 
  in Appendix~\ref{sec:app-hardness}.
  
  \begin{lemma}\label{lem:int-flat}
    Let $c \in \R_{\ge 0}$ and $p > 0$. Let $f : \R_{\ge 0}^n \to \R$ be defined as 
    $f(x) = \sum_{i=1}^n (c+x_i)^p$ where $x \in \R^n$ and $x_i$ is the $i$-th coordinate
    of $x$. Consider the program parameterized by $s\in \N$:
    \[
      \text{maximize } f(x) \text{ over } x\in \Z^n, x \ge 0 \text{ s.t. } \sum_{i=1}^n x_i = s.        
    \]
    The program is uniquely maximized at the vector of integers with sum equal to $s$
    and each entry is in $\{\ceil{s/n}, \ceil{s/n} - 1\}$.
    
    If $p < 0$ and the optimization problem is a minimization problem, the program
    is uniquely minimized at the vector of integers with sum equal to $s$ and each
    entry is in $\{\ceil{s/n}, \ceil{s/n}-1\}$.
  \end{lemma}
  
  Now consider again the case when the input $\cS$ does not contain an exact
  3-cover. From the proof outline above, we want to argue that the quantity 
  in~(\ref{eq:outline-1.5}) is strictly less than $g(1)$, which is equal to $\rho^*$.
  Reevaluating~(\ref{eq:outline-1.5})
  with Lemma~\ref{lem:int-flat} in hand,
  we have a tighter bound on $\sum_{v \in A} (d + d_{S+v}(v))^p$ compared
  to the bound in~(\ref{eq:concavity}) as $d_{S+v}(v)$ is integral for all $v \in A$.
  After carefully choosing the value of $d$, this improved bound allows us to prove 
  that the quantity in~(\ref{eq:outline-1.5}) is strictly less than $g(1)$.
  
  We use the proof outline above 
  to prove the following theorem, which is exactly Theorem~\ref{thm:hard-main}
  when assuming $p > 0$.
    
  \begin{theorem}\label{thm:hard-weighted}
    \pdsg is NP-hard for $p \in (0, \frac{1}{4})$ and
    weighted \pdsg is NP-hard for $p \in (0,1)$.
  \end{theorem}
  
  We note again that we prove NP-hardness for values of $p < 0$ in
  the appendix (see Appendix~\ref{sec:hardness-neg-app}).

%% file: experiments.tex
\section{Experiments}\label{sec:experiments}
  
  We evaluate the algorithms described in previous sections on 
  ten real-world graphs. The graphs are publicly accessible from SNAP~\cite{snap} 
  and SuiteSparse Matrix Collection~\cite{suitesparse}. The size of each graph is given in
  Table~\ref{fig:size-table}. The graphs we consider are a subset of the
  graphs from~\cite{vbk-21}. We focused on the graphs 
  in~\cite{vbk-21} for which they reported raw statistics, which was done purposefully 
  to facilitate direct comparisons. Some graphs vary slightly from the graphs in~\cite{vbk-21}
  due to differences in preprocessing as we remove all isolated vertices.
  
  All algorithms are implemented in C++. We use the implementation of 
  $\origgpp$ from~\cite{boob-20}, which  serves as our implementation 
  for $\sgreedy$ and $\sgpp$.  We did not use the implementation for $\greedy$ 
  of~\cite{vbk-21} as it was written in Julia. The reported running times of
  $\greedy$ in~\cite{vbk-21} are roughly twice as fast as our implementation. This 
  difference largely does not impact our results and we mention why this is the case in 
  each section. We do not
  measure the time it takes to read the graph in order to reduce variability (the time it takes 
  to read the largest graph is 6 seconds). The experiments were done on a
  Slurm-based cluster, where we ran each experiment on 1 node with 16 cores. 
  Each node/core had Xeon PHI 5100 CPUs and 16 GB of RAM. 
  
  In each section, we present partial results for a few graphs but include all results 
  in the appendix (see Appendix~\ref{sec:experiments-app}).
  We give a few highlights here and discuss details in the following 
  sections.
  
  \begin{itemize}
    \item 
      For $p > 1$, $\lgreedy$ is much faster than $\greedy$ while returning subgraphs with
      comparable density.
    \item
      For $\lgpp$ with $p > 1$, multiple iterations improve upon the density of $\lgreedy$
      for some graphs.
    \item
       For $p < 1$, $\sgreedy$ returns solutions with similar density to the algorithms we
       compare it against but runs significantly faster.
  \end{itemize}
  
  \begin{table}[t]
    \centering
    \begin{tabular}{r | l l l l l}
      & Astro & CM05 & BrKite & Enron & roadCA\\ \hline
      $n$ & 18,771 & 39,577 & 58,228 & 36,692 & 1,965,206\\
      $m$ & 198,050 & 175,693 & 214,078 & 183,831 & 2,766,607\\
    \end{tabular}

    \bigskip

    \begin{tabular}{r | l l l l l}
      & roadTX & webG & webBS & Amaz & YTube\\ \hline
      $n$ & 1,379,917 & 875,713 & 685,230 & 334,863 & 1,134,890\\
      $m$ & 1,921,660 & 4,322,051 & 6,649,470 & 925,872 & 2,987,624\\
    \end{tabular}
    
    \caption{Sizes of the real-world graphs used in experiments.}
    \label{fig:size-table}
  \end{table}

  \subsection{$\lgreedy$ outperforms $\greedy$}\label{sec:exp-lazy}    
    \begin{table}[t]
    \begin{subtable}{\textwidth}
        \centering
    \begin{tabular}{c|l|llll}
      \hline
      metric & algorithm  & Astro & roadCA & webBS & YTube\\\hline
      \multirow{3}{*}{time (s)} & $\abblgreedy$ ($\eps$$=$$0.1$) & 0.144 & 3.148 & 13.32 & 7.685 \\
       & $\abblgreedy$ ($\eps$$=$$1.0$) & \textbf{0.105} & 3.196 & \textbf{9.952} & \textbf{4.623} \\
       & $\greedy$ & 0.312 & \textbf{3.146} & 687.2 & 75.98 \\
      \hline
      \multirow{3}{*}{\begin{tabular}{c}density \\ ($M_p$) \end{tabular}} & $\abblgreedy$ ($\eps$$=$$0.1$) & 61.6 & \textbf{3.756} & \textbf{207.4} & \textbf{95.5} \\
       & $\abblgreedy$ ($\eps$$=$$1.0$) & \textbf{61.76} & \textbf{3.756} & \textbf{207.4} & \textbf{95.5} \\
       & $\greedy$ & 61.6 & 3.683 & \textbf{207.4} & \textbf{95.5} \\
        \hline
    \end{tabular}
    \caption{Results for $p = 1.25$}
    \end{subtable}
    
    \bigskip
    
    \begin{subtable}{\textwidth}
    \centering
    \begin{tabular}{c|l|llll}
      \hline
      metric & algorithm  & Astro & roadCA & webBS & YTube\\\hline
      \multirow{3}{*}{time (s)} & $\abblgreedy$ ($\eps$$=$$0.1$) & 0.157 & 3.135 & 14.16 & 6.931 \\
       & $\abblgreedy$ ($\eps$$=$$1.0$) & \textbf{0.108} & \textbf{3.104} & \textbf{9.931} & \textbf{4.044} \\
       & $\greedy$ & 0.31 & 3.219 & 742.4 & 68.5 \\
      \hline
      \multirow{3}{*}{\begin{tabular}{c}density \\ ($M_p$) \end{tabular}} & $\abblgreedy$ ($\eps$$=$$0.1$) & 67.12 & \textbf{3.816} & \textbf{387.1} & \textbf{112.2} \\
       & $\abblgreedy$ ($\eps$$=$$1.0$) & \textbf{67.88} & 3.769 & \textbf{387.1} & \textbf{112.2} \\
       & $\greedy$ & 67.12 & \textbf{3.816} & \textbf{387.1} & \textbf{112.2} \\
        \hline
    \end{tabular}
    \caption{Results for $p = 1.75$}
    \end{subtable}
    \caption{Wall clock times and densities of returned solutions of 
    $\lgreedy$ and $\greedy$ on four of the ten graphs and for $p \in \{1.25, 1.75\}$. 
    We write $\lgreedy$ as $\abblgreedy$ for spacing issues.
    Results for all ten graphs and more values of $p$ are in Appendix~\ref{sec:experiments-app}.}
    \label{fig:exp-lazy}
    \end{table}  
    
    We find that $\lgreedy$ outperforms $\greedy$ in terms 
    of running time and typically returns a solution with similar density.
    We run experiments for $p \in \{1.05, 1.25, 1.5,1.75, 2\}$. Note that we stop at 
    $p=2$ as the objective starts rewarding subgraphs with large degrees and therefore 
    the returned subgraphs are often a large fraction of the graph. We evaluate
    $\lgreedy$ with values of $\eps \in \{0.01,0.1,1\}$, all of which result in similar
    running times and densities.
    In Table~\ref{fig:exp-lazy}, we present the running times and densities of returned
    solutions for $\lgreedy$ and $\greedy$ on four of the ten real-world graphs and for
    two values of $p$.  The full table of the results is given in 
    Appendix~\ref{sec:experiments-app}.
    
    Ignoring the road networks roadCA and roadTX, $\lgreedy$ is roughly 
    at least twice as fast as $\greedy$. The road networks have a maximum
    degree of 12, so the savings from $\lgreedy$ are 
    marginal. On graphs webG and webBS, $\lgreedy$ is
    roughly at least $4$ times faster. Since webG and webBS
    have many vertices of large degree, this was expected as we know from the
    analysis of $\lgreedy$, the running time of $\greedy$ depends on the sum of the 
    squared degrees. We therefore suggest the use of $\lgreedy$ over $\greedy$, 
    especially for graphs with large degrees, which we expect in many real-world
    graphs. We also point out that the reported running time difference 
    with~\cite{vbk-21} does not impact these results as this is a relative comparison; 
    we use the same code for $\greedy$ and $\lgreedy$ where the only 
    difference is when the $f_p(v \mid S - v)$ values in the min-heap are updated. 
  
  \subsection{Finding near-optimal solutions for $p > 1$}\label{sec:exp-near-opt}
    \begin{figure}[t]
      \centering
      \includegraphics[scale=0.6]{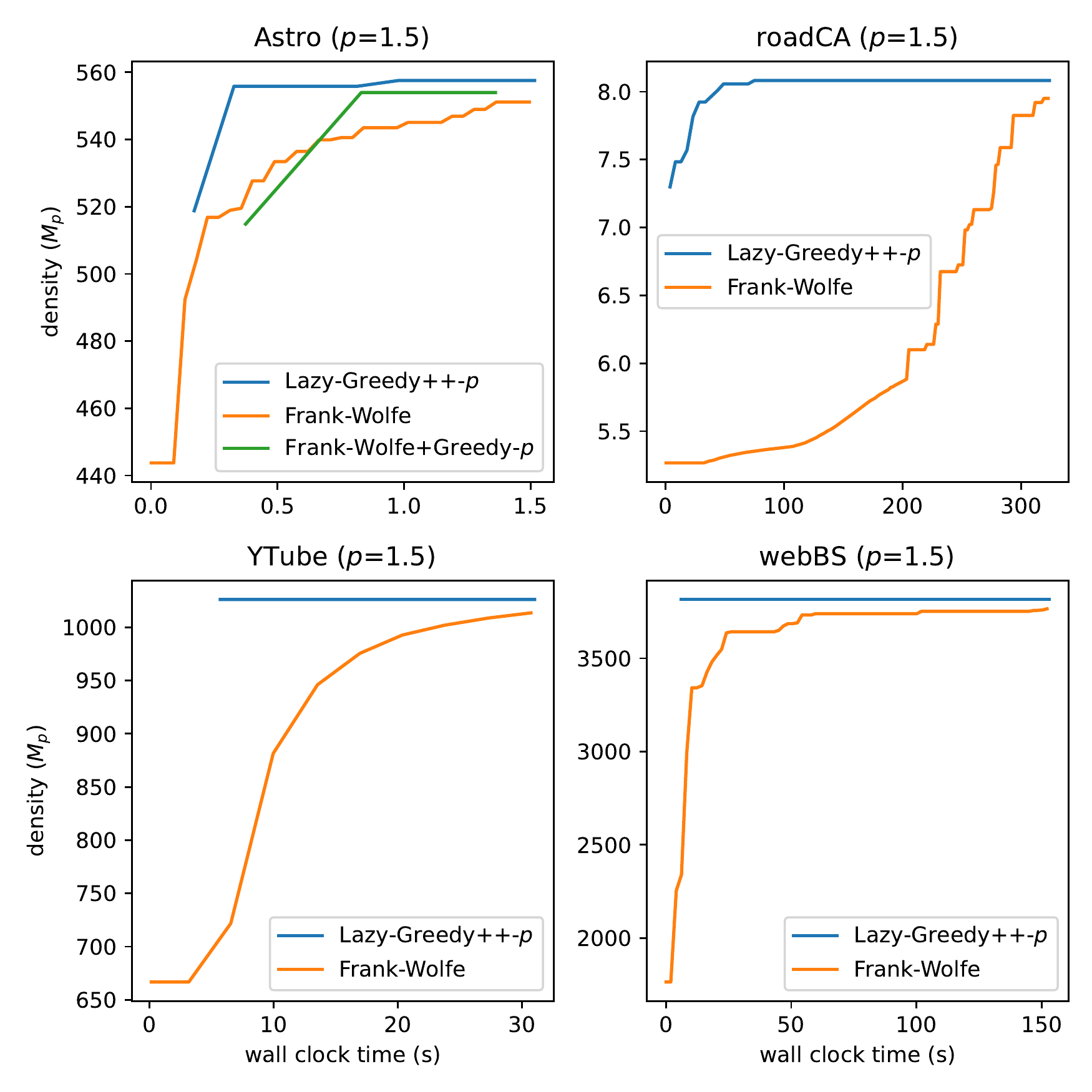}
      \caption{Results for convergence of iterative algorithms for finding near-optimal
      solutions where $p = 1.5$ on four of the ten real-world graphs. Results for all ten 
      graphs and more values of $p$ are in  
      Appendix~\ref{sec:experiments-app}.}
      \label{fig:exp-near-opt}
    \end{figure}
    
    We compare $\lgpp$, $\gpp$, and the Frank-Wolfe method in terms of running time
    and solution quality. In~\cite{hqc-22}, Harb et al.\ show how one can use Frank-Wolfe to solve 
    an appropriate convex programming relaxation of \pdsg. We provide details of the
    algorithm in the appendix (see Appendix~\ref{sec:experiments-app}).
   
    We run $\lgpp$ and
    $\gpp$ for $100$ iterations and Frank-Wolfe for $500$ iterations. We only
    run $\gpp$ on the five smallest graphs for even if we used the faster implementation
    of $\greedy$ in~\cite{vbk-21}, the running time on large graphs could reach 
    roughly around ten hours. Thus, for large graphs, this points to
    $\lgpp$ and Frank-Wolfe being better options than $\gpp$. For $\lgpp$, we set 
    $\eps =1$. We run experiments for $p \in \{1.05, 1.25, 1.5, 1.75, 2\}$. In 
    Figure~\ref{fig:exp-near-opt}, we present plots showing the rate of convergence for 
    four of the ten real-world graphs with $p = 1.5$.
    For each plot, we truncated the data to the first point at
    which all algorithms reached at least $99\%$ of the optimal density achieved in 
    order to more easily see trends in the data. The full collection of results is 
    given in Appendix~\ref{sec:experiments-app}.
    
    We find there is a benefit to iteration; $\lgpp$ produces solutions with
    increasing density for multiple iterations for around half of the graphs (e.g., see
    roadCA and Astro plots in Figure~\ref{fig:exp-near-opt}). We also note that
    the rate of convergence of $\lgpp$ is faster than Frank-Wolfe on almost every single
    graph.

  \subsection{Approximation algorithms for $p < 1$}\label{sec:exp-lt1}
    \begin{table}[t]
      \begin{subtable}{\textwidth}
      \centering
      \begin{tabular}{c|l|llll}
        \hline
        metric & algorithm  & Astro & roadCA & webBS & YTube\\\hline
        \multirow{3}{*}{time (s)} & $\abbsgreedy$ & \textbf{0.022} & \textbf{1.146} & \textbf{0.836} & \textbf{1.151} \\
        & 1-mean DSG & 2.179 & 164.9 & 76.95 & 117.9 \\
        & $\abbsgpp$ & 2.179 & 164.9 & 76.95 & 117.9 \\
        \hline
        \multirow{3}{*}{\begin{tabular}{c}density \\ ($M_p$) \end{tabular}} & $\abbsgreedy$ & 56.91 & 3.27 & \textbf{204.0} & \textbf{75.56} \\
        & 1-mean DSG & 55.52 & 3.583 & \textbf{204.0} & 72.8 \\
        & $\abbsgpp$ & \textbf{58.92} & \textbf{3.75} & \textbf{204.0} & \textbf{75.56} \\
        \hline
      \end{tabular}
      \caption{Results for $p = -1$}
      \end{subtable}
      
      \bigskip
      
      \begin{subtable}{\textwidth}      
      \centering
      \begin{tabular}{c|l|llll}
        \hline
        metric & algorithm  & Astro & roadCA & webBS & YTube\\\hline
        \multirow{4}{*}{time (s)} & $\greedy$ & 1.553 & 4.364 & 4483.8 & 292.3 \\
        & $\abbsgreedy$ & \textbf{0.036} & \textbf{0.857} & \textbf{1.31} & \textbf{0.887} \\
        & 1-mean DSG & 3.772 & 99.07 & 124.0 & 90.44 \\
        & $\abbsgpp$ & 3.772 & 99.07 & 124.0 & 90.44 \\
        \hline
        \multirow{4}{*}{\begin{tabular}{c}density \\ ($M_p$) \end{tabular}}  & $\greedy$ & 56.86 & 3.296 & \textbf{205.9} & \textbf{85.28} \\
        & $\abbsgreedy$ & 57.48 & 3.333 & \textbf{205.9} & 85.23 \\
        & 1-mean DSG & 61.59 & 3.848 & \textbf{205.9} & 84.74 \\
        & $\abbsgpp$ & \textbf{61.87} & \textbf{3.867} & \textbf{205.9} & 85.24 \\
        \hline
      \end{tabular}
      \caption{Results for $p = 0.5$}    
      \end{subtable}
      \caption{Wall clock times and densities of returned solutions for the algorithms 
      we tested for $p \in \{-1, 0.5\}$. We write $\sgreedy$ as $\abbsgreedy$ 
      and $\sgpp$ as $\abbsgpp$ for spacing issues. Results for all 
      ten graphs and more values of $p$ are in 
      Appendix~\ref{sec:experiments-app}.}
      \label{fig:exp-lt1}
    \end{table}
    
    For $p < 1$, we evaluate our two approximation algorithms from
    Section~\ref{sec:apx}, $\greedy$, and the heuristic $\sgpp$ described in
    Section~\ref{sec:iter-heuristics}.  We run all four algorithms on all ten real-world graphs and for  
    $p \in \{0.25, 0.5, 0.75\}$. We only run our approximation algorithms and $\sgpp$
    for $p \in \{-1, -0.5\}$ as $\greedy$ is not well-defined for $p < 0$. 
    We run $\sgpp$ for 100 iterations. For our approximation algorithm returning the $1$-mean DSG,
    as noted in Section~\ref{sec:iter-heuristics}, we use $\origgpp$ to compute a near-optimal solution. 
    In Table~\ref{fig:exp-lt1}, we present the results for 
    four of the ten graphs and for two values of $p$. The full table of results is given 
    in Appendix~\ref{sec:experiments-app}.
    
    We first want to highlight the fact that the orderings produced by $\sgreedy$ and $\sgpp$ 
    are not dependent on $p$. One can therefore run 
    these algorithms once to compute approximations to \pdsg for many different values of
    $p < 1$.
    
    For values of $p \in (0,1)$, all four algorithms return graphs with roughly the same
    density, however, $\sgreedy$ is much faster than the other three algorithms. 
    As $\sgreedy$ is only $1$ iteration of $\sgpp$ and we run $\sgpp$ for 
    $100$ iterations, we find that $\sgreedy$ is roughly around $100$ times faster
    than $\sgpp$ and the algorithm finding the 1-mean densest subgraph. Furthermore, 
    $\sgreedy$ is anywhere from $2$ to $20$ times as fast as $\greedy$.
    For values of $p < 0$, $\sgreedy$ again performs the best as it returns subgraphs of
    comparable density and is also the fastest algorithm. We also want to point out that 
    although $\sgpp$ sometimes produces subgraphs with larger density than all
    algorithms tested, the benefit of iteration with $\sgpp$ is marginal.

%% file: conclusion.tex
\section{Conclusion}

In this paper, we provide a deeper understanding of \pdsg for the entire range
of $p$. For $p \in (-\frac{1}{8}, 0) \cup (0, \frac{1}{4})$, we show \pdsg is NP-Hard
via a nontrivial reduction and we show how to extend this NP-hardness proof to 
$p \in (-3, 0) \cup (0,1)$ for the weighted version of the problem. This mostly 
resolves the open question regarding the complexity of \pdsg for $p < 1$.
We presented two approximation algorithms for $p < 1$, which are the
first algorithms with any provable approximation ratio for this regime of $p$. We
provided a faster implementation of $\greedy$ for $p > 1$ called $\lgreedy$. We 
developed a heuristic based on iterative algorithms for finding near-optimal solutions 
to \pdsg that utilized $\lgreedy$. This algorithm is ideal when $\lgreedy$ does not find 
a near-optimal solution after a single iteration.
Experiments show that our approximation algorithms for both $p > 1$ and 
$p < 1$ are highly effective and scalable.

We outline a few future directions. We established NP-Hardness
for \pdsg where $p \in (-\frac{1}{8},0) \cup (0,\frac{1}{4})$ and
for weighted \pdsg where $p\in (-3, 0)\cup(0, 1)$ and also provided simple $1/2$-approximation
algorithms. We believe that our NP-Hardness proof can be strengthened
to show that it is APX-Hard to approximate the optimum, that is,
there is some fixed $\delta > 0$ such that it is NP-Hard to obtain
a $(1-\delta)$-approximation. The resulting $\delta$ is likely to be
quite small. It is important to obtain tight approximation bounds
for \pdsg that closes that gap between $1/2$ and $(1-\delta)$.  This effort is
likely to lead to the development of more sophisticated approximation algorithms
and heuristics for the problem.

\paragraph{Acknowledgements.} The authors would like to thank Farouk Harb for helpful 
discussions and for providing code for the Frank-Wolfe algorithm.

%% file: appendix.tex
\section{Omitted Proofs from Section~\ref{sec:prelim}}
\label{sec:prelim-proofs-app}
  Proposition~\ref{prop:means} is folklore and we only include the following proof 
  for the sake of completeness.
  \begin{proof}[Proof of Proposition~\ref{prop:means}]
    It suffices to show that $M_p(S)$ is an increasing function in $p$ for a fixed $S$. 
    Fix $S \subseteq V$ and assume $p \not \in \{-\infty, 0, \infty\}$. We have
    \begin{align}
      \frac{d}{dp}\,M_p(S)
      &=
      M_p(S) \cdot \frac{d}{dp}\, \ln M_p(S)\notag\\
      &\ge 
      \frac{d}{dp}\, \ln M_p(S)\notag\\
      &=
      -\frac{1}{p^2} \ln \rho_p(S) + \frac{1}{p\cdot f_p(S)}\cdot\frac{d}{dp}\,f_p(S)\notag\\
      &=
      -\frac{1}{p^2} \ln \rho_p(S) + \frac{1}{p}\cdot \frac{\sum_{v \in S} d_S(v)^p\ln d_S(v)}
        {\sum_{v \in S} d_S(v)^p}\label{eq:objective-deriv}
    \end{align}
    where~(\ref{eq:objective-deriv}) is at least $0$ if and only if
    \begin{equation}\label{eq:M_p-increasing}
      \ln \rho_p(S) 
      \le 
      \frac{\sum_{v \in S} d_S(v)^p\ln d_S(v)^p}{\sum_{v \in S} d_S(v)^p}.
    \end{equation}
    Rearranging,~(\ref{eq:M_p-increasing}) is equivalent to
    \[
      \ln(1/\abs{S}) 
      \le 
      \sum_{v \in S} x_v \ln x_v
    \]
    where $x_v := \frac{d_S(v)^p}{\sum_{v \in S} d_S(v)^p}$.
    As $x$ is a probability distribution over $S$, the final inequality
    holds by recognizing that the uniform distribution maximizes the entropy function.
    Thus, $M_p(S)$ is an increasing function in $p$ over $(-\infty, 0)$ and $(0, \infty)$.
    One can extend the proof for $p \in \{-\infty, 0, \infty\}$ by recognizing that 
    the definition for $M_p(S)$ in these cases is derived by taking the limit
    (e.g., $\lim_{p\to\infty} M_p(S) = M_\infty(S)$). This concludes the proof.
  \end{proof}

\section{Omitted Proofs from Section~\ref{sec:fast}}
\label{sec:fast-proofs-app}
  \begin{proof}[Proof of Lemma~\ref{lem:approx-peel}]
    We have that $f_p$ is supermodular for $p \ge 1$~\cite{vbk-21}.
    Let $S^*$ be an optimal subset and $\rho^* = \rho_p(S^*)$. In the proof of 
    Theorem 3.1 of~\cite{cqt-22}, it is shown that for all $v \in S^*$, 
    $f_p(v \mid S^* - v) \ge \rho^*$.
    
    Let $v_i$ be the first element of $S^*$ that is peeled. Let $S_i$ be the set of
    vertices including $v_i$ and all subsequent vertices in the peeling order. 
    By the update rule
    in~(\ref{eq:apx}), we have $f_p(v_i \mid S_i - v_i) \le (1+\eps)\min_{u \in S_i}
    f_p(u \mid S_i - u)$.
    In Proposition 3.2 of~\cite{cqt-22}, it is shown that for all $S \subseteq V$,
    we have $\sum_{v\in S} f_p(v \mid S - v) \le (p+1)f(S)$. Repeating the 
    argument in Theorem 3.1 in~\cite{cqt-22} only with a small change, we have
    \begin{align*}
      \frac{f_p(S_i)}{\abs{S_i}}
      &=
      \frac{f_p(S_i)}{\sum_{u \in S_i}f(u\mid S_i - u)}
        \cdot \frac{\sum_{u \in S_i}f(u \mid S_i - u)}{\abs{S_i}}\\
      &\ge
      \frac{1}{p+1}\cdot\frac{\abs{S_i}\cdot f(v_i \mid S_i - v_i)}{(1+\eps)\abs{S_i}}\\
      &\ge 
      \frac{\rho^*}{(1+\eps)(p+1)}\\
      &\ge 
      \frac{(1-\eps)\rho^*}{p+1}.
    \end{align*}
    Raising both sides of the inequality to the $\frac{1}{p}$-th power, we have
    $M_p(S_i) \ge (\frac{1-\eps}{p+1})^{1/p} \cdot M_p^*$.
    This concludes the proof.
  \end{proof}
  \begin{proof}[Proof of Lemma~\ref{lem:fast-ind-approx}]
      Note that it suffices to consider $\alpha = 1+\frac{\eps}{p}$ as the left-hand
      side is increasing in $\alpha$. Consider the function
      \begin{equation}\label{eq:f}
        \frac{(\alpha d)^p - (\alpha d - 1)^p}{e^\eps}.
      \end{equation}
      Suppose we can show this function is decreasing in $\eps$. This would imply
      $\frac{(\alpha d)^p - (\alpha d-1)^p}{e^\eps} \le d^p - (d-1)^p$, which would
      prove the statement as $e^\eps \le 1 + 2\eps$. We prove that the function
      is decreasing in $\eps$ by arguing that the derivative of~(\ref{eq:f}) 
      with respect to $\eps$ is non-positive. 
        
      Note that $\frac{d}{d\eps} \alpha d = \frac{d}{p}$. We have
      \[
        \frac{d}{d\eps} \frac{(\alpha d)^p - (\alpha d - 1)^p}{e^\eps}
        =
        \frac{d[(\alpha d)^{p-1} - (\alpha d - 1)^{p-1}]e^\eps - [(\alpha d)^p - (\alpha d - 1)^p]
          e^\eps}{\eps^{2\eps}}.
      \]
      Since we need this quantity to be non-positive, it suffices to show
      \[
        d
        \le
        \frac{(\alpha d)^p - (\alpha d - 1)^p}{(\alpha d)^{p-1} - (\alpha d - 1)^{p-1}}
      \]
      We have
      \[
        \frac{(\alpha d)^p - (\alpha d - 1)^p}{(\alpha d)^{p-1} - (\alpha d - 1)^{p-1}}
        =
        \alpha d \cdot \frac{1 - (1 - \frac{1}{\alpha d})^p}{1 - (1 - \frac{1}{\alpha d})^{p-1}}
        \ge 
        \alpha d
        \ge
        d,
      \]
      where the first inequality holds as
      $(1 - \frac{1}{\alpha d})^p \le (1 - \frac{1}{\alpha d})^{p-1}$.  This concludes
      the proof.
    \end{proof}  
  
  \begin{proof}[Proof of Lemma~\ref{lem:fast-time}]
      The initialization of $D$ and $D'$ takes $O(m)$ time and the initialization of
      $A$ takes $O(m + n\log n)$ time. At each iteration, as $A$ is a min-heap, finding
      the minimum and removing it only takes $O(1)$ time and $O(\log n)$ time,
      respectively. Suppose $v_i$ is the vertex chosen at the $i$-th iteration.
      We consider the first three lines of the inner for loop. The number of iterations
      of the inner for loop is $O(d_G(v_i))$ and the first three lines take $O(\log n)$
      time. Thus, the total time for these three lines throughout the entire run of the
      algorithm is $O(m \log n)$ time.
      
      The only part of the algorithm unaccounted for is the conditional statement in the 
      inner for loop. In this if-statement,
      for a vertex $u$, we update $A[w]$ for all $w \in N(u) \cap S_{i+1}$ only 
      when the approximate degree, $D'[u]$, exceeds 
      $(1 + \frac{\eps}{p})D[u]$, where $D[u]$ is the exact current degree. 
      $u$ is therefore updated only when the actual degree drops by at least a 
      $1+\frac{\eps}{p}$ factor. Since $D[u] \le n$, the maximum number of times
      we update any vertex is $O(\log_{1+\frac{\eps}{p}}(n))$. If $p \ge \eps$, this is at
      most $O(\frac{p\log n}{\eps})$.
      The time it takes to update $A$ is $O(d_G(u)\log n)$. Thus, the total
      time spent on a single vertex $u$ throughout the entire run
      of the algorithm is $O\left(d_G(u) \cdot \frac{\log^2n}{\log(1+\frac{\eps}{p})}\right)$. This 
      concludes the proof.
    \end{proof}
    
        \begin{proof}[Proof of Theorem~\ref{thm:fast}.]
      Running time is analyzed in Lemma~\ref{lem:fast-time}, so we focus on
      the approximation guarantee. Consider an arbitrary iteration $i$ of
      the algorithm. By Lemma~\ref{lem:approx-peel}, it suffices to show that
      $f_p(v_i \mid S_i - v_i) \le (1+\eps)\min_{v \in S_i} f_p(v \mid S_i - v)$.
      We show something stronger. Let $A_i$ be the state of the min-heap $A$ at
      the start of iteration $i$. We show that for all $v \in S_i$, we have
      $A_i[v] \le (1+\eps)f_p(v \mid S_i - v)$. To prove this, we first note that this
      holds on the first iteration as we compute the values $f_p(v \mid V - v)$ exactly. Let
      $D_i'$ be the state of $D'$ at the beginning of iteration $i$. As 
      $D$ maintains exact degrees, we use $d_{S_i}(v)$ in place of $D[v]$. Note that
      the algorithm maintains $A$ to ensure
      $A_i[v] = d_{S_i}(v)^p + \sum_{u \in N(v) \cap S_i} D_i'[u]^p - (D_i'[u]-1)^p$.
      We see that $D'[u] \le (1 + \frac{\eps}{p})d_{S_i}(u)$ is guaranteed for
      all $u \in S_i$
      by the condition in the inner for loop of iteration $i-1$ (if $i = 1$, then this statement
      holds as $D_1'[v] = d_G(v)$). Therefore,
      \[
        A_i[v]
        \le
        d_{S_i}(v) + (1+\eps)\sum_{u \in N(v) \cap S_i} d_{S_i}(u)^p - (d_{S_i}(u)-1)^p
      \]
      by Lemma~\ref{lem:fast-ind-approx}.
      Using the rewriting of $f_p(v \mid S_i - v)$ given in Equation~(\ref{eq:marginal}),
      it follows that $A_i[v] \le (1+\eps)f_p(v \mid S_i - v)$. This
      concludes the proof.
    \end{proof}

\section{Omitted Proofs from Section~\ref{sec:apx}}
    \begin{proof}[Proof of Lemma~\ref{lem:bad-construct}]
      We start with~(1), which follows from the
      fact that $M_p^*(H) \ge \alpha d$.
      
      We next prove~(2). The algorithm $\sgreedy$ (from 
      Section~\ref{sec:apx-maxcore}) first peels $H$ as $M_{-\infty}^*(H) = d$, then it
      peels $K_{d+1,D}$, and it finally peels all of the cliques.
      We have to analyze the density of all suffixes of this ordering and show 
      that all have density at most $d+2$ as $r \to \infty$.
      
      We start by analyzing the suffixes while peeling $H$. 
      Let $H_i$ be the remaining vertices from 
      $H$ after peeling $i$ vertices for $i = 0,1,\ldots, n_H$.
      Let $T$ be the vertices of $G$ without $H$, so $\abs{T} = (d+1+D) + r(d+3)$.
      Thus, on the suffixes while peeling $H$, the algorithm obtains value
      \begin{align*}
        \max_{i=0,1,\ldots,n_H} \left(\frac{\sum_{v \in H_i} d_{H_i}(v)^p 
          + f_p(T)}{\abs{H_i} + (d+1+D) + r(d+3)}\right)^{1/p}.
      \end{align*}
      We have 
      \[
        f_p(T) = (d+1)D^p  + D(d+1)^p + r(d+3)(d+2)^p,
      \] 
      so $f_p(T) \le 2(d+1)D + r(d+3)(d+2)^p$ and note that $dD = o(r)$. 
      As $\sum_{v \in H_i} d_{H_i}(v)^p \le n_H^{1+p}$ and $n_H^{1+p} = o(r)$, 
      \[
        \lim_{r\to\infty}\frac{\sum_{v \in H_i} d_{H_i}(v)^p + f_p(T)}
          {\abs{H_i} + (d+1+D)+ r(d+3)}
        =
        (d+2)^p.
      \]
      Thus, as $r \to \infty$, the maximum density attained while peeling
      $H$ is at most $d+2$.
      
      We next analyze the suffixes while peeling $K_{d+1,D}$. We note
      that the density achieved while peeling $K_{d+1,D}$ is at most
      \[
        \max_{\substack{\ell \in [1,d+1] \\ q \in [1,D]}}
          \left(\frac{\ell q^p + q \ell^p + r(d+3)(d+2)^p}{\ell + q + r(d+3)}\right)^{1/p}.
      \]
      As $\ell q^p + q\ell^p \le 2 (d+1)D$ and $dD = o(r)$, we have 
      \[
        \lim_{r \to \infty}
        \frac{\ell q^p + q \ell^p + r(d+3)(d+2)^p}{\ell + q + r(d+3)}
        =
        (d+2)^p.
      \] 
      Therefore, $\lim_{r \to \infty} M_p(S_1) = d+1$.
      
      We now prove~(3). We start by showing that $M_1^*(G)$
      is $K_{d+1,D}$. We have 
      \[
        M_1(K_{d+1,D})
        = 
        \frac{2(d+1)D}{d+1 +D}
        > 
        2d
      \]
      for $D \ge 2d^2$.
      Furthermore, as $M_{-\infty}^*(H) = d$, Proposition~\ref{prop:core-to-dsg} 
      implies that $M_1^*(H) \le 2d$. We also have that 
      $M_1^*(K_{d+3}) = d+2$ and $M_1(K_{d+1,D}) > d+2$ for
      $D \ge d+6$. Therefore, $M_1^*(G) = K_{d+1,D}$.
      
      All that remains is to upper bound $M_p(K_{d+1,D})$. We have
      \begin{align*}
        M_p(K_{d+1,D})
        &=
        \left(\frac{(d+1)D^p + D(d+1)^p}{d+1+D}\right)^{1/p}.
      \end{align*}
      As $p < 1$, we have $D^p = o(D)$, which implies
      \[
        \lim_{D \to \infty} \frac{(d+1)D^p + D(d+1)^p}{d+1+D}
        =
        (d+1)^p.
      \]
      Therefore, $\lim_{D \to \infty} M_p(K_{d+1,D}) = d+1$.
    \end{proof}      

    To prove Theorem~\ref{thm:comparing-p-to-1}, we construct a graph whose
    degeneracy is roughly half of the density of an optimal solution to \pdsg{p}.
    Note that Proposition~\ref{prop:core-to-dsg} shows $M_p^* \le 2M_{-\infty}^*$.
    We therefore attempt to construct a $d$-degenerate graph where as
    many vertices as possible have degree $2d$. Towards this end, 
    we say that a $d$-degenerate graph is 
    \emph{edge maximal} if adding any edge to the graph makes the graph
    $(d+1)$-degenerate. The following theorem gives a characterization of 
    edge-maximal $d$-degenerate graphs via degree sequences.
  
    \begin{theorem}[\cite{b-12}]\label{thm:bickle}
      Let $d$ be a positive integer.
      A sequence of positive integers $d_1 \ge d_2 \ge \cdots \ge d_n$ is the degree 
      sequence of an edge-maximal $d$-degenerate graph if and only if
      \begin{enumerate}
        \item
          $d \le d_i \le \min\{n-1, n + d - i\}$ for all $i \in [n]$, and 
        \item
          $\sum_{i=1}^n d_i = 2\left(dn - \binom{d+1}{2}\right)$.
      \end{enumerate}
    \end{theorem}
    
    We use this theorem of Bickle in the following proof.

    \begin{proof}[Proof of Theorem~\ref{thm:comparing-p-to-1}]
      Consider the following sequence:
      \[
        d, d+1, d+2, \ldots, 2d- 1, 
          \underbrace{2d-1, 2d-1, \ldots, 2d-1}_{\text{$\binom{d+1}{2}$ copies}},
          \underbrace{2d, 2d, \ldots, 2d}_{\text{$n - d - \binom{d+1}{2}$ copies}}
      \] 
      It is easy to check that both conditions (1) and (2) of Theorem~\ref{thm:bickle}
      hold. Thus, given such integers $n$ and $d$, we can construct a 
      $d$-degenerate graph on $n$ vertices $G$ with the given degree 
      sequence.
      
      We now argue that $M_p(G) \ge (1-\eps)2d$, which would prove the theorem. 
      For any $p \in (-\infty, 1) \setminus \{0\}$, we have that $M_p(G)$ is
      \begin{align}
        &\left(\frac{(n-d-\binom{d+1}{2})(2d)^p + \binom{d+1}{2}(2d-1)^p 
          +\sum_{i=0}^{d-1} (d + i)^p}{n}\right)^{1/p} \\      
        &\ge\label{eq:mpg-rewrite}
        \left(\frac{n +\binom{d+1}{2}((1 - \frac{1}{2d})^p - 1)
          +d(2^{-p} - 1)}{n}\right)^{1/p} \cdot 2d,        
      \end{align} 
      where we use the fact that $(\frac{a+bx^p}{c})^{1/p}$ is an increasing
      function in $x$ when $x$, $a$, $b$, and $c$ are all positive.
      
      For $p \in (0,1)$, using the rewriting of $M_p(G)$ above in~(\ref{eq:mpg-rewrite}), 
      we have
      \[
        \frac{M_p(G)}{2d}
        \ge
        \left(1 - \frac{\frac{p(d+1)}{2} + dp}{n}\right)^{1/p}
        \ge
        \left(1 - \frac{\eps(\frac{p(d+1)}{2} + dp)}{2d}\right)^{1/p}
        \ge
        1 - \eps\cdot \frac{\frac{d+1}{2} + d}{2d}
        \ge 
        1-\eps,
      \]
      where the first inequality holds as $1 - (1 - \frac{1}{2d})^p \le \frac{p}{d}$
      and $1 - 2^{-p} \le p$, the second inequality uses the fact that 
      $n \ge \frac{2d}{\eps}$, and the third inequality holds as
      $(1+x)^{1/p} \ge 1+\frac{x}{p}$ for $x \ge 0$.

      Now for $p < 0$, again using the rewriting of $M_p(G)$
      in~(\ref{eq:mpg-rewrite}), we have
      \begin{equation}\label{eq:p-le-0}
        \frac{M_p(G)}{2d}
        \ge
        \left(1 + \frac{\binom{d+1}{2}((1 - \frac{1}{2d})^p - 1) + d(2^{-p} - 1)}{n}\right)^{1/p}
        \ge
        \left(1 + \abs{p}\cdot \eps\right)^{1/p} 
        \ge
        1 - \eps,
      \end{equation}  
      where the first inequality holds as
      $n \ge \frac{\binom{d+1}{2}((1 - \frac{1}{2d})^p - 1) + d(2^{-p} - 1)}
      {\eps\cdot \abs{p}}$, and
      the last inequality holds as 
      $1 - p\cdot \eps \le (1 - \eps)^p $.
      
      If we were to assume $p \in (-1,0)$, from~(\ref{eq:p-le-0}) above, we have
      \[
        \frac{\binom{d+1}{2}((1 - \frac{1}{2d})^p - 1) + d(2^{-p} - 1)}
          {\eps\cdot \abs{p}}
        \le
        \frac{\binom{d+1}{2}\frac{\abs{p}}{d} + 2d\abs{p}}{\eps \cdot \abs{p}}
        \\=
        \frac{\frac{d+1}{2} + d}{\eps}
        \le
        \frac{2d}{\eps},
      \]
      where the first inequality holds as $(1 - \frac{1}{2d})^p - 1 \le \frac{\abs{p}}{d}$ and 
      $2^{-p} - 1 \le \abs{p}$. Thus, for the case of $p \in (-1,0)$, it suffices to take
      $n \ge \frac{2d}{\eps}$. This concludes the proof.
    \end{proof}

\section{Omitted Proofs from Section~\ref{sec:hardness}}
  \label{sec:app-hardness}
  \subsection{Hardness for $p \in (0,1)$}
  We start by providing the proof of Lemma~\ref{lem:int-flat}.
  
  \begin{proof}[Proof of Lemma~\ref{lem:int-flat}]
    Fix $s \in \N$.  Let $x^*$ be a vector such that 
    $x_i^* \in \{\ceil{s/n}, \ceil{s/n}-1\}$ for all $i \in [n]$ and $\sum_i x_i^* = s$.
    Let $x \in \Z^n$ and $x \ge 0$ such that $\sum_i x_i = s$ and at least one entry of $x$ 
    is not in $\{\ceil{s/n}, \ceil{s/n}-1\}$. Let $2r = \sum_{i=1}^n \abs{x_i - x_i^*}$. 
    We prove that $f(x^*) > f(x)$ by induction on $r$.
        
    For the base case $r = 1$, there are three cases to consider. 
    We can handle the first two cases together. In the first case, we have
    $(x_1^*, x_2^*) = (\ceil{s/n}, \ceil{s/n})$ and 
    $(x_1,x_2) = (\ceil{s/n}+1, \ceil{s/n}-1)$ and $x_i = x_i^*$ for all $i = 3,4,\cdots,n$. 
    In the second case, we have
    $(x_1^*, x_2^*) = (\ceil{s/n}-1,\ceil{s/n}-1)$ and
    $(x_1,x_2) = (\ceil{s/n}, \ceil{s/n}-2)$ and $x_i = x_i^*$ for all $i = 3,4,\cdots,n$.
    For both cases,
    all we have to show is that 
    $(c+x_1)^p + (c+x_2)^p < (c+x_1^*)^p + (c+x_2^*)^p$. As 
    $x_1^* = x_2^*$, this follows immediately from the fact that the function
    $\sum_{i=1}^n (c + y_i)^p$ is concave in $y$ and, subject to the constraint
    that $\sum_i y_i = s$, is maximized when all terms are the same 
    (i.e.\ $y_i = \frac{s}{n}$).
        
    For the third case, $x_1 = \ceil{s/n} + 1$ and $x_1^* = \ceil{s/n}$, 
    $x_2 = \ceil{s/n} - 2$ and $x_2^* = \ceil{s/n}-1$, and $x_i = x_i^*$ for 
    all $i =3,4,\ldots,n$. Again, we have to show 
    $(c+x_1)^p + (c+x_2)^p < (c+x_1^*)^p + (c+x_2^*)^p$. We have 
    $x_1 > x_1^* > x_2^* > x_2$. As $(c+x)^p$ is concave in $x$ and 
    $x_1 - x_1^* = x_2^* - x_2$, we have 
    $(c+x_1)^p - (c+x_1^*)^p < (c+x_2^*)^p - (c+x_2)^p$. 
    This gives us the inequality we want up to rearranging.
        
    For the inductive step, we assume that the statement holds for $r \ge 1$ 
    and prove that it holds for $r+1$. Suppose that $x \in \Z^n$ and $x \ge 0$ such that 
    $\sum_i x_i = s$, at least one entry of $x$ is not in $\{\ceil{s/n},\ceil{s/n}-1\}$,
    and $2(r+1) = \sum_{i=1}^n \abs{x_i - x_i^*}$. Using a similar argument 
    as for the base case, we can reduce $\sum_{i=1}^n \abs{x_i-x_i^*}$ by 
    $2$. This concludes the proof for $p > 0$.
    
    For the case of $p < 0$, it is easy to see that everything still holds 
    as $\sum_{i=1}^n(c+y_i)^p$ is convex in $y$.
  \end{proof}
  
  In order to prove Theorem~\ref{thm:hard-weighted}, we first need a few helpful
  technical lemmas.
  
    \begin{lemma}\label{lem:strict-inc}
    Let $d \ge 1$ be an integer and $p \in \R$. 
    Define $h : [0,1] \to \R$ as 
    \[
      h(\beta) 
      :=
      \frac{3^p\beta + 3\beta (d+1)^p + 3(1-\beta) d^p}{\beta + 3}.
    \]
    Let $h'$ denote the derivative of $h$ with respect to $\beta$. Fix $\beta \in [0,1]$. 
    Then $h'(\beta) > 0$ if and only if $3^p + 3(d+1)^p > 4d^p$.
  \end{lemma}
  \begin{proof}
    We have that $\frac{d}{d\beta} h(\beta)$ is equal to
    \[
      \frac{(3^p + 3(d+1)^p - 3 d^p)(\beta + 3) - 
        (3^p\beta + 3\beta (d+1)^p + 3(1-\beta) d^p)}{(\beta + 3)^2}.
    \]
    Simplifying the numerator, $h'(\beta) > 0$ if and only if 
    \[
      3^{1+p} + 3^2(d+1)^p - 3^2 d^p - 3 d^p > 0.
    \]
    The lemma follows after dividing through by $3$ and rearranging. 
  \end{proof}
  \begin{lemma}\label{lem:strict-dec}
    Let $d,t \ge 1$ be integers, $p \in \R$, and $\beta \in [0,1]$. 
    Define $g:[0,1] \to \R$ as
    \[
      g(\beta)
      :=
      \frac{3^p (t+\beta) + 3\beta(d+t+1)^p + 3(1-\beta) (d+t)^p}{t +\beta + 3}.
    \]

    \begin{enumerate}\itemsep0pt
      \item\label{lem:strict-dec-case-1}
        Assume $p \in (0,1)$. If $4(d+2)^p + 3^p < 5(d+1)^p$, then $g$ is strictly decreasing.
      \item\label{lem:strict-dec-case-2}
        Assume $p < 0$. If $4(d+2)^p + 3^p > 5(d+1)^p$, then $g$ is strictly increasing.      
    \end{enumerate}
  \end{lemma}
  \begin{proof}
    We have $\frac{d}{d\beta} g(\beta)$ is
    \begin{multline*}
      \frac{(3^p + 3(d+t+1)^p - 3(d+t)^p)(t+\beta+3)}{(t+\beta+3)^2} \\
        - \frac{(3^p (t+\beta) + 3\beta(d+t+1)^p + 3(1-\beta) (d+t)^p)}{(t+\beta+3)^2}.
    \end{multline*}
    Simplifying the numerator of both terms, $\frac{d}{d\beta} g(\beta) < 0$ if and only if
    \[
      3t(d+t+1)^p - 3t(d+t)^p + 3^{1+p} + 3^2(d+t+1)^p
        -3^2(d+t)^p - 3 (d+t)^p 
      < 
      0.
    \]
    Dividing through by $3$ and rearranging, this is equivalent to
    \begin{equation}\label{eq:strict-dec-main}
      (4+t)(d+t)^p - (3+t)(d+t+1)^p - 3^p
      >
      0.
    \end{equation}
    To prove~(\ref{lem:strict-dec-case-1}), it suffices to show the left-hand
    side of~(\ref{eq:strict-dec-main}) is increasing in $t$ for $p \in (0,1)$. To 
    prove~(\ref{lem:strict-dec-case-2}), it suffices to show the left-hand side
    of~(\ref{eq:strict-dec-main}) is decreasing in $t$ for $p < 0$.
    
    Assume $p > 0$.
    The derivative of the left-hand side of~(\ref{eq:strict-dec-main}) with respect to $t$ is
    \begin{align*}
      &(d+t)^p + p(4+t)(d+t)^{p-1} - (d+t+1)^p - p(3+t)(d+t+1)^{p-1}\\
      &\ge 
      -p(d+t)^{p-1} + p(4+t)(d+t)^{p-1} - p(3+t)(d+t+1)^{p-1}\\
      &=
      p(3+t)[(d+t)^{p-1} - (d+t+1)^{p-1}]\\
      &\ge 
      0,
    \end{align*}
    where the first inequality follows from the fact that $(x+1)^p - x^p \le px^{p-1}$
    for all $x > 0$ as $x^p$ is concave and the second inequality follows from the
    fact that $x^{p-1}$ is a decreasing function in $x$ as $p <1$. This 
    proves~(\ref{lem:strict-dec-case-1}). 
    
    Now assume $p < 0$. As $x^p$ is convex in $x$, we have $(x+1)^p - x^p \ge 
    px^{p-1}$ for all $x >0$. Thus, one can use the same argument as for $p > 0$
    to show that the derivative of~(\ref{eq:strict-dec-main}) is non-positive. This
    proves~(\ref{lem:strict-dec-case-2}). 
  \end{proof}
  
  The following lemma considers the functions
  \begin{equation}\label{eq:good-d-1}
    3^p + 3(d+1)^p - 4d^p
  \end{equation}
  and
  \begin{equation}\label{eq:good-d-2}
    3^p + 4(d+2)^p - 5(d+1)^p,
  \end{equation}
  where $d = 1.23p +4.77$ for the weighted case and $d = 5$ for the unweighted case.
  The lemma states that for both the weighted and unweighted cases,~(\ref{eq:good-d-1}) is 
  strictly positive and~(\ref{eq:good-d-2}) is strictly negative. 
  This is easy to compute explicitly for 
  specific values of $p$; for example, when $p = 1/2$ and $d = 1.23p+4.77$, 
  the first function has value
  $\approx 0.0303$ and the second function has value $\approx -0.032$.
  
  One approach to proving the lemma for the weighted case would be to
  show that the derivative of~(\ref{eq:good-d-1}) only has one zero on the interval $[0,1]$ 
  and one can show this critical point is a maximum. Since~(\ref{eq:good-d-1}) is
  $0$ when $p \in \{0,1\}$, this would prove~(\ref{eq:good-d-1}) is strictly positive when
  $p \in (0,1)$. A similar approach can be taken 
  for~(\ref{eq:good-d-2}). For the unweighted case, the same proof idea could be used
  for verifying~(\ref{eq:good-d-2}). For~(\ref{eq:good-d-1}), one would only have to show the 
  function is strictly increasing over $(0, \frac{1}{4})$. 
  Formal verification of these facts would be quite lengthy and 
  not necessarily illuminating, so we therefore verified these inequalities analytically
  via MATLAB. To give an idea of what these functions look like, we plot both in
  Figure~\ref{fig:good-d}.

  \begin{figure}[t]
    \centering
    \includegraphics[scale=0.5]{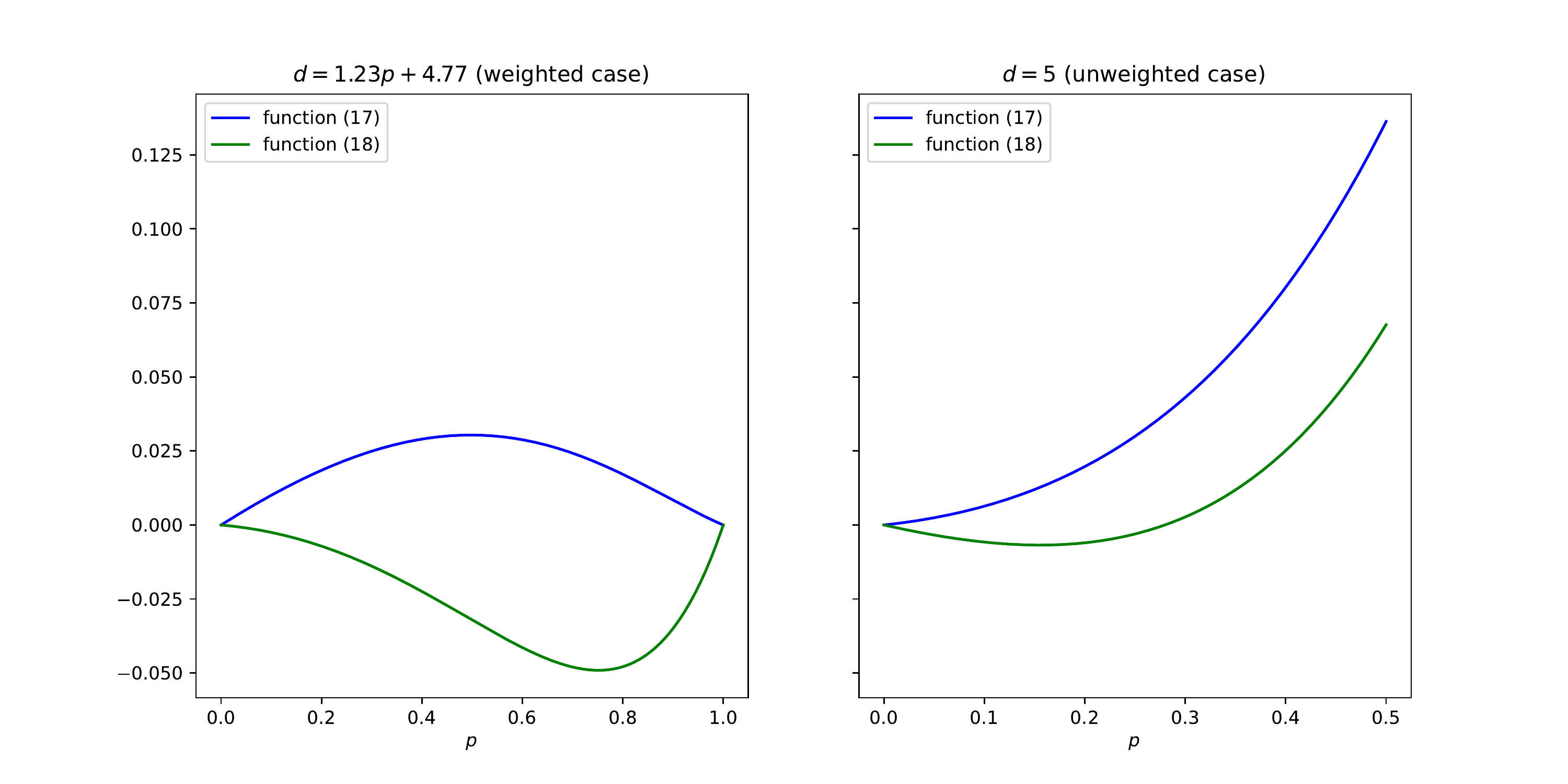}
    \caption{Plot of functions~(\ref{eq:good-d-1}) and~(\ref{eq:good-d-2}) for 
    different values of $d$.}
    \label{fig:good-d}    
  \end{figure}  

  \begin{lemma}\label{lem:good-d}
    Let $f_1 : (0, 1) \to \R$ be defined $f_1(p) = 1.23p + 4.77$ and let $f_2 : (0,\frac{1}{4})
    \to \R$ be defined $f_2(p) = 5$. 
    For $i \in \{1,2\}$ and all $p$ in the domain of $f_i$, we have
    \[
      3^p + 3(f_i(p)+1)^p > 4f_i(p)^p
      \quad\text{and}\quad
      4(f_i(p)+2)^p + 3^p < 5(f_i(p)+1)^p.
    \]
  \end{lemma}

  We need one more lemma regarding the general form of the density when solutions
  in the reduction do not take all of $A$.
  
  \begin{lemma}\label{lem:strict-subset-A}
    Let $G = (L\cup A, E)$ be the graph of the reduction for either the weighted or
    unweighted case. For the weighted case, assume $p \in (0,1)$ and $d = 1.23p+4.77$.
    For the unweighted case, assume $p \in (0, \frac{1}{4})$ and $d = 5$.
    
    Fix $S \subseteq L$ and $A' \subseteq A$. Then 
    \[
      \rho_p(S \cup A')
      =
      \frac{f_p(S \cup A')}{\abs{S \cup A'}} 
      \le
      \frac{3^p \cdot \abs{S} + \sum_{v \in A'} (d + d_{S+v}(v))^p}{\abs{S} + \abs{A'}}         
    \]
    If $\abs{A'} < \abs{A}$, the inequality above is strict.
  \end{lemma}
  \begin{proof}
    As the degree of each vertex in $S$ is exactly $3$ in both the weighted and 
    unweighted cases, 
    \[
      \rho_p(S\cup A')
      =
      \frac{f_p(S \cup A')}{\abs{S \cup A'}}
      =
      \frac{3^p \cdot \abs{S} + \sum_{v \in A'} d_{S\cup A'}(v)^p}{\abs{S} + \abs{A'}}.
    \]
    
    In the weighted case, for all $v \in A'$, 
    $d_{S\cup A'}(v)^p = \left(d\cdot \frac{\abs{A'}-1}{\abs{A}-1} + d_{S+v}(v)\right)^p$.
    As $x^p$ is a strictly increasing function in $x$, we have 
    $d_{S\cup A'}(v)^p \le (d + d_{S+v}(v))^p$. If $\abs{A'} < \abs{A}$, then the inequality
    is strict. In the unweighted case, as 
    $G[A]$ is $d$-regular, then $d_{S\cup A'}(v)^p \le (d + d_{S+v}(v))^p$ for all
    $v \in A'$. If $\abs{A'} < \abs{A}$, as $G[A]$ is connected, 
    there must exist a vertex $u \in A'$
    such that $d_{A'}(u) \le 4 < d$. Thus, $d_{S\cup A'}(u)^p < (d + d_{S+u}(u))^p$.
    Therefore, in both the weighted and unweighted cases, we have 
    $\sum_{v \in A'} d_{S\cup A'}(v)^p \le \sum_{v \in A'}(d + d_{S+v}(v))^p$, 
    where the inequality is strict if $\abs{A'} < \abs{A}$. This concludes the proof.
  \end{proof}    

  We are now ready to prove Theorem~\ref{thm:hard-weighted}.
  
  \begin{proof}[Proof of Theorem~\ref{thm:hard-weighted}]
    Consider an instance of the \ec problem: let $\cS = \{S_1,\ldots,S_m\}$ 
    be a family of subsets of the ground set $\cU = \{e_1,e_2,\ldots, e_{3 n}\}$.
    We use the reductions given in Section~\ref{sec:hardness} for the weighted
    and unweighted cases.
    Recall that both reductions construct the graph $G = (L \cup A, E)$ and
    return TRUE iff $\OPT_G \ge \redopt$ where 
    $\redopt = \frac{3^p + 3(d+1)^p}{4}$.  Note $d = 1.23p + 4.77$ for the weighted
    case and $d = 5$ for the unweighted case.
      
    If $\cS$ contains an exact $3$-cover, we showed in 
    Equation~(\ref{eq:hardness-easy-dir}) this implies $\OPT_G \ge \redopt$
    for both the weighted and unweighted cases.

    Assume $\cS$ does not contain an exact $3$-cover. Let
    $S \subseteq L$ and $A' \subseteq A$. 
    We proceed via cases based on the value of the ratio $\frac{\abs{S}}{\abs{A'}}$.
      
    \textit{Case 1: $3\abs{S} = \abs{A'}$.}
    We have
    \begin{align}
      \frac{f_p(S \cup A')}{\abs{S \cup A'}}
      &\le
      \frac{3^p \cdot \abs{S} + \sum_{v \in A'} (d + d_{S+v}(v))^p}{\abs{S} + \abs{A'}}
        \label{eq:first-d}\\
      &\le
      \frac{3^p \cdot \abs{S} + \abs{A'} \cdot \left(d + 
        \frac{3\cdot \abs{S}}{\abs{A'}}\right)^p}{\abs{S} + \abs{A'}}
        \label{eq:third-d}\\\label{eq:fourth-d}
      &=
      \redopt.
    \end{align}
    where (\ref{eq:first-d}) holds by Lemma~\ref{lem:strict-subset-A}, (\ref{eq:third-d}) holds 
    as $\sum_{i=1}^n (c+x_i)^p$ is concave in $x$ and therefore is 
    maximized when all terms are the same, and
    (\ref{eq:fourth-d}) follows from $\abs{S} = \frac{1}{3} \cdot \abs{A'}$.
      
    As $S$ does not correspond to an exact cover, either 
    $\abs{A'} < \abs{A}$ or there exists $u,v \in A$ such that 
    $d_{S\cup A}(u) \ne d_{S\cup A}(v)$. In the former case, 
    Inequality~(\ref{eq:first-d}) is strict by Lemma~\ref{lem:strict-subset-A}. In the latter case, 
    Inequality~(\ref{eq:third-d}) is strict. So in either case, we have that the 
    density of $S \cup A'$ is strictly smaller than $\redopt$.

    \textit{Case 2: $3\cdot\abs{S} = \beta \abs{A'}$ for $\beta \in [0,1)$.}
    By Lemma~\ref{lem:strict-subset-A}, $\rho_p(S\cup A')
    \le \frac{3^p \cdot \abs{S} + \sum_{v \in A'} (d + d_{S+v}(v))^p}{\abs{S} + \abs{A'}}$.
    Thus, by Lemma~\ref{lem:int-flat}, the density $\rho_p(S \cup A')$ is maximized when
    the vertices in $A'$ have degree either $d$ or
    $d+ 1$. We have
    \begin{align}
      \rho_p(S\cup A')
      &\le \label{eq:case2-1}
      \frac{3^p \cdot \abs{S} + \beta\cdot\abs{A'} (d+1)^p + (1-\beta)\abs{A'} d^p}
        {\abs{S}+\abs{A'}}\\
      &=\label{eq:case2-2}
      \frac{3^p\beta + 3\beta (d+1)^p + 3(1-\beta) d^p}{\beta + 3},
    \end{align}
    where the equality uses the fact that $\abs{S} = \frac{\beta \abs{A'}}{3}$. By
    differentiating~(\ref{eq:case2-2}) with respect to $\beta$, it is easy to show that this
    function is strictly increasing in $\beta$ if and only if $3^p + 3(d+1)^p > 4d^p$.
    (We give a proof in Lemma~\ref{lem:strict-inc}.)
    By the choice of $d = 1.23p + 4.77$ for the weighted case and $d = 5$ for the 
    unweighted case, Lemma~\ref{lem:good-d} shows 
    the inequality is satisfied. Thus, the density $\rho_p(S\cup A')$ is strictly less than the value 
    of the function in~(\ref{eq:case2-2}) when $\beta = 1$,
    which is exactly $\redopt$.
      
    \textit{Case 3: $3\cdot\abs{S} = \alpha \abs{A'}$ for $\alpha > 1$.}
    We reparameterize such that $\alpha = t + \beta$ where $t\ge 1$ is an integer
    and $\beta \in [0,1]$. Note that we can assume that $(t, \beta) \ne (1, 0)$ as this
    was already handled in Case 1. By Lemma~\ref{lem:strict-subset-A}, $\rho_p(S\cup A')
    \le \frac{3^p \cdot \abs{S} + \sum_{v \in A'} (d + d_{S+v}(v))^p}{\abs{S} + \abs{A'}}$.
    By Lemma~\ref{lem:int-flat}, the density 
    $\rho_p(S \cup A')$ is maximized when the vertices in $A'$ are either 
    $d + t$ or $d+ t+1$. We have
    \begin{align}
      \rho_p(S\cup A')
      &\le \label{eq:case3-1}
      \frac{3^p \cdot \abs{S} + \beta\cdot\abs{A'} (d+t+1)^p + (1-\beta)\abs{A'} (d+t)^p}
        {\abs{S}+\abs{A'}}\\
      &= \label{eq:case3-2}
      \frac{3^p (t+\beta) + 3\beta(d+t+1)^p + 3(1-\beta) (d+t)^p}{t +\beta + 3},
    \end{align}
    where the equality uses the fact that $\abs{S} = \frac{(t+\beta)\abs{A'}}{3}$. 
    By differentiating~(\ref{eq:case3-2}) with respect to $\beta$, it is easy to show that this
    function is strictly decreasing in $\beta$ if
    $4(d+2)^p + 3^p < 5(d+1)^p$. (We give a proof in 
    Lemma~\ref{lem:strict-dec}.) By the choice of 
    $d = 1.23p + 4.77$, Lemma~\ref{lem:good-d} shows the inequality is satisfied. 
    Therefore,~(\ref{eq:case3-2}) is uniquely maximized at $(t,\beta) = (1,0)$, which has
    value $\redopt$. This implies $\rho_p(S \cup A') < \redopt$.  
  \end{proof}
  
  \subsection{Hardness for $p \in (-3,0)$}
  \label{sec:hardness-neg-app}
    As $p < 0$, we have $\max_S M_p(S)$ is equivalent to the problem of 
    $\min_S \rho_p(S)$ since 
    $x^p$ is a decreasing function in $x$. We therefore focus on the problem of
    $\min_S \rho_p(S)$.\footnote{Note that, although it is not explicitly stated, 
    we are optimizing over all sets $S$ with no vertices of degree $0$, for 
    such sets have undefined density when $p < 0$.} 
    
    The reductions for $p < 0$ are nearly identical to that of the case of $p >0$. 
    For the reduction for the weighted case, we change the value of 
    $d$ from $1.23p + 4.77$ to $p/2 + 5$. The value of $d$ stays the same for 
    the unweighted case. For both the weighted and unweighted
    case, we change the reduction by returning TRUE iff $\OPT_G \le \rho^*$.
    These changes only minimally alter the analyses. The same reasoning
    in~(\ref{eq:hardness-easy-dir}) holds, so when the given instance contains an 
    exact $3$-cover, we have $\OPT_G \le \rho^*$. Now suppose the given
    instance does not contain an exact $3$-cover.
    As $p < 0$, we have $(c + x)^p$ is a convex function in $x$. 
    The same general outline for the proof given in~(\ref{eq:outline-1})-(\ref{eq:outline-2}) 
    therefore holds for $p < 0$. In particular, for $S \subseteq L$ and assuming
    $3\cdot\abs{S} = \alpha \cdot \abs{A}$, 
    \begin{align*}
      \rho_p(S \cup A)
      &=
      \frac{3^p \cdot \abs{S} + \sum_{v \in A} d_{S\cup A}(v)^p}{\abs{S} + \abs{A}}\\
      &\ge
      \frac{3^p \cdot \abs{S} + \abs{A}\cdot (d + \frac{3 \abs{S}}{\abs{A}})^p}
        {\abs{S} + \abs{A}}\\
      &=
      \frac{3^p \cdot \alpha + 3\cdot (d + \alpha)^p}{\alpha +3},
    \end{align*}
    where the inequality now uses convexity instead of concavity. Our goal is now
    to choose $d$ such that this function is \emph{minimized} when $\alpha =1$.
    As is the case with $p \in (0,1)$, we use a stronger bound on 
    $\sum_{v \in A} d_{S\cup A}(v)^p$. 
    In particular, we note Lemma~\ref{lem:int-flat} applies also when $p < 0$. 
    With this lemma in hand and using this general outline, we can
    prove the following theorem.
    
    \begin{theorem}\label{thm:hard-weighted-neg}
      \pdsg is NP-hard for $p \in (-\frac{1}{8}, 0)$ and
      weighted \pdsg is NP-hard for $p \in (-3,0)$.    
    \end{theorem}
    
    The proof of Theorem~\ref{thm:hard-weighted-neg} is nearly identical to that of
    Theorem~\ref{thm:hard-weighted} except for a few small differences, which is why we
    do not rewrite all of the details. We point out these small differences here. We first note that 
    we need an analogous lemma to 
    Lemma~\ref{lem:good-d} for the case of $p < 0$. 
    
    \begin{lemma}\label{lem:good-d-neg}
      Let $f_1 : (-3, 0) \to \R$ be defined $f_1(p) = \frac{p}{2} + 5$ and let 
      $f_2 : (-\frac{1}{8}, 0)\to \R$ be defined $f_2(p) = 5$.
      For $i \in \{1,2\}$ and all $p$ in the domain of $f_i$, we have
      \[
        3^p + 3(f_i(p)+1)^p < 4f_i(p)^p
        \quad\text{and}\quad
       4(f_i(p)+2)^p + 3^p > 5(f_i(p)+1)^p.
      \]
    \end{lemma}
    We consider the following functions:
    \begin{equation}\label{eq:good-d-1-neg}
      3^p + 3(d+1)^p - 4d^p
    \end{equation}
    and 
    \begin{equation}\label{eq:good-d-2-neg}
      4(d+2)^p + 3^p - 5(d+1)^p,
    \end{equation}

    where $d=\frac{p}{2}+5$ for the weighted case and $d = 5$ for the unweighted case.
    To prove Lemma~\ref{lem:good-d-neg},
    we need to show that~(\ref{eq:good-d-1-neg}) is negative 
    and~(\ref{eq:good-d-2-neg}) is positive for both cases.
    For the weighted case,
    in one approach to prove~(\ref{eq:good-d-2-neg}) is positive, we could
    first show the derivative only has one zero on the interval $[-3,0]$ and that
    this critical point is a maximum. Since~(\ref{eq:good-d-2-neg}) is non-negative
    at $p=-3$ and $p=0$, this would prove~(\ref{eq:good-d-2-neg}) is positive.
    A more straightforward analysis for showing~(\ref{eq:good-d-1-neg}) is 
    negative is to argue that the derivative of~(\ref{eq:good-d-1-neg}) is positive
    over the interval $[-3,0]$, and therefore~(\ref{eq:good-d-1-neg}) is strictly 
    increasing over the interval. Since~(\ref{eq:good-d-1-neg}) is $0$ at $p = 0$, 
    this would prove~(\ref{eq:good-d-1-neg}) is negative over the interval.
    For the unweighted case, we could show~(\ref{eq:good-d-2-neg}) is negative
    by showing that the single critical point on the interval $(-\frac{1}{8}, 0)$ is a
    minimum and the function at the endpoints of the interval is non-positive. To
    show~(\ref{eq:good-d-1-neg}) is positive, one could simply show the function
    is strictly decreasing on the interval $(-\frac{1}{8},0)$ and the function is non-negative
    at $0$. 
    
    As was the case for Lemma~\ref{lem:good-d}, the formal verification of these
    facts would be lengthy and not necessarily useful, so we again verified these
    inequalities analytically via MATLAB. We plot these functions
    over the respective domains in Figure~\ref{fig:good-d-neg}.
    
    \begin{figure}[t]
      \centering
      \includegraphics[scale=0.5]{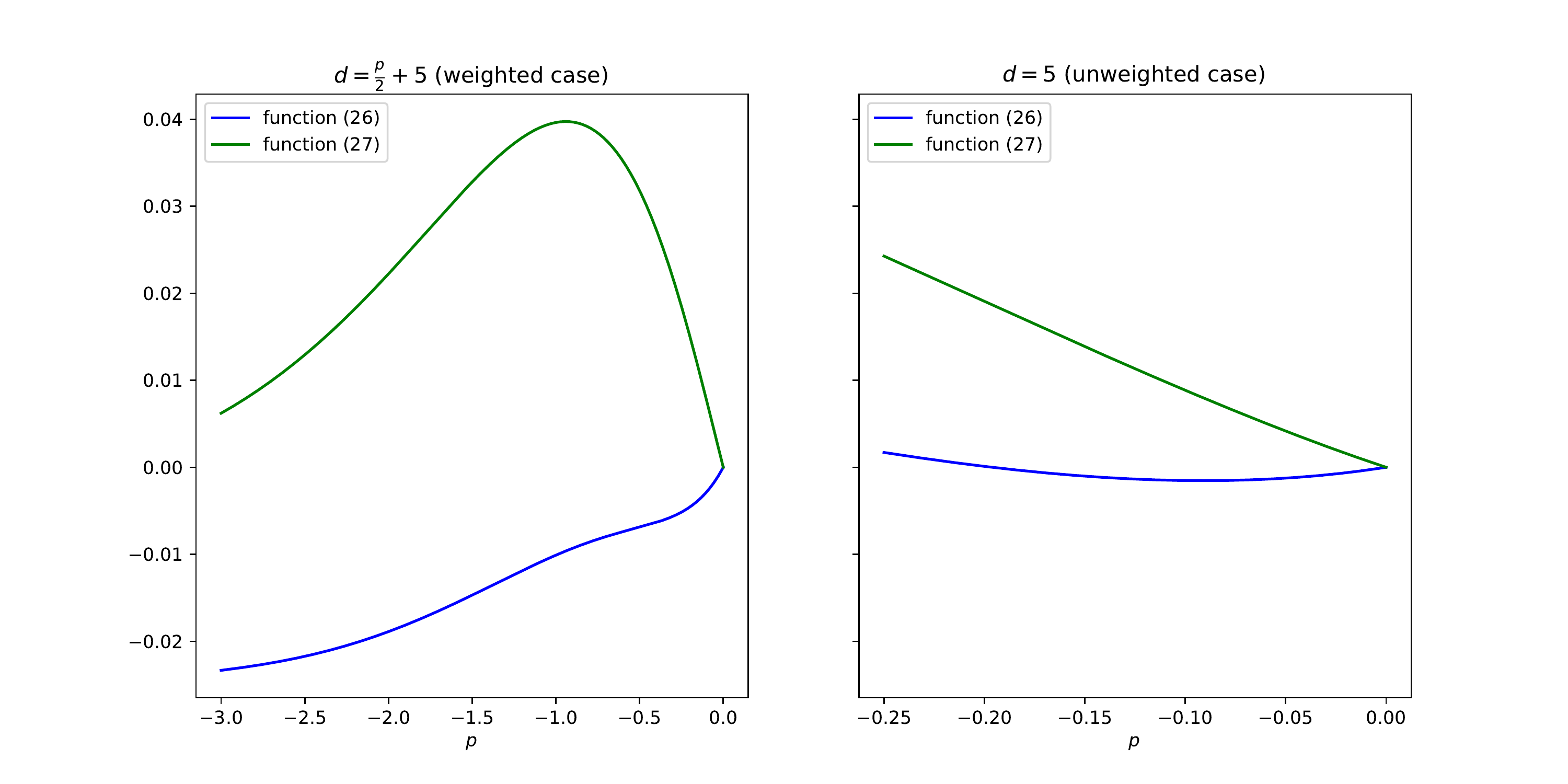}
      \caption{Plot of functions~(\ref{eq:good-d-1-neg}) and~(\ref{eq:good-d-2-neg}).}
      \label{fig:good-d-neg}    
    \end{figure}  
    
    We also need an analogous lemma to Lemma~\ref{lem:strict-subset-A} for the 
    case of $p < 0$. It is easy to argue that since $x^p$ is now a strictly decreasing 
    function in $x$, $\rho_p(S \cup A') \ge
    \frac{3^p \cdot \abs{S} + \sum_{v \in A'} (d + d_{S+v}(v))^p}{\abs{S} + \abs{A'}}$.
    The inequality is strict when $\abs{A'} < \abs{A}$.

    \paragraph{Proof sketch for Theorem~\ref{thm:hard-weighted-neg}.}
    We are now ready to discuss the few differences in the proof of
    Theorem~\ref{thm:hard-weighted-neg} compared to the proof of
    Theorem~\ref{thm:hard-weighted}. At the beginning of this section, we discussed the
    changes to the reduction for the case of $p < 0$. We also noted that a similar argument
    to that in Equation~\ref{eq:hardness-easy-dir} would hold for $p < 0$, showing that
    $\OPT_G \le \rho^*$ when the given instance $(\cS, \cU)$ contains an exact $3$-cover.
    
    Now suppose the \ec instance does not contain an exact $3$-cover. We discuss how the
    three cases in the proof of Theorem~\ref{thm:hard-weighted} would change for $p < 0$.
    Recall that we take arbitrary sets $S \subseteq L$ and $A' \subseteq A$ and the goal
    is to show that $\rho_p(S \cup A') > \rho^*$. For case 1 
    (corresponding to $3\abs{S} = \abs{A'}$),
    the only difference is that the function $\sum_{i=1}^n (c+x_i)^p$ is convex instead of
    concave. The reasoning for this case still remains the same. Now we consider case 2
    (corresponding to $3\abs{S} = \beta \abs{A'}$ for some $\beta \in [0,1)$). There are two
    differences here. We first see that Lemma~\ref{lem:int-flat} holds for $p < 0$, implying
    that $\rho_p(S \cup A')$ is \emph{at least} the quantity in~(\ref{eq:case2-1}). The 
    second difference is in guaranteeing that the function in~(\ref{eq:case2-2}) is 
    strictly decreasing in $\beta$ if and only if $3^p + 3(d+1)^p < 4d^p$. This 
    inequality holds because of Lemma~\ref{lem:good-d-neg} 
    (instead of Lemma~\ref{lem:good-d}).
    
    Now we consider case 3 (corresponding to $3\abs{S} = \alpha \abs{A'}$ for $\alpha > 1$).
    There are again two differences. Lemma~\ref{lem:int-flat} holds for $p < 0$, so we have that
    $\rho_p(S\cup A')$ is \emph{at least} the quantity in~(\ref{eq:case3-1}). Then we have
    that the function in~(\ref{eq:case3-2}) is strictly increasing in $\beta$ if 
    $4(d+2)^p + 3^p > 5(d+1)^p$.  This holds because of Lemma~\ref{lem:good-d-neg}.
   
  \subsection{Discussion on extending hardness results to all $p \in (-\infty, 0) \cup (0,1)$}
  \label{sec:hardness-all-neg-app}
    \paragraph{Weighted graphs.}
    Suppose one wants to prove that the weighted \pdsg is NP-hard for some
    $p \le -3$. All one would have to do is define $d$ in terms of $p$ such that
    the inequalities in Lemma~\ref{lem:good-d-neg} hold. For example, take 
    $p = -10$. It would suffice to set $d = 0.2p + 5$ in order to prove the inequalities
    in Lemma~\ref{lem:good-d-neg} hold for $p \in [-10 - \delta, -10 + \delta]$ for some small
    $\delta$. (It is important to note that this choice of $d$ would not work for 
    $p \in [-3, -0.5]$, for example.)
    Once we have a definition of $d$ that allows us to prove a version of
    Lemma~\ref{lem:good-d-neg} for our new choice of $p$, the proof of 
    Theorem~\ref{thm:hard-weighted-neg} will still work with this new version of 
    Lemma~\ref{lem:good-d-neg}.    
    
    \paragraph{Unweighted graphs.}
    In the unweighted case, we have less flexibility than we do for the weighted case.
    In this case, one potential approach to proving NP-hardness for values of $p$ in
    $(-\infty, -\frac{1}{8}) \cup (\frac{1}{4}, 1)$
    is to reduce from \aec{\ell} for some $\ell > 3$ instead of \ec and carefully choose
    the value of $d$, the degree of the regular graph $G[A]$. For example, the analysis
    goes through as above for $p = \frac{1}{2}$ if one reduces from \aec{6} and sets $d =11$.
    The difficulty in this approach is being able to choose a value for $d$ and $\ell$ for 
    every given $p$.  
  
\section{Experiments}
\label{sec:experiments-app}
  \subsection{Background on Frank-Wolfe}
    The Frank-Wolfe method is used to solve constrained convex optimization problems.
    Suppose we want to minimize a differentiable, convex function $f$ over a convex set $S$. 
    The Frank-Wolfe method critically relies on the ability to efficiently optimize
    linear functions over $S$. The algorithm proceeds as follows. We start with an
    arbitrary point $x_0 \in S$. We start iteration $k$ by first solving 
    $y_k := \argmin_{z \in S} z^T \nabla f(x_k)$. Then we let
    $x_{k+1} = (1 - \alpha_k) x_k + \alpha_k y_k$ where $\alpha_k$ is the step size.
    The standard Frank-Wolfe method sets $\alpha_k = \frac{2}{k+2}$. 

    We now discuss how we use the Frank-Wolfe method to solve \pdsg. 
    We use the idea from~\cite{hqc-22} where they take the same approach for
    solving \apdsg{1}. Harb et al.\ introduce the following convex program for 
    a supermodular function $f$:
    \begin{equation}\label{pr:dsg-sq}
      \text{minimize } \sum_{v \in V} b_v^2 \text{ subject to } b \in B_{f}
    \end{equation}
    where 
    \[
      B_{f} = \{x \in \R^V \mid x \ge 0; x(S) \ge f(S), \forall S \subseteq V; x(V) = f(V)\}.
    \]
    Note that $B_f$ is the base contrapolymatroid associated
    with $f$ (see, e.g.,~\cite{s-03}). \cite{hqc-22} shows that if one could obtain an
    optimal solution $b$ to~(\ref{pr:dsg-sq}), then the vertices with the largest value
    in $b$ will form an optimal solution to \dss (i.e.\ $\max_S \frac{f(S)}{\abs{S}}$).
    We can then use the Frank-Wolfe method to obtain near-optimal solutions
    to~(\ref{pr:dsg-sq}). For rounding a near-optimal
    solution $b$ to~(\ref{pr:dsg-sq}), we use the same heuristic approach 
    that~\cite{hqc-22} and~\cite{dcs-17} used for \apdsg{1},
    which is to sort the vertices in increasing order of $b$ and output the suffix with the
    largest density. 
    
    We address a couple of implementation aspects of using the Frank-Wolfe
    method to find solutions to~(\ref{pr:dsg-sq}). We use Frank-Wolfe to optimize over the 
    base contrapolymatroid associated with $f_p$, $B_{f_p}$.
    The Frank-Wolfe method requires an initial point in the given convex set $B_{f_p}$. In our
    experiments, we start with a simple point in $B_{f_p}$: for all $v \in V$, we set 
    $x_v = d_G(v)^p$. We experimented with other starting points, such as those based
    on the ordering produced by $\greedy$, but this simple solution worked best and was
    much faster.
    
    Furthermore, we note that when applying Frank-Wolfe to solve the convex
    program~(\ref{pr:dsg-sq}), the optimization problem at each iteration is essentially
    $\min_{z \in B_f} z^T x_k$ where $x_k$ is our current solution. We can easily solve this
    optimization problem as algorithms for optimizing linear functions over the base 
    contrapolymatroid are well-understood~\cite{l-83}.
    
  \subsection{Full reporting of results}
    We include the full table of results in Table~\ref{fig:exp-lazy-full} for the experiments 
    comparing $\greedy$ and $\lgreedy$ on all ten real-world graphs and all values of 
    $p \in \{1.05, 1.25, 1.5, 1.75, 2\}$. For $\lgreedy$, we report results
    for values of $\eps \in \{0.01, 0.1, 1\}$.
    
    We include results for the experiments comparing $\gpp$, $\lgpp$, and Frank-Wolfe
    on all ten real-world graphs and all values of $p \in \{1.05, 1.25, 1.5, 1.75, 2\}$.
    We include separate figures for each value of $p$ in
    Figures~\ref{fig:exp-near-opt-1.05}, \ref{fig:exp-near-opt-1.25},
    \ref{fig:exp-near-opt-1.5}, \ref{fig:exp-near-opt-1.75}, \ref{fig:exp-near-opt-2}.

    We include the full table of results in Table~\ref{fig:exp-lt1-full} for the experiments 
    for \pdsg for $p < 1$. We consider $p \in \{-1, -0.5, 0.25, 0.5, 0.75\}$.
  
    \begin{table*}
    \resizebox{\textwidth}{!}{
    \begin{tabular}{c|c|l|llllllllll}
      \hline
      & metric & algorithm  & Astro & CM05 & BrKite & Enron & roadCA & roadTX & webG & webBS & Amaz & YTube\\\hline
      \multirow{8}{*}{$p=1.05$} &      \multirow{4}{*}{time (s)} & $\abblgreedy$ ($\eps$$=$$0.01$) & 0.272 & 0.173 & 0.339 & 0.325 & \textbf{3.142} & \textbf{2.222} & 23.25 & 32.68 & 1.36 & 15.19 \\
      &        & $\abblgreedy$ ($\eps$$=$$0.1$) & 0.138 & 0.134 & 0.199 & 0.147 & 3.145 & 2.224 & 12.58 & 13.13 & 1.234 & 5.917 \\
      &        & $\abblgreedy$ ($\eps$$=$$1.0$) & \textbf{0.103} & \textbf{0.107} & \textbf{0.148} & \textbf{0.109} & 3.191 & 2.227 & \textbf{9.521} & \textbf{10.07} & \textbf{1.041} & \textbf{3.886} \\
      &        & $\greedy$ & 0.312 & 0.175 & 0.421 & 0.588 & 3.364 & 2.432 & 40.31 & 686.5 & 1.36 & 73.33 \\
      \cline{2-13}
      &       \multirow{4}{*}{\begin{tabular}{c}density \\ ($M_p$) \end{tabular}} & $\abblgreedy$ ($\eps$$=$$0.01$) & 59.49 & \textbf{31.74} & \textbf{81.42} & \textbf{75.16} & \textbf{3.751} & \textbf{4.164} & \textbf{54.39} & \textbf{206.9} & \textbf{8.867} & \textbf{91.91} \\
      &        & $\abblgreedy$ ($\eps$$=$$0.1$) & 59.49 & \textbf{31.74} & \textbf{81.42} & \textbf{75.16} & \textbf{3.751} & \textbf{4.164} & \textbf{54.39} & \textbf{206.9} & 8.864 & \textbf{91.91} \\
      &        & $\abblgreedy$ ($\eps$$=$$1.0$) & \textbf{59.5} & \textbf{31.74} & \textbf{81.42} & \textbf{75.16} & \textbf{3.751} & 3.915 & \textbf{54.39} & \textbf{206.9} & 8.841 & \textbf{91.91} \\
      &        & $\greedy$ & 59.49 & \textbf{31.74} & \textbf{81.42} & \textbf{75.16} & \textbf{3.751} & 4.007 & \textbf{54.39} & \textbf{206.9} & \textbf{8.867} & \textbf{91.91} \\
        \hline
      \multirow{8}{*}{$p=1.25$} &      \multirow{4}{*}{time (s)} & $\abblgreedy$ ($\eps$$=$$0.01$) & 0.285 & 0.176 & 0.349 & 0.35 & \textbf{3.145} & \textbf{2.213} & 24.88 & 35.29 & 1.565 & 16.53 \\
      &        & $\abblgreedy$ ($\eps$$=$$0.1$) & 0.144 & 0.135 & 0.209 & 0.154 & 3.148 & 2.214 & 13.09 & 13.32 & 1.342 & 7.685 \\
      &        & $\abblgreedy$ ($\eps$$=$$1.0$) & \textbf{0.105} & \textbf{0.114} & \textbf{0.151} & \textbf{0.113} & 3.196 & 2.234 & \textbf{9.714} & \textbf{9.952} & \textbf{1.052} & \textbf{4.623} \\
      &        & $\greedy$ & 0.312 & 0.174 & 0.424 & 0.581 & 3.146 & 2.284 & 41.85 & 687.2 & 1.38 & 75.98 \\
      \cline{2-13}
      &       \multirow{4}{*}{\begin{tabular}{c}density \\ ($M_p$) \end{tabular}} & $\abblgreedy$ ($\eps$$=$$0.01$) & 61.6 & 32.7 & \textbf{82.7} & \textbf{77.21} & \textbf{3.756} & \textbf{4.054} & \textbf{54.63} & \textbf{207.4} & 11.46 & \textbf{95.5} \\
      &        & $\abblgreedy$ ($\eps$$=$$0.1$) & 61.6 & 32.71 & \textbf{82.7} & \textbf{77.21} & \textbf{3.756} & \textbf{4.054} & \textbf{54.63} & \textbf{207.4} & 11.46 & \textbf{95.5} \\
      &        & $\abblgreedy$ ($\eps$$=$$1.0$) & \textbf{61.76} & \textbf{32.74} & \textbf{82.7} & \textbf{77.21} & \textbf{3.756} & 3.958 & \textbf{54.63} & \textbf{207.4} & \textbf{11.48} & \textbf{95.5} \\
      &        & $\greedy$ & 61.6 & 32.7 & \textbf{82.7} & \textbf{77.21} & 3.683 & 4.036 & \textbf{54.63} & \textbf{207.4} & 11.46 & \textbf{95.5} \\
        \hline
      \multirow{8}{*}{$p=1.5$} &      \multirow{4}{*}{time (s)} & $\abblgreedy$ ($\eps$$=$$0.01$) & 0.294 & 0.174 & 0.36 & 0.329 & 3.151 & 2.215 & 23.99 & 37.91 & 1.365 & 17.34 \\
      &        & $\abblgreedy$ ($\eps$$=$$0.1$) & 0.14 & 0.138 & 0.216 & 0.161 & 3.145 & 2.215 & 12.99 & 13.56 & 1.255 & 6.733 \\
      &        & $\abblgreedy$ ($\eps$$=$$1.0$) & \textbf{0.106} & \textbf{0.109} & \textbf{0.156} & \textbf{0.113} & \textbf{3.107} & \textbf{2.172} & \textbf{9.553} & \textbf{9.887} & \textbf{1.052} & \textbf{4.051} \\
      &        & $\greedy$ & 0.314 & 0.175 & 0.42 & 0.582 & 3.141 & 2.657 & 36.11 & 752.4 & 1.505 & 60.99 \\
      \cline{2-13}
      &       \multirow{4}{*}{\begin{tabular}{c}density \\ ($M_p$) \end{tabular}} & $\abblgreedy$ ($\eps$$=$$0.01$) & 64.24 & \textbf{33.94} & \textbf{84.46} & \textbf{80.31} & 3.72 & \textbf{4.074} & 54.86 & \textbf{244.2} & 13.81 & \textbf{101.7} \\
      &        & $\abblgreedy$ ($\eps$$=$$0.1$) & 64.32 & \textbf{33.94} & \textbf{84.46} & \textbf{80.31} & 3.72 & \textbf{4.074} & 54.86 & \textbf{244.2} & 13.81 & \textbf{101.7} \\
      &        & $\abblgreedy$ ($\eps$$=$$1.0$) & \textbf{64.57} & \textbf{33.94} & 84.45 & \textbf{80.31} & \textbf{3.763} & 4.014 & \textbf{54.91} & \textbf{244.2} & \textbf{13.87} & \textbf{101.7} \\
      &        & $\greedy$ & 64.24 & \textbf{33.94} & \textbf{84.46} & \textbf{80.31} & 3.72 & \textbf{4.074} & 54.86 & \textbf{244.2} & 13.81 & \textbf{101.7} \\
        \hline
      \multirow{8}{*}{$p=1.75$} &      \multirow{4}{*}{time (s)} & $\abblgreedy$ ($\eps$$=$$0.01$) & 0.297 & 0.173 & 0.362 & 0.383 & 3.132 & 2.2 & 20.87 & 41.33 & 1.333 & 17.2 \\
      &        & $\abblgreedy$ ($\eps$$=$$0.1$) & 0.157 & 0.141 & 0.223 & 0.176 & 3.135 & 2.198 & 12.28 & 14.16 & 1.245 & 6.931 \\
      &        & $\abblgreedy$ ($\eps$$=$$1.0$) & \textbf{0.108} & \textbf{0.109} & \textbf{0.154} & \textbf{0.113} & \textbf{3.104} & \textbf{2.163} & \textbf{9.399} & \textbf{9.931} & \textbf{1.049} & \textbf{4.044} \\
      &        & $\greedy$ & 0.31 & 0.172 & 0.406 & 0.586 & 3.219 & 2.541 & 31.26 & 742.4 & 1.406 & 68.5 \\
      \cline{2-13}
      &       \multirow{4}{*}{\begin{tabular}{c}density \\ ($M_p$) \end{tabular}} & $\abblgreedy$ ($\eps$$=$$0.01$) & 67.12 & \textbf{35.38} & \textbf{86.45} & \textbf{84.19} & \textbf{3.816} & \textbf{4.113} & 66.5 & \textbf{387.1} & 15.63 & \textbf{112.2} \\
      &        & $\abblgreedy$ ($\eps$$=$$0.1$) & 67.12 & \textbf{35.38} & \textbf{86.45} & \textbf{84.19} & \textbf{3.816} & \textbf{4.113} & 66.51 & \textbf{387.1} & 15.63 & \textbf{112.2} \\
      &        & $\abblgreedy$ ($\eps$$=$$1.0$) & \textbf{67.88} & 35.37 & 86.44 & \textbf{84.19} & 3.769 & 4.071 & \textbf{67.68} & \textbf{387.1} & \textbf{16.76} & \textbf{112.2} \\
      &        & $\greedy$ & 67.12 & \textbf{35.38} & \textbf{86.45} & \textbf{84.19} & \textbf{3.816} & \textbf{4.113} & 66.5 & \textbf{387.1} & 15.63 & \textbf{112.2} \\
        \hline
      \multirow{8}{*}{$p=2.0$} &      \multirow{4}{*}{time (s)} & $\abblgreedy$ ($\eps$$=$$0.01$) & 0.248 & 0.125 & 0.307 & 0.345 & 2.187 & 1.555 & 18.4 & 41.71 & 0.966 & 16.21 \\
      &        & $\abblgreedy$ ($\eps$$=$$0.1$) & 0.112 & 0.096 & 0.17 & 0.124 & 2.184 & 1.554 & 9.105 & 10.29 & 0.919 & 6.069 \\
      &        & $\abblgreedy$ ($\eps$$=$$1.0$) & \textbf{0.057} & \textbf{0.064} & \textbf{0.101} & \textbf{0.066} & \textbf{2.161} & \textbf{1.528} & \textbf{6.089} & \textbf{5.217} & \textbf{0.748} & \textbf{3.272} \\
      &        & $\greedy$ & 0.257 & 0.124 & 0.337 & 0.518 & \textbf{2.161} & 1.534 & 25.65 & 662.0 & 0.98 & 60.97 \\
      \cline{2-13}
      &       \multirow{4}{*}{\begin{tabular}{c}density \\ ($M_p$) \end{tabular}} & $\abblgreedy$ ($\eps$$=$$0.01$) & 71.46 & \textbf{37.2} & \textbf{88.78} & \textbf{88.99} & 3.759 & \textbf{4.224} & 98.65 & \textbf{675.3} & \textbf{23.5} & \textbf{182.4} \\
      &        & $\abblgreedy$ ($\eps$$=$$0.1$) & 71.44 & \textbf{37.2} & \textbf{88.78} & \textbf{88.99} & 3.759 & \textbf{4.224} & \textbf{99.21} & \textbf{675.3} & \textbf{23.5} & \textbf{182.4} \\
      &        & $\abblgreedy$ ($\eps$$=$$1.0$) & \textbf{71.7} & \textbf{37.2} & 88.77 & 88.97 & \textbf{3.775} & 4.092 & 98.67 & 669.0 & \textbf{23.5} & \textbf{182.4} \\
      &        & $\greedy$ & 71.46 & \textbf{37.2} & \textbf{88.78} & \textbf{88.99} & 3.759 & \textbf{4.224} & 98.65 & \textbf{675.3} & \textbf{23.5} & \textbf{182.4} \\
        \hline
    \end{tabular}
    }
    \caption{Results for experiments comparing $\greedy$ to $\lgreedy$. 
    We write $\lgreedy$ as $\abblgreedy$ for spacing issues.
    See Section~\ref{sec:exp-lazy} for more details.}
    \label{fig:exp-lazy-full}
    \end{table*}

\begin{figure*}
\centering
\includegraphics[scale=0.45]{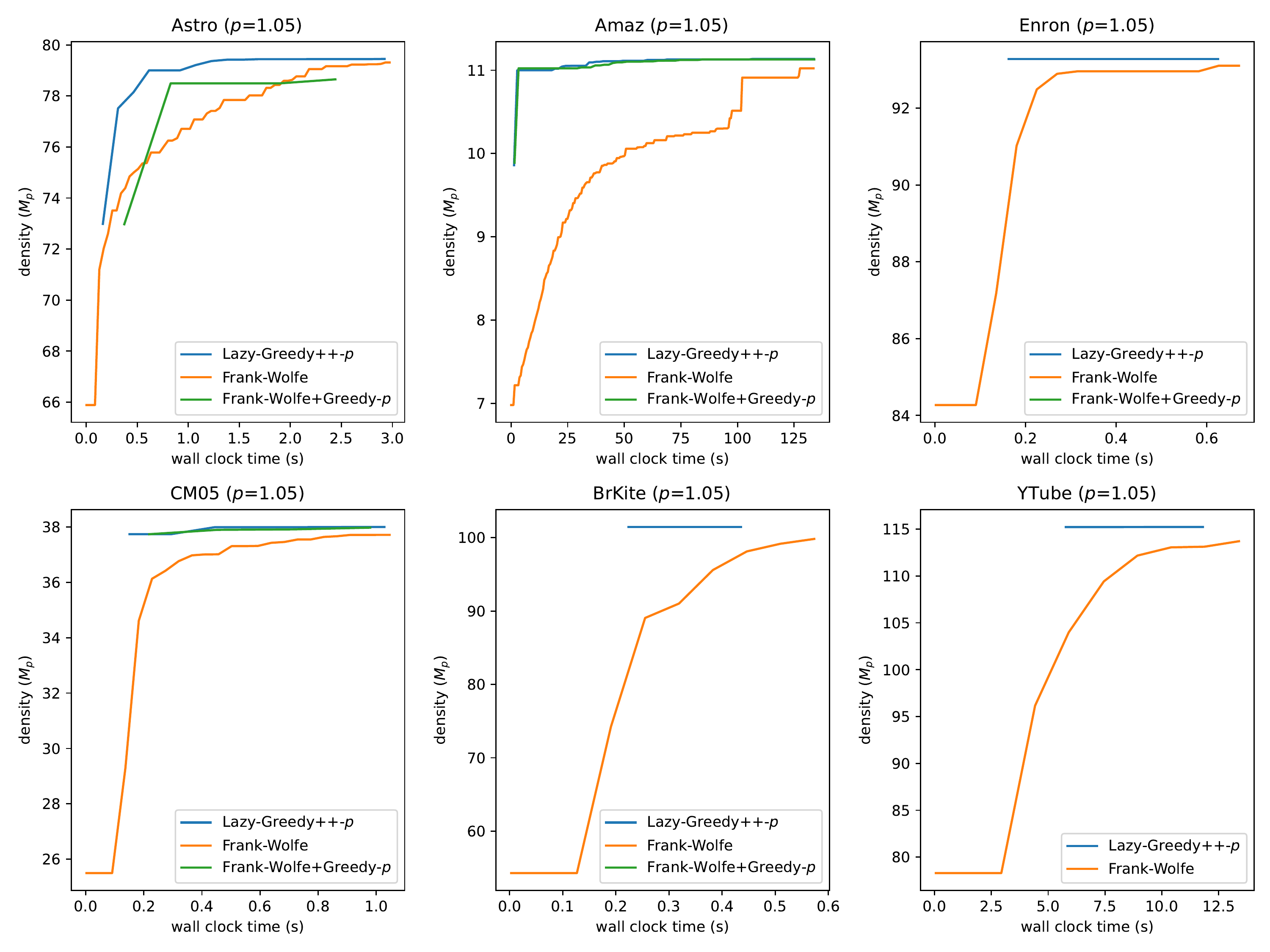}

\bigskip

\includegraphics[scale=0.45]{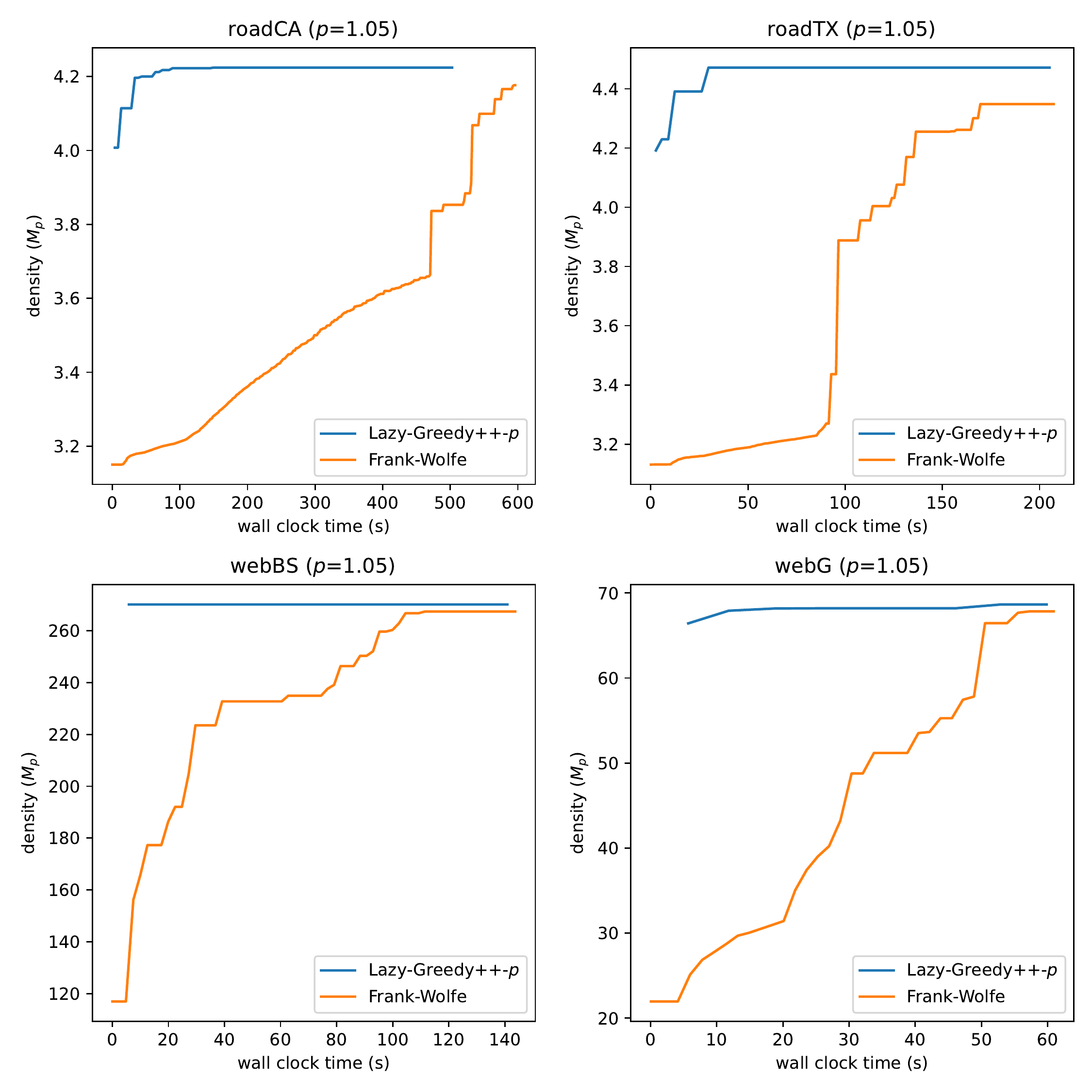}
\caption{Experiments comparing convergence rates of the
algorithms we consider in Section~\ref{sec:exp-near-opt}. Results for $p = 1.05$. }
\label{fig:exp-near-opt-1.05}
\end{figure*}

\begin{figure*}
\centering
\includegraphics[scale=0.45]{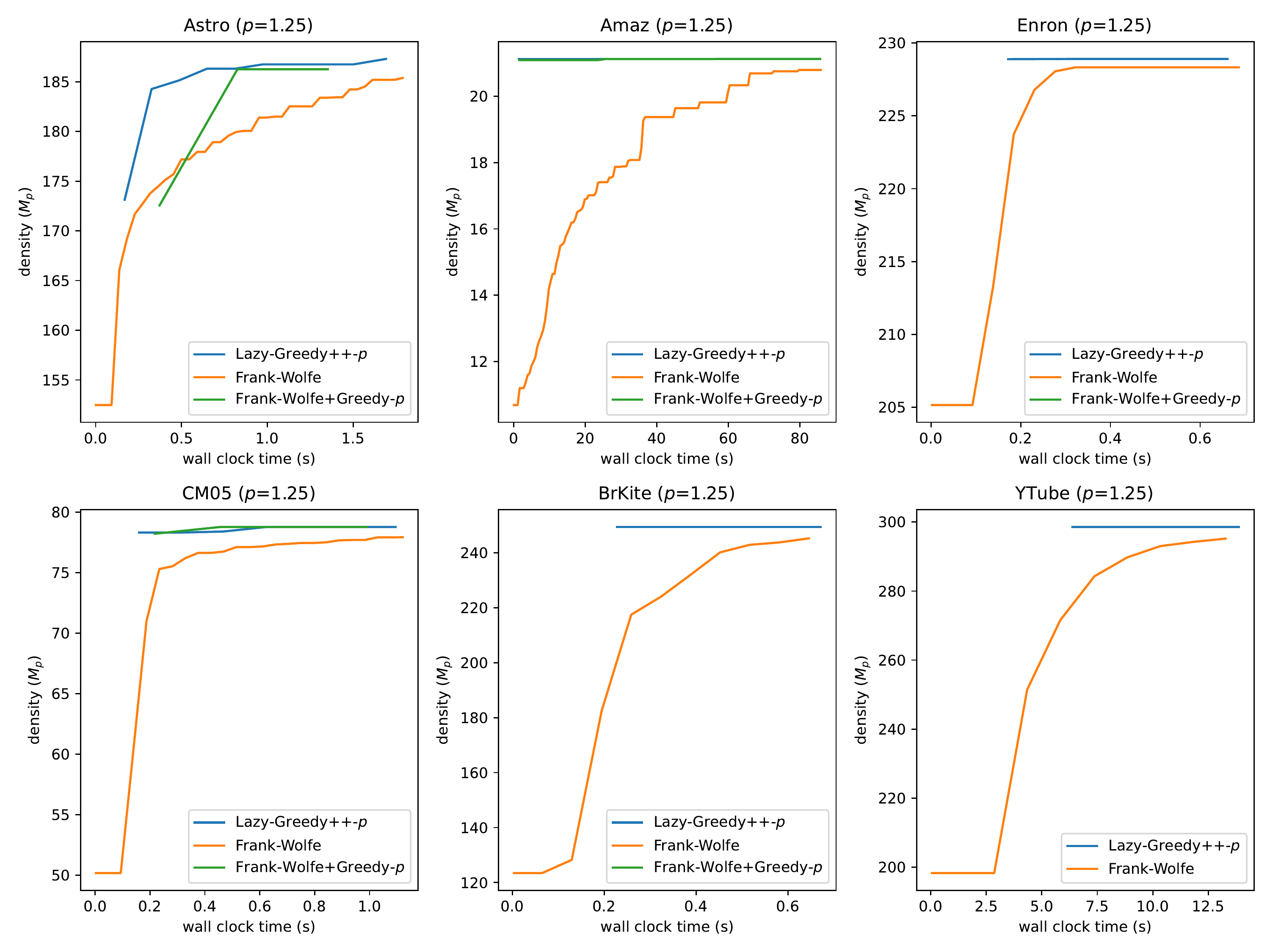}

\bigskip

\includegraphics[scale=0.45]{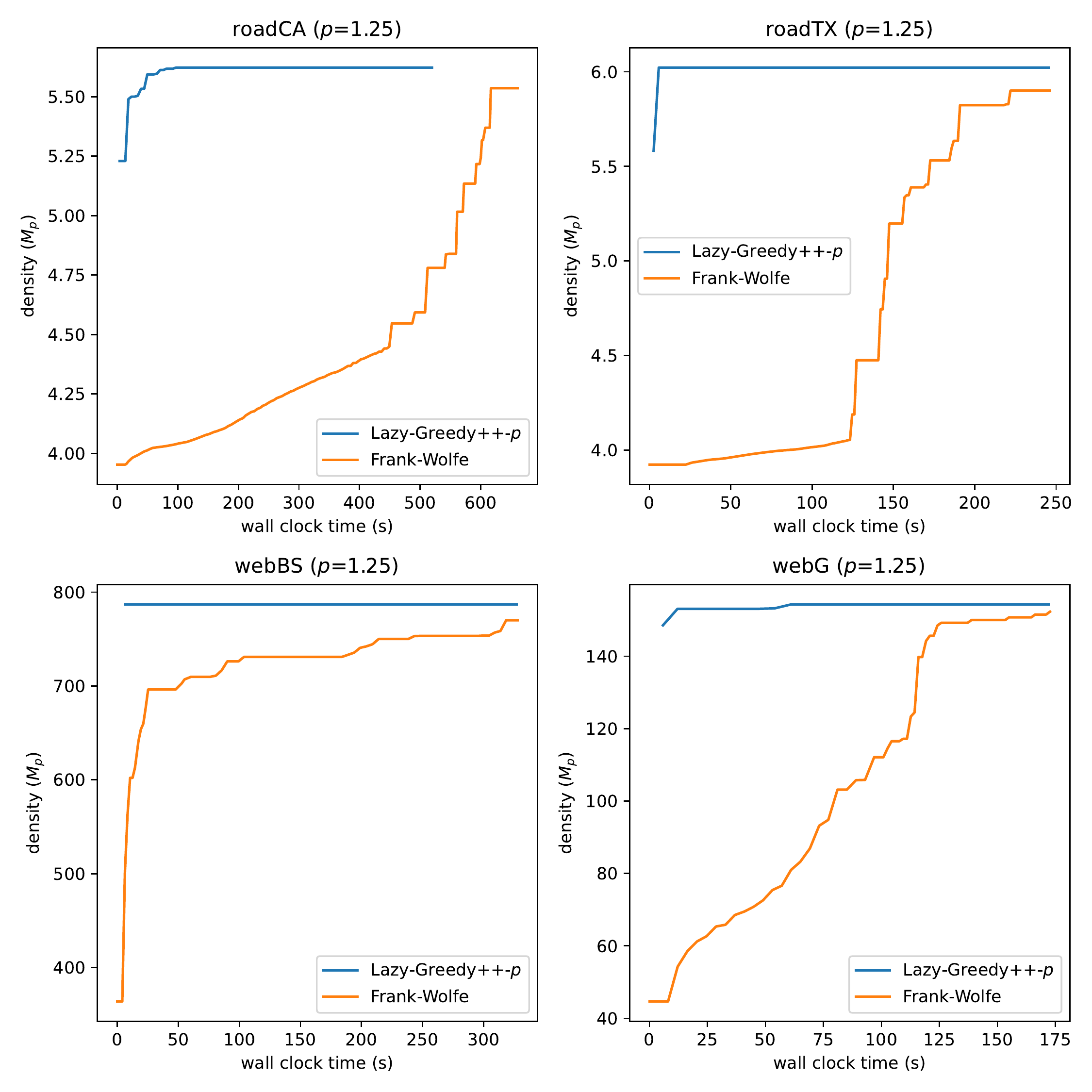}
\caption{Experiments comparing convergence rates of the
algorithms we consider in Section~\ref{sec:exp-near-opt}. Results for $p = 1.25$. }
\label{fig:exp-near-opt-1.25}
\end{figure*}

\begin{figure*}
\centering
\includegraphics[scale=0.45]{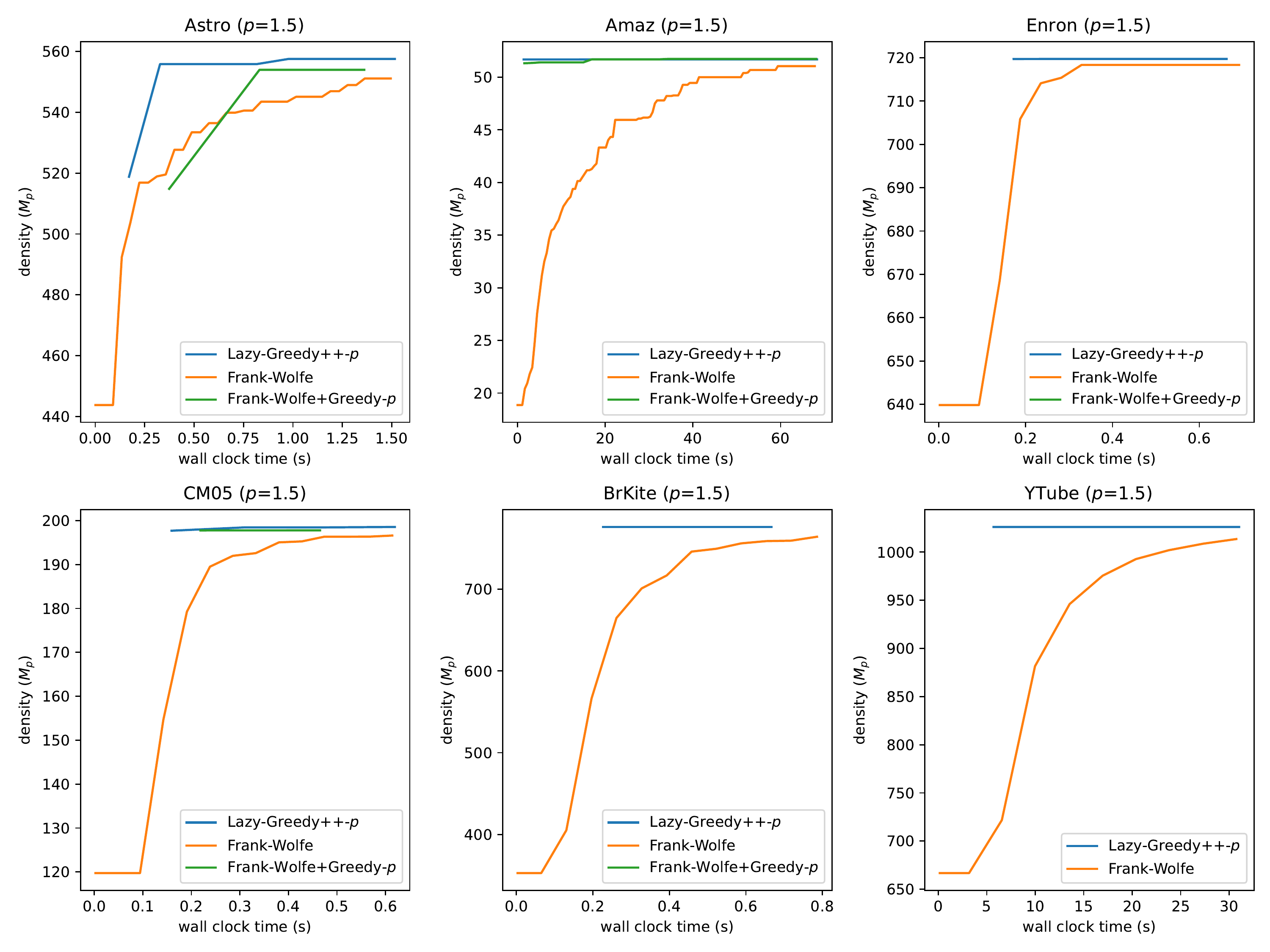}

\bigskip

\includegraphics[scale=0.45]{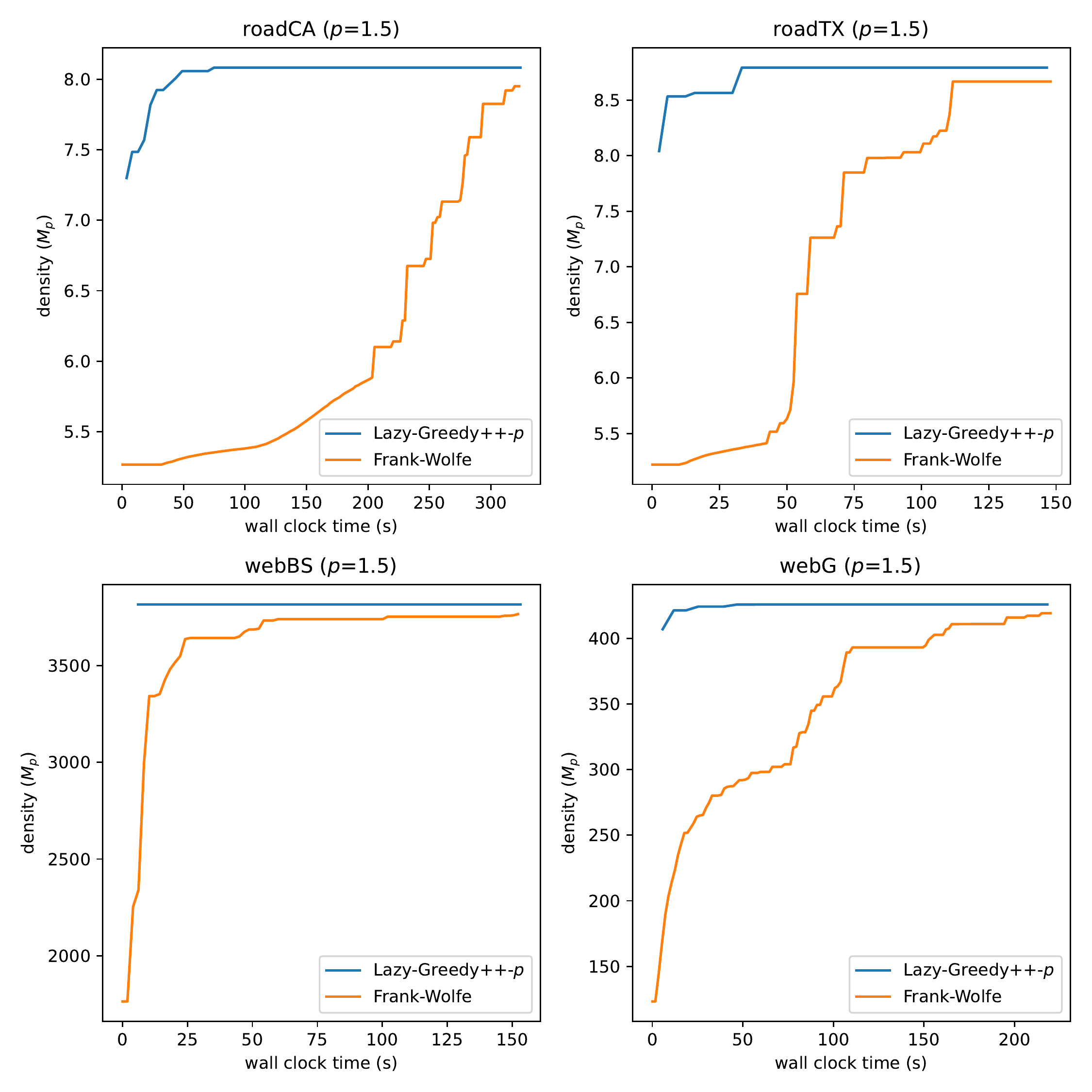}
\caption{Experiments comparing convergence rates of the
algorithms we consider in Section~\ref{sec:exp-near-opt}. Results for $p = 1.5$.}
\label{fig:exp-near-opt-1.5}
\end{figure*}

\begin{figure*}
\centering
\includegraphics[scale=0.45]{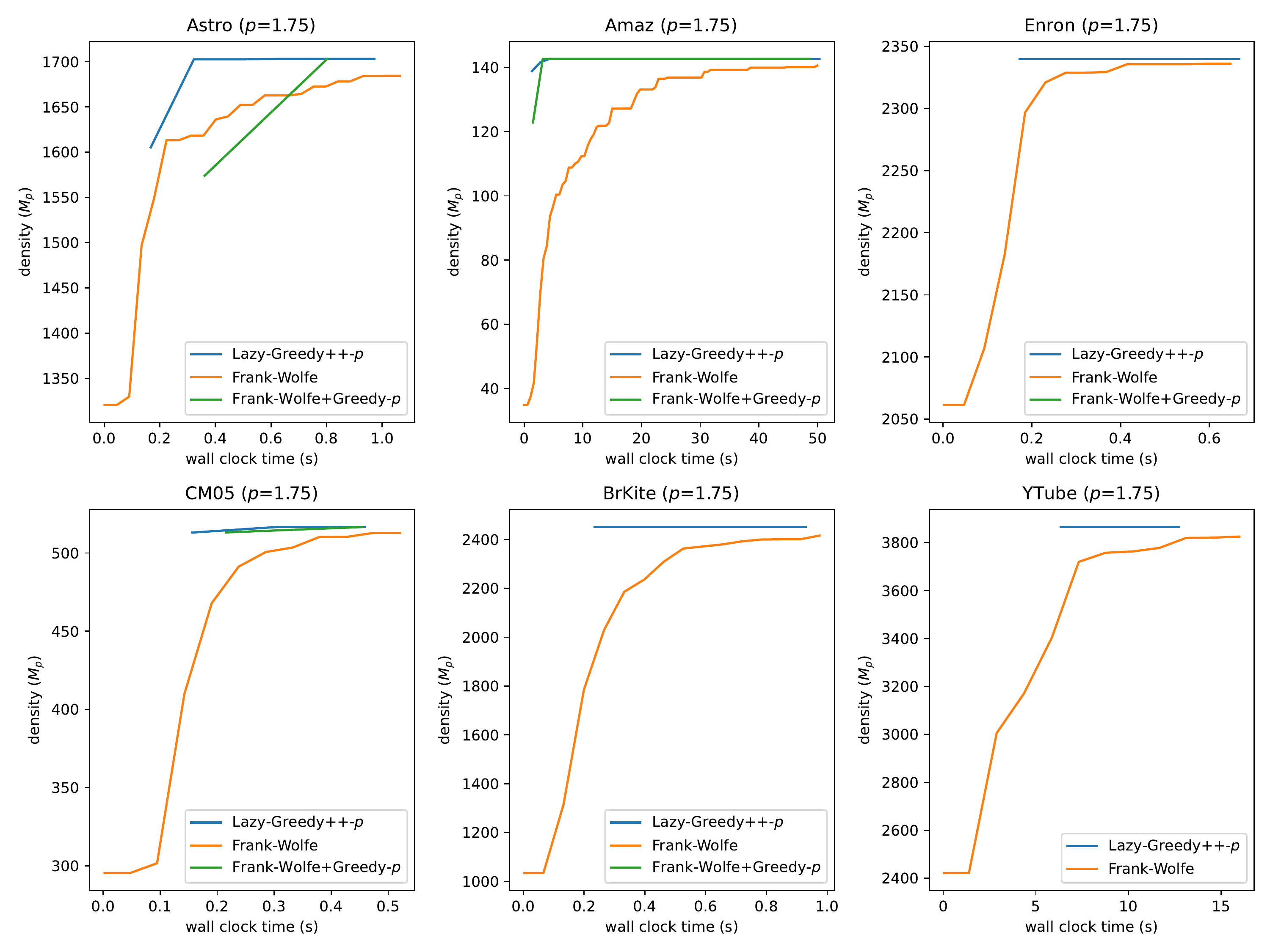}

\bigskip

\includegraphics[scale=0.45]{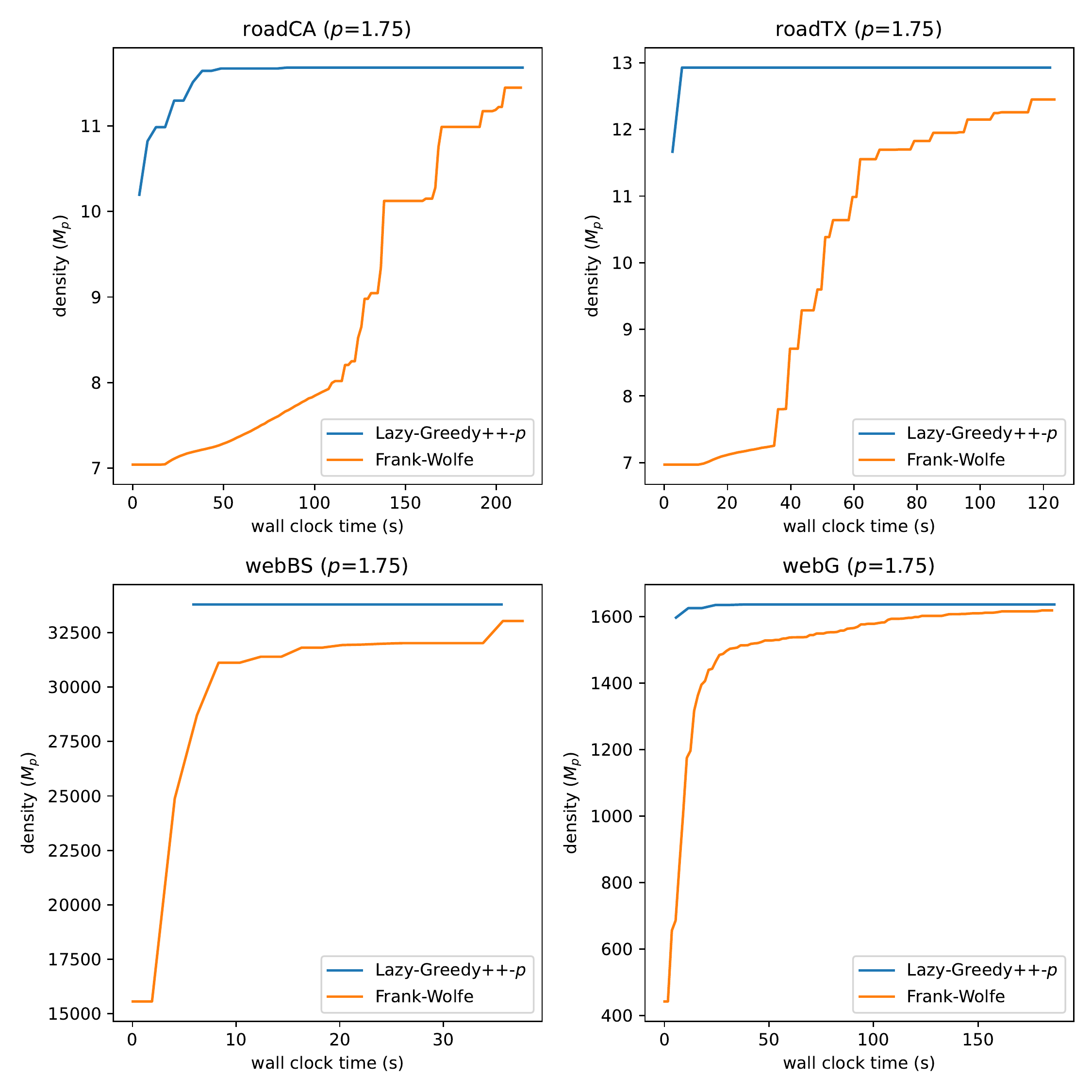}
\caption{Experiments comparing convergence rates of the
algorithms we consider in Section~\ref{sec:exp-near-opt}. Results for $p = 1.75$.}
\label{fig:exp-near-opt-1.75}
\end{figure*}

\begin{figure*}
\centering
\includegraphics[scale=0.45]{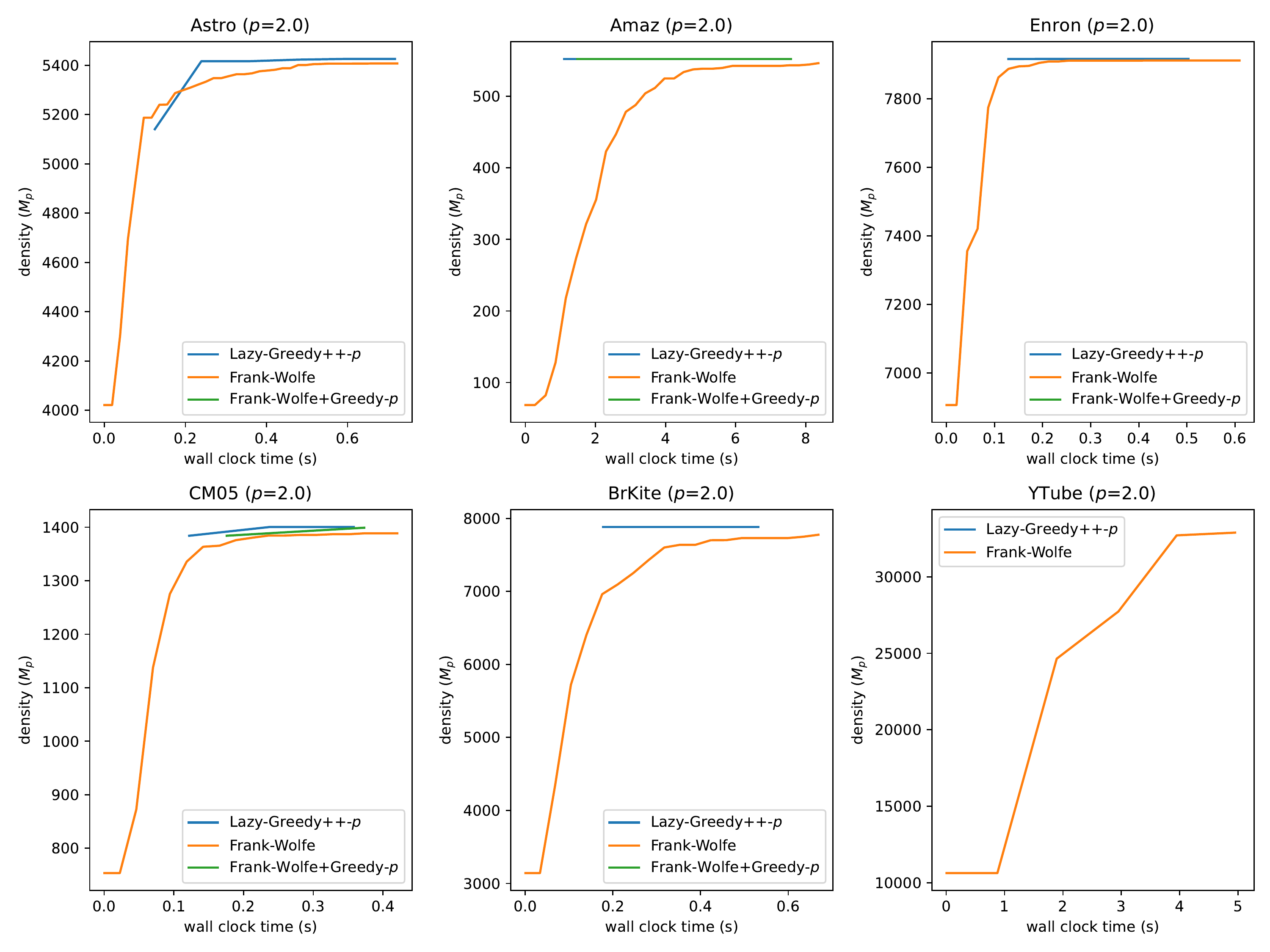}

\bigskip

\includegraphics[scale=0.45]{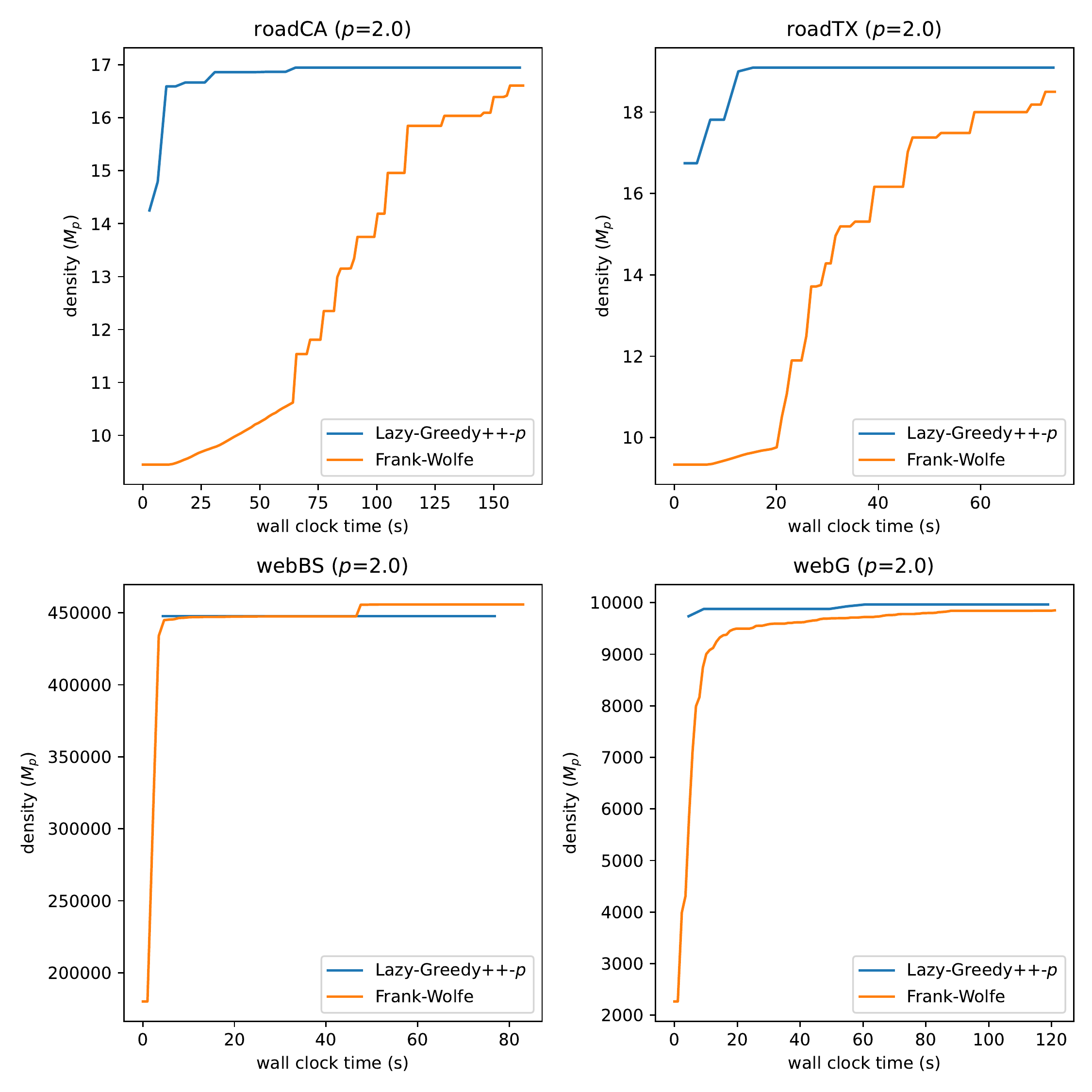}
\caption{Experiments comparing convergence rates of the
algorithms we consider in Section~\ref{sec:exp-near-opt}. Results for $p = 2$.}
\label{fig:exp-near-opt-2}
\end{figure*}

\begin{table*}
      \resizebox{\textwidth}{!}{
      \begin{tabular}{c|c|l|llllllllll}
        \hline
        & metric & algorithm  & Astro & CM05 & BrKite & Enron & roadCA & roadTX & webG & webBS & Amaz & YTube\\\hline
        \multirow{6}{*}{$p=-1.0$} &        \multirow{3}{*}{time (s)} & $\abbsgreedy$ & \textbf{0.022} & \textbf{0.032} & \textbf{0.043} & \textbf{0.028} & \textbf{1.146} & \textbf{0.793} & \textbf{1.297} & \textbf{0.836} & \textbf{0.356} & \textbf{1.151} \\
        &         & 1-mean DSG & 2.179 & 2.984 & 4.4 & 2.834 & 164.9 & 111.8 & 117.9 & 76.95 & 37.81 & 117.9 \\
        &         & $\abbsgpp$ & 2.179 & 2.984 & 4.4 & 2.834 & 164.9 & 111.8 & 117.9 & 76.95 & 37.81 & 117.9 \\
        \cline{2-13}
        &         \multirow{3}{*}{\begin{tabular}{c}density \\ ($M_p$) \end{tabular}} & $\abbsgreedy$ & 56.91 & \textbf{29.0} & \textbf{73.08} & \textbf{63.21} & 3.27 & 3.257 & 52.87 & \textbf{204.0} & \textbf{7.398} & \textbf{75.56} \\
        &         & 1-mean DSG & 55.52 & 26.64 & 70.28 & 61.19 & 3.583 & \textbf{3.833} & 53.99 & \textbf{204.0} & 5.45 & 72.8 \\
        &         & $\abbsgpp$ & \textbf{58.92} & \textbf{29.0} & \textbf{73.08} & \textbf{63.21} & \textbf{3.75} & \textbf{3.833} & \textbf{54.17} & \textbf{204.0} & \textbf{7.398} & \textbf{75.56} \\
        \hline
        \multirow{6}{*}{$p=-0.5$} &        \multirow{3}{*}{time (s)} & $\abbsgreedy$ & \textbf{0.049} & \textbf{0.06} & \textbf{0.075} & \textbf{0.054} & \textbf{1.8} & \textbf{1.101} & \textbf{1.781} & \textbf{1.558} & \textbf{0.507} & \textbf{1.615} \\
        &         & 1-mean DSG & 4.649 & 5.505 & 7.566 & 5.26 & 213.7 & 145.4 & 176.7 & 160.7 & 51.3 & 163.9 \\
        &         & $\abbsgpp$ & 4.649 & 5.505 & 7.566 & 5.26 & 213.7 & 145.4 & 176.7 & 160.7 & 51.3 & 163.9 \\
        \cline{2-13}
        &         \multirow{3}{*}{\begin{tabular}{c}density \\ ($M_p$) \end{tabular}} & $\abbsgreedy$ & 57.08 & 29.0 & \textbf{74.66} & \textbf{65.09} & 3.29 & 3.278 & 53.25 & \textbf{204.5} & \textbf{7.456} & \textbf{78.0} \\
        &         & 1-mean DSG & 57.28 & 27.65 & 72.71 & 63.7 & 3.672 & \textbf{3.899} & 54.55 & \textbf{204.5} & 5.692 & 75.89 \\
        &         & $\abbsgpp$ & \textbf{59.83} & \textbf{29.05} & \textbf{74.66} & \textbf{65.09} & \textbf{3.774} & \textbf{3.899} & \textbf{54.65} & \textbf{204.5} & \textbf{7.456} & \textbf{78.0} \\
        \hline
        \multirow{8}{*}{$p=0.25$} &        \multirow{4}{*}{time (s)} & $\greedy$ & 1.601 & 0.651 & 1.859 & 3.443 & 4.337 & 3.003 & 151.7 & 4049.2 & 3.08 & 294.0 \\
        &         & $\abbsgreedy$ & \textbf{0.037} & \textbf{0.038} & \textbf{0.047} & \textbf{0.038} & \textbf{0.861} & \textbf{0.586} & \textbf{1.104} & \textbf{1.22} & \textbf{0.272} & \textbf{0.891} \\
        &         & 1-mean DSG & 3.694 & 3.801 & 4.838 & 3.833 & 99.06 & 67.22 & 109.3 & 124.1 & 26.12 & 89.37 \\
        &         & $\abbsgpp$ & 3.694 & 3.801 & 4.838 & 3.833 & 99.06 & 67.22 & 109.3 & 124.1 & 26.12 & 89.37 \\
        \cline{2-13}
        &         \multirow{4}{*}{\begin{tabular}{c}density \\ ($M_p$) \end{tabular}} & $\greedy$ & 56.78 & 29.0 & \textbf{77.47} & \textbf{68.98} & 2.935 & 2.924 & 52.71 & \textbf{205.5} & 5.89 & \textbf{83.08} \\
        &         & $\abbsgreedy$ & 57.37 & 29.01 & 77.45 & 68.95 & 3.322 & 3.317 & 53.81 & \textbf{205.5} & \textbf{7.546} & 83.01 \\
        &         & 1-mean DSG & 60.4 & 29.5 & 76.74 & 68.46 & 3.804 & \textbf{4.015} & 55.35 & \textbf{205.5} & 6.637 & 82.1 \\
        &         & $\abbsgpp$ & \textbf{61.33} & \textbf{30.14} & 77.45 & 68.95 & \textbf{3.842} & \textbf{4.015} & \textbf{55.36} & \textbf{205.5} & \textbf{7.546} & 83.02 \\
        \hline
        \multirow{8}{*}{$p=0.5$} &        \multirow{4}{*}{time (s)} & $\greedy$ & 1.553 & 0.644 & 1.868 & 3.403 & 4.364 & 3.005 & 147.0 & 4483.8 & 2.792 & 292.3 \\
        &         & $\abbsgreedy$ & \textbf{0.036} & \textbf{0.038} & \textbf{0.05} & \textbf{0.039} & \textbf{0.857} & \textbf{0.596} & \textbf{1.165} & \textbf{1.31} & \textbf{0.26} & \textbf{0.887} \\
        &         & 1-mean DSG & 3.772 & 3.792 & 4.896 & 3.798 & 99.07 & 67.15 & 107.8 & 124.0 & 25.56 & 90.44 \\
        &         & $\abbsgpp$ & 3.772 & 3.792 & 4.896 & 3.798 & 99.07 & 67.15 & 107.8 & 124.0 & 25.56 & 90.44 \\
        \cline{2-13}
        &         \multirow{4}{*}{\begin{tabular}{c}density \\ ($M_p$) \end{tabular}} & $\greedy$ & 56.86 & 29.0 & \textbf{78.56} & \textbf{70.62} & 3.296 & 3.31 & 53.3 & \textbf{205.9} & 6.723 & \textbf{85.28} \\
        &         & $\abbsgreedy$ & 57.48 & 29.44 & 78.54 & 70.61 & 3.333 & 3.332 & 54.0 & \textbf{205.9} & 7.576 & 85.23 \\
        &         & 1-mean DSG & 61.59 & 30.22 & 78.16 & 70.36 & 3.848 & \textbf{4.059} & \textbf{55.6} & \textbf{205.9} & 7.297 & 84.74 \\
        &         & $\abbsgpp$ & \textbf{61.87} & \textbf{30.62} & 78.54 & 70.61 & \textbf{3.867} & \textbf{4.059} & \textbf{55.6} & \textbf{205.9} & \textbf{7.869} & 85.24 \\
        \hline
        \multirow{8}{*}{$p=0.75$} &        \multirow{4}{*}{time (s)} & $\greedy$ & 1.609 & 0.645 & 1.863 & 3.383 & 4.23 & 2.905 & 136.4 & 3566.4 & 2.751 & 280.2 \\
        &         & $\abbsgreedy$ & \textbf{0.036} & \textbf{0.038} & \textbf{0.047} & \textbf{0.038} & \textbf{0.847} & \textbf{0.582} & \textbf{1.087} & \textbf{1.206} & \textbf{0.262} & \textbf{0.882} \\
        &         & 1-mean DSG & 3.586 & 3.714 & 4.807 & 3.83 & 97.5 & 66.34 & 107.7 & 123.4 & 26.12 & 89.1 \\
        &         & $\abbsgpp$ & 3.586 & 3.714 & 4.807 & 3.83 & 97.5 & 66.34 & 107.7 & 123.4 & 26.12 & 89.1 \\
        \cline{2-13}
        &         \multirow{4}{*}{\begin{tabular}{c}density \\ ($M_p$) \end{tabular}} & $\greedy$ & 57.6 & 29.48 & \textbf{79.75} & \textbf{72.52} & 3.305 & 3.325 & 53.95 & \textbf{206.3} & 6.997 & \textbf{87.94} \\
        &         & $\abbsgreedy$ & 57.99 & 30.01 & \textbf{79.75} & 72.51 & 3.344 & 3.348 & 54.18 & \textbf{206.3} & 7.606 & 87.93 \\
        &         & 1-mean DSG & 62.86 & 31.0 & 79.62 & 72.43 & 3.892 & \textbf{4.105} & \textbf{55.84} & \textbf{206.3} & 8.26 & 87.75 \\
        &         & $\abbsgpp$ & \textbf{62.87} & \textbf{31.21} & \textbf{79.75} & \textbf{72.52} & \textbf{3.897} & \textbf{4.105} & \textbf{55.84} & \textbf{206.3} & \textbf{8.42} & 87.93 \\
        \hline
      \end{tabular}
      }
\caption{Results for algorithms for $p < 1$. 
We write $\sgreedy$ as $\abbsgreedy$ and $\sgpp$ as $\abbsgpp$ for spacing issues. 
See Section~\ref{sec:exp-lt1} for more details.}
\label{fig:exp-lt1-full}
\end{table*}

%% file: main.bbl
\newcommand{\etalchar}[1]{$^{#1}$}
\begin{thebibliography}{MGPV20}

\bibitem[AC09]{ac-09}
Reid Andersen and Kumar Chellapilla.
\newblock Finding dense subgraphs with size bounds.
\newblock In {\em International Workshop on Algorithms and Models for the
  Web-Graph}, pages 25--37. Springer, 2009.

\bibitem[AITT00]{aitt-00}
Yuichi Asahiro, Kazuo Iwama, Hisao Tamaki, and Takeshi Tokuyama.
\newblock Greedily finding a dense subgraph.
\newblock {\em Journal of Algorithms}, 34(2):203--221, 2000.

\bibitem[AKS{\etalchar{+}}14]{angel-14}
Albert Angel, Nick Koudas, Nikos Sarkas, Divesh Srivastava, Michael Svendsen,
  and Srikanta Tirthapura.
\newblock Dense subgraph maintenance under streaming edge weight updates for
  real-time story identification.
\newblock {\em The VLDB journal}, 23:175--199, 2014.

\bibitem[BGM14]{bgm-14}
Bahman Bahmani, Ashish Goel, and Kamesh Munagala.
\newblock Efficient primal-dual graph algorithms for {MapReduce}.
\newblock In {\em International Workshop on Algorithms and Models for the
  Web-Graph}, pages 59--78. Springer, 2014.

\bibitem[BGP{\etalchar{+}}20]{boob-20}
Digvijay Boob, Yu~Gao, Richard Peng, Saurabh Sawlani, Charalampos Tsourakakis,
  Di~Wang, and Junxing Wang.
\newblock Flowless: Extracting densest subgraphs without flow computations.
\newblock In {\em Proceedings of The Web Conference 2020}, pages 573--583,
  2020.

\bibitem[BH03]{bader-03}
Gary~D Bader and Christopher~WV Hogue.
\newblock An automated method for finding molecular complexes in large protein
  interaction networks.
\newblock {\em BMC bioinformatics}, 4(1):1--27, 2003.

\bibitem[Bic12]{b-12}
Allan Bickle.
\newblock Structural results on maximal $k$-degenerate graphs.
\newblock {\em Discussiones Mathematicae Graph Theory}, 32(4):659--676, 2012.

\bibitem[BSW19]{bsw-19}
Digvijay Boob, Saurabh Sawlani, and Di~Wang.
\newblock Faster width-dependent algorithm for mixed packing and covering
  {LPs}.
\newblock {\em Advances in Neural Information Processing Systems 32 (NIPS
  2019)}, 2019.

\bibitem[Cha00]{c-00}
Moses Charikar.
\newblock Greedy approximation algorithms for finding dense components in a
  graph.
\newblock In {\em International Workshop on Approximation Algorithms for
  Combinatorial Optimization}, pages 84--95. Springer, 2000.

\bibitem[CQT22]{cqt-22}
Chandra Chekuri, Kent Quanrud, and Manuel~R Torres.
\newblock Densest subgraph: Supermodularity, iterative peeling, and flow.
\newblock In {\em Proceedings of the 2022 Annual ACM-SIAM Symposium on Discrete
  Algorithms (SODA)}, pages 1531--1555. SIAM, 2022.

\bibitem[DCS17]{dcs-17}
Maximilien Danisch, T.{-}H.~Hubert Chan, and Mauro Sozio.
\newblock Large scale density-friendly graph decomposition via convex
  programming.
\newblock In {\em Proceedings of the 26th International Conference on World
  Wide Web}, pages 233--242, 2017.

\bibitem[DGP07]{dourisboure-07}
Yon Dourisboure, Filippo Geraci, and Marco Pellegrini.
\newblock Extraction and classification of dense communities in the web.
\newblock In {\em Proceedings of the 16th international conference on World
  Wide Web}, pages 461--470, 2007.

\bibitem[DH11]{suitesparse}
Timothy~A Davis and Yifan Hu.
\newblock The {University of Florida} sparse matrix collection.
\newblock {\em ACM Transactions on Mathematical Software (TOMS)}, 38(1):1--25,
  2011.

\bibitem[DJD{\etalchar{+}}09]{du-09}
Xiaoxi Du, Ruoming Jin, Liang Ding, Victor~E Lee, and John~H Thornton~Jr.
\newblock Migration motif: a spatial-temporal pattern mining approach for
  financial markets.
\newblock In {\em Proceedings of the 15th ACM SIGKDD international conference
  on knowledge discovery and data mining}, pages 1135--1144, 2009.

\bibitem[Far08]{f-08}
Andr{\'a}s Farag{\'o}.
\newblock A general tractable density concept for graphs.
\newblock {\em Mathematics in Computer Science}, 1(4):689--699, 2008.

\bibitem[FCT14]{ft-14}
Martin Farach-Colton and Meng-Tsung Tsai.
\newblock Computing the degeneracy of large graphs.
\newblock In {\em LATIN 2014: Theoretical Informatics: 11th Latin American
  Symposium, Montevideo, Uruguay, March 31--April 4, 2014. Proceedings 11},
  pages 250--260. Springer, 2014.

\bibitem[FNBB06]{fratkin-06}
Eugene Fratkin, Brian~T Naughton, Douglas~L Brutlag, and Serafim Batzoglou.
\newblock {MotifCut:} regulatory motifs finding with maximum density subgraphs.
\newblock {\em Bioinformatics}, 22(14):e150--e157, 2006.

\bibitem[FRM19]{f-19}
Andr{\'a}s Farag{\'o} and Zohre R~Mojaveri.
\newblock In search of the densest subgraph.
\newblock {\em Algorithms}, 12(8):157, 2019.

\bibitem[GGT89]{ggt-89}
Giorgio Gallo, Michael~D Grigoriadis, and Robert~E Tarjan.
\newblock A fast parametric maximum flow algorithm and applications.
\newblock {\em SIAM Journal on Computing}, 18(1):30--55, 1989.

\bibitem[GKT05]{gibson-05}
David Gibson, Ravi Kumar, and Andrew Tomkins.
\newblock Discovering large dense subgraphs in massive graphs.
\newblock In {\em Proceedings of the 31st international conference on Very
  large data bases}, pages 721--732, 2005.

\bibitem[Gol84]{g-84}
Andrew~V Goldberg.
\newblock {\em Finding a maximum density subgraph}.
\newblock University of California Berkeley, 1984.

\bibitem[GT15]{gt-15}
Aristides Gionis and Charalampos~E Tsourakakis.
\newblock Dense subgraph discovery: {KDD} 2015 tutorial.
\newblock In {\em Proceedings of the 21th ACM SIGKDD International Conference
  on Knowledge Discovery and Data Mining}, pages 2313--2314, 2015.

\bibitem[HQC22]{hqc-22}
Elfarouk Harb, Kent Quanrud, and Chandra Chekuri.
\newblock Faster and scalable algorithms for densest subgraph and
  decomposition.
\newblock In {\em Advances in Neural Information Processing Systems}, 2022.

\bibitem[HWC17]{hwc-17}
Shuguang Hu, Xiaowei Wu, and TH~Hubert Chan.
\newblock Maintaining densest subsets efficiently in evolving hypergraphs.
\newblock In {\em Proceedings of the 2017 ACM on Conference on Information and
  Knowledge Management}, pages 929--938, 2017.

\bibitem[HYH{\etalchar{+}}05]{hu-05}
Haiyan Hu, Xifeng Yan, Yu~Huang, Jiawei Han, and Xianghong~Jasmine Zhou.
\newblock Mining coherent dense subgraphs across massive biological networks
  for functional discovery.
\newblock {\em Bioinformatics}, 21(suppl\_1):i213--i221, 2005.

\bibitem[JZT{\etalchar{+}}22]{ji-22}
Yingsheng Ji, Zheng Zhang, Xinlei Tang, Jiachen Shen, Xi~Zhang, and Guangwen
  Yang.
\newblock Detecting cash-out users via dense subgraphs.
\newblock In {\em Proceedings of the 28th ACM SIGKDD Conference on Knowledge
  Discovery and Data Mining}, pages 687--697, 2022.

\bibitem[KM18]{km-18}
Yasushi Kawase and Atsushi Miyauchi.
\newblock The densest subgraph problem with a convex/concave size function.
\newblock {\em Algorithmica}, 80(12):3461--3480, 2018.

\bibitem[KNT06]{kumar-06}
Ravi Kumar, Jasmine Novak, and Andrew Tomkins.
\newblock Structure and evolution of online social networks.
\newblock In {\em Proceedings of the 12th ACM SIGKDD international conference
  on Knowledge discovery and data mining}, pages 611--617, 2006.

\bibitem[Law76]{l-76}
E.L. Lawler.
\newblock {\em Combinatorial Optimization: Networks and Matroids}.
\newblock Holt, Rinehart and Winston, 1976.

\bibitem[LBG20]{lanciano-20}
Tommaso Lanciano, Francesco Bonchi, and Aristides Gionis.
\newblock Explainable classification of brain networks via contrast subgraphs.
\newblock In {\em Proceedings of the 26th ACM SIGKDD International Conference
  on Knowledge Discovery \& Data Mining}, pages 3308--3318, 2020.

\bibitem[LK14]{snap}
Jure Leskovec and Andrej Krevl.
\newblock {SNAP Datasets}: {Stanford} large network dataset collection.
\newblock \url{http://snap.stanford.edu/data}, June 2014.

\bibitem[LMFB23]{lmfb-23}
Tommaso Lanciano, Atsushi Miyauchi, Adriano Fazzone, and Francesco Bonchi.
\newblock A survey on the densest subgraph problem and its variants.
\newblock {\em arXiv preprint arXiv:2303.14467}, 2023.

\bibitem[Lov83]{l-83}
L{\'a}szl{\'o} Lov{\'a}sz.
\newblock Submodular functions and convexity.
\newblock In {\em Mathematical programming the state of the art}, pages
  235--257. Springer, 1983.

\bibitem[LRJA10]{lrja-10}
Victor~E Lee, Ning Ruan, Ruoming Jin, and Charu Aggarwal.
\newblock A survey of algorithms for dense subgraph discovery.
\newblock In {\em Managing and Mining Graph Data}, pages 303--336. Springer,
  2010.

\bibitem[LZW{\etalchar{+}}08]{li-08}
Zhenping Li, Shihua Zhang, Rui-Sheng Wang, Xiang-Sun Zhang, and Luonan Chen.
\newblock Quantitative function for community detection.
\newblock {\em Physical review E}, 77(3):036109, 2008.

\bibitem[MB83]{mb-83}
David~W Matula and Leland~L Beck.
\newblock Smallest-last ordering and clustering and graph coloring algorithms.
\newblock {\em Journal of the ACM (JACM)}, 30(3):417--427, 1983.

\bibitem[MGPV20]{mgpv-20}
Fragkiskos~D Malliaros, Christos Giatsidis, Apostolos~N Papadopoulos, and
  Michalis Vazirgiannis.
\newblock The core decomposition of networks: Theory, algorithms and
  applications.
\newblock {\em The VLDB Journal}, 29:61--92, 2020.

\bibitem[MK18]{mk-18}
Atsushi Miyauchi and Naonori Kakimura.
\newblock Finding a dense subgraph with sparse cut.
\newblock In {\em Proceedings of the 27th ACM International Conference on
  Information and Knowledge Management}, pages 547--556, 2018.

\bibitem[PQ82]{pq-82}
Jean-Claude Picard and Maurice Queyranne.
\newblock A network flow solution to some nonlinear 0-1 programming problems,
  with applications to graph theory.
\newblock {\em Networks}, 12(2):141--159, 1982.

\bibitem[Sch03]{s-03}
Alexander Schrijver.
\newblock {\em Combinatorial optimization: polyhedra and efficiency},
  volume~24.
\newblock Springer Science \& Business Media, 2003.

\bibitem[SM03]{spirin-03}
Victor Spirin and Leonid~A Mirny.
\newblock Protein complexes and functional modules in molecular networks.
\newblock {\em Proceedings of the national Academy of sciences},
  100(21):12123--12128, 2003.

\bibitem[SMT22]{smt-22}
Issey Sukeda, Atsushi Miyauchi, and Akiko Takeda.
\newblock A study on modularity density maximization: Column generation
  acceleration and computational complexity analysis.
\newblock {\em arXiv preprint arXiv:2206.10901}, 2022.

\bibitem[TC21]{tc-21}
Charalampos Tsourakakis and Tianyi Chen.
\newblock Dense subgraph discovery: Theory and application ({Tutoral at SDM
  2021}), 2021.
\newblock
  \url{https://tsourakakis.com/dense-subgraph-discovery-theory-and-applications-tutorial-sdm-2021/}.

\bibitem[Tso14]{t-14}
Charalampos~E Tsourakakis.
\newblock A novel approach to finding near-cliques: The triangle-densest
  subgraph problem.
\newblock {\em arXiv preprint arXiv:1405.1477}, 2014.

\bibitem[Tso15]{t-15}
Charalampos Tsourakakis.
\newblock The $k$-clique densest subgraph problem.
\newblock In {\em Proceedings of the 24th international conference on world
  wide web}, pages 1122--1132, 2015.

\bibitem[VBK21]{vbk-21}
Nate Veldt, Austin~R Benson, and Jon Kleinberg.
\newblock The generalized mean densest subgraph problem.
\newblock In {\em Proceedings of the 27th ACM SIGKDD Conference on Knowledge
  Discovery \& Data Mining}, pages 1604--1614, 2021.

\bibitem[ZZY{\etalchar{+}}17]{zhang-17}
Si~Zhang, Dawei Zhou, Mehmet~Yigit Yildirim, Scott Alcorn, Jingrui He, Hasan
  Davulcu, and Hanghang Tong.
\newblock Hidden: hierarchical dense subgraph detection with application to
  financial fraud detection.
\newblock In {\em Proceedings of the 2017 SIAM International Conference on Data
  Mining}, pages 570--578. SIAM, 2017.

\end{thebibliography}
